\documentclass[USletter,10pt]{article}
\usepackage[dvips]{geometry}
\usepackage{graphicx,cancel}
\usepackage{graphics}
\usepackage{hyperref}
\usepackage{color, soul}
\usepackage{amsmath, amsthm, slashed}
\usepackage{amssymb}
\usepackage{rotating}
\usepackage{mathrsfs}
\usepackage{epsfig}
\usepackage{verbatim}
\usepackage[affil-it]{authblk}
\usepackage{accents}
\newtheorem{definition}{Definition}[section]
\newtheorem{theorem}{Theorem}[section]
 
\newtheorem*{conjecture*}{Conjecture}
\newtheorem{corollary}{Corollary}[section]
\newtheorem*{theorem*}{Theorem}
\newtheorem*{corollary*}{Corollary}
\newtheorem{proposition}{Proposition}[subsection]
\newtheorem{lemma}{Lemma}[subsection]
\newtheorem{remark}{Remark}[section]

\newtheoremstyle{named}{}{}{\itshape}{}{\bfseries}{.}{.5em}{\thmnote{#3 }#1}
\theoremstyle{named}

 \newcommand{\bea}{\begin{eqnarray}}
\newcommand{\eea}{\end{eqnarray}}
\newcommand{\beaa}{\begin{eqnarray*}}
\newcommand{\eeaa}{\end{eqnarray*}}
\newcommand{\bsplit}{\begin{split}}
\newcommand{\nn}{\nonumber}
\newcommand{\les}{\lesssim}
\newcommand{\ges}{\gtrsim}
\newcommand{\lot}{\mbox{l.o.t.}}
\newcommand{\dual}{{\,^\star \mkern-2mu}}
\newcommand{\tr}{\mbox{tr}}
\newcommand{\nabb}{\nab\mkern-13mu /\,}

\newcommand{\nab}{\nabla}
\renewcommand{\div}{\mbox{div }}
\newcommand{\curl}{\mbox{curl }}
\newcommand{\divv}{\mbox{div}\mkern-19mu /\,\,\,\,}
\newcommand{\lapp}{\mbox{$\bigtriangleup  \mkern-13mu / \,$}}
\newcommand{\curll}{\mbox{curl}\mkern-19mu /\,\,\,\,}
\newcommand{\pr}{\partial}

\newcommand{\DD}{{\mathcal D}}
\newcommand{\DDs}{ \, \DD \hspace{-2.4pt}\dual    \mkern-16mu /}
\newcommand{\DDd}{ \, \DD \hspace{-2.4pt}    \mkern-8mu /}
\newcommand{\Db}{\dot{\D}}

\renewcommand{\a}{\alpha}
\renewcommand{\b}{\beta}
\newcommand{\de}{\delta}
\newcommand{\ep}{\epsilon}
\newcommand{\la}{\lambda}
\newcommand{\La}{\Lambda}

\newcommand{\Si}{\Sigma}
\newcommand{\om}{\omega}

\renewcommand{\th}{\theta}

\newcommand{\ka}{\kappa}
\newcommand{\ze}{\zeta}
\newcommand{\Up}{\Upsilon}

\renewcommand{\aa}{\protect\underline{\a}}
\newcommand{\bb}{\protect\underline{\b}}
\newcommand{\omb}{{\underline{\om}}}
\newcommand{\Lb}{{\underline{L}}}

\newcommand{\chib}{\underline{\chi}}
\newcommand{\xib}{\underline{\xi}}
\newcommand{\etab}{\underline{\eta}}
\newcommand{\kab}{\underline{\kappa}}
\newcommand{\chih}{\widehat{\chi}}
\newcommand{\chibh}{\widehat{\chib}}

\newcommand{\bF}{\,^{(F)} \hspace{-2.2pt}\b}
\newcommand{\bbF}{\,^{(F)} \hspace{-2.2pt}\bb}
\newcommand{\rhoF}{\,^{(F)} \hspace{-2.2pt}\rho}
\newcommand{\sigmaF}{\,^{(F)} \hspace{-2.2pt}\sigma}

\renewcommand{\AA}{{\mathcal A}}
\newcommand{\BB}{{\mathcal B}}
\newcommand{\EE}{{\mathcal E}}
\newcommand{\HH}{{\mathcal H}}
\newcommand{\II}{{\mathcal I}}
\newcommand{\LL}{{\mathcal L}}
\newcommand{\MM}{{\mathcal M}}
\newcommand{\PP}{{\mathcal P}}
\newcommand{\QQ}{{\mathcal Q}}
\newcommand{\C}{{\bf C}}
\newcommand{\D}{{\bf D}}
\newcommand{\F}{{\bf F}}
\newcommand{\M}{{\bf M}}
\newcommand{\R}{{\bf R}}
\newcommand{\W}{{\bf W}}
\newcommand{\Lbb}{{\bf L}}
\newcommand{\g}{{\bf g}}
\newcommand{\pf}{\frak{p}}
\newcommand{\qf}{\frak{q}}
\newcommand{\ff}{\frak{f}}
\newcommand{\ffb}{\protect\underline{\ff}}

\newcommand{\piT}{{\,^{(T)} \pi }}
\newcommand{\piR}{\,^{(R)}\pi}
\newcommand{\piX}{\,^{(X)}\pi}
\newcommand{\piY}{\,^{(Y)}\pi}
\newcommand{\piZ}{  \,^{(Z)}\pi}
\newcommand{\piN}{  \,^{(N)}\pi}

\newcommand{\fb}{\underline{f}}

\newcommand{\gS}{ g \mkern-8.5mu/\,}

\newcommand{\hot}{\widehat{\otimes}}
\renewcommand{\c}{\cdot}

\newcommand{\Mor}{\mbox{Mor}}
\newcommand{\Morr}{\mbox{Morr}}
\newcommand{\MMdot}{\dot{\MM}}
\newcommand{\Ed}{\dot{E}}
\newcommand{\ec}{\check{e}}

\newcommand{\Rbrev}{\breve{R}}
        \newcommand{\Tbrev}{\breve{T}}
\newcommand{\Omegab}{\underline{\Omega}}

\title{Boundedness and decay for the Teukolsky system of spin $\pm2$ on Reissner-Nordstr{\"o}m spacetime: the case $|Q| \ll M$}

\author{Elena Giorgi\footnote{Gravity Initiative, Princeton University, egiorgi@princeton.edu}}

\begin{document}

\maketitle

\begin{abstract}
We prove boundedness and decay statements for solutions to the spin $\pm2$ generalized Teukolsky system on a  Reissner--Nordstr{\"o}m background with very small charge. The first equation of the system is the generalization of the Teukolsky equation in Schwarzschild for the extreme component of the curvature $\a$. The second equation, coupled with the first one, is a new equation for a new gauge-invariant quantity involving the electromagnetic curvature components. The proof is based on the use of derived quantities, introduced in previous works on linear stability of Schwarzschild \cite{DHR}.  These derived quantities are shown to verify a generalized coupled Regge--Wheeler system.

These equations are the ones verified by the extreme null curvature and electromagnetic components of a gravitational and electromagnetic perturbation of the Reissner--Nordstr{\"o}m spacetime. Consequently, as in the Schwarzschild case, these bounds provide the first step in proving the full linear stability of the Reissner--Nordstr{\"o}m metric for small charge to coupled gravitational and electromagnetic perturbations. 
\end{abstract}

\tableofcontents

\section{Introduction}\label{intro}
The problem of stability of   the Kerr family  in the context of the Einstein vacuum equations occupies a center stage  in mathematical General Relativity, see for example the introductions of 
  \cite{DHR}   and   \cite{stabilitySchwarzschild}  for  the formulation of the problem  and discussions of the main difficulties.

An essential  step  in the program of settling the stability of Kerr conjecture  is to understand the behavior of solutions to the   so-called \textit{Teukolsky  equations}. These are   wave  equations  verified  by  the extreme null   components of the curvature tensor  which decouple,   to second order, from all other curvature components.

 In linear  theory,   the Teukolsky equation, combined with cleverly chosen gauge conditions,  allows one to prove  the weakest version of stability, i.e the lack of exponentially growing modes for all curvature components. Extensive literature by the physics community covers these results (see \cite{Teukolsky}, \cite{Bardeen}, \cite{Chandra}, \cite{Whiting}).     This weak version of stability is however not sufficient to prove boundedness and decay of the solution to the Teukolsky equation;  one  needs instead    to derive    sufficiently strong decay estimates to hope to apply them in  the nonlinear  framework. 
 
  The first important  breakthrough  in this direction   was made by \cite{DHR}   in the context   of the linear stability of the Schwarzschild metric.  In their work, Dafermos, Holzegel and Rodnianski    derive the first   quantitative  decay estimates  for the Teukolsky equations in Schwarzschild and use them  to prove the first quantitative   stability result of the  full linearized  gravitational system around   a fixed Schwarzschild solution.   Different results and proofs of the linear stability of the Schwarzschild spacetime have followed, using the original   Regge--Wheeler  approach  of metric perturbations (see \cite{mu-tao}), and using wave gauge (see \cite{pei-ken}, \cite{pei-ken-2}) and generalized wave gauge (see \cite{Johnson1}, \cite{Johnson2}).    The nonlinear stability of the Schwarzschild  metric, for restricted polarized  perturbations, was recently proved in  \cite{stabilitySchwarzschild}.  
  
  The   results   on boundedness and  decay for  the Teukolsky equations were extended  to  slowly rotating    Kerr in  \cite{TeukolskyDHR} and in \cite{ma2}.    The  analysis of the full linearized equations  near slowly rotating Kerr recently appeared in \cite{Kerr-lin1} and \cite{Kerr-lin2}.  In the positive cosmological setting, the stability of Kerr-de Sitter with small angular momentum was obtained in \cite{Hintz-Vasy}. 

The approach of  \cite{DHR}  and  \cite{TeukolskyDHR}  to derive   boundedness and quantitative decay for the Teukolsky equations relies on the following ingredients:
\begin{enumerate}
\item       A    map which takes a solution to the Teukolsky equation to  a solution of   a wave equation which is simpler   to analyze. In the case of Schwarzschild,  this equation is known as the so called \textit{Regge--Wheeler  equation}.   The first such transformation was discovered by  Chandrasekhar (see \cite{Chandra})  in the context of mode decompositions. 
The physical version of this transformation first appears in \cite{DHR}.
\item  A vector field-type method    to  get   quantitative decay for the    new wave equation.
\item A method by which we  can derive estimates  for  solutions to the Teukolsky equation from 
those of solutions to the transformed  Regge--Wheeler  equation.
\end{enumerate}

Analogous problems appear in the mathematical study of stability of charged black holes, which are solutions of the coupled Einstein--Maxwell system in general relativity. The problem of stability of charged black holes has as a final goal the proof of non-linear stability of the Kerr-Newman family $(\MM, g_{M, Q, a})$ as solutions to the Einstein--Maxwell equation
\begin{equation} \label{Einsteineq}
Ric(g)_{\mu\nu}=T(F)_{\mu\nu}:=2 F_{\mu \lambda} {F_\nu}^{\lambda} - \frac 1 2 g_{\mu\nu} F^{\alpha\beta} F_{\alpha\beta}
\end{equation}
where $F$ is a $2$-form satisfying Maxwell's equations 
\begin{equation} \label{Max}
D_{[\alpha} F_{\beta\gamma]}=0, \qquad D^\alpha F_{\alpha\beta}=0.
\end{equation}
The Einstein--Maxwell system describes coupled gravitational and electromagnetic fields. The presence of a right hand side in the Einstein equation \eqref{Einsteineq} and the Maxwell's equations add new difficulties to the analysis of the problem, presenting coupling between the gravitational and the electromagnetic perturbations of solutions.
 In the positive cosmological setting, the stability of Kerr-Newman-de Sitter with small angular momentum was proved in \cite{Hintz-M}.

An intermediate step towards the proof of non-linear stability of charged black holes is the linear stability of the simplest non-trivial solution of the Einstein--Maxwell equations, the Reissner--Nordstr{\"o}m spacetime.

  The Reissner--Nordstr{\"o}m family
  of spacetimes $(\mathcal{M},\g_{M, Q})$ describe an electrically charged, nonrotating, spherically symmetric black hole and can be   expressed  in local coordinates as
  \begin{equation}
  \label{RNintro}
 \g_{M, Q}=-\left(1-\frac{2M}{r}+\frac{Q^2}{r^2}\right)dt^2 +\left(1-\frac{2M}{r}+\frac{Q^2}{r^2}\right)^{-1}dr^2 +r^2(d\theta^2+\sin^2\theta d\phi^2),
  \end{equation}
where $M$ and $Q$ are two real parameters, which can be interpreted as the mass and the charge of the black hole, for $|Q|<M$. This solution corresponds to the Kerr--Newman solution for $a=0$ (like the Schwarzschild solution corresponds to Kerr for $a=0$). Observe that Schwarzschild also corresponds to the Reissner--Nordstr{\"o}m solution for $Q=0$.

To extend the results on the linear of Schwarzschild to Reissner--Nordstr{\"o}m, a key step is to understand the analogous Teukolsky equation. The gauge-independent quantities involved, and the structure of the equations that they verify in physical space, were not fully understood up to this point. Indeed, we  need  to first find equations in physical space similar to the Teukolsky equations, for which one can prove quantitative decay.

The aim of this paper is to generalize the method of \cite{DHR} to the case of a charged black hole with a perturbed electromagnetic and gravitational field and very small charge. More precisely, we derive the relevant Teukolsky system for linear perturbations of the Einstein--Maxwell system, which did not seem to exist, at least in the form used in this paper, previously in the literature. 

We rely on the following ingredients:
\begin{enumerate}
\item Computations in physical space which show the Teukolsky-type equations verified by the extreme null curvature components in Reissner--Nordstr{\"o}m spacetime. We obtain a system of two coupled Teukolsky-type equations.
\item       A    map which takes solutions to the Teukolsky-type equation to  solutions to a Regge--Wheeler-type  equation.  We obtain a system of two coupled Regge--Wheeler-type  equations. As in the rotating case, in the charged case the Chandrasekhar transformation leaves lower order terms in the main equations. In addition, in the charged case, we also see coupling terms appear. 
\item  A vector field method    to  get   quantitative decay for the system.  The analysis is highly affected by the fact that we are dealing with a system, as opposed to a single equation. 
\item A method by which we  can derive estimates  for  solutions to the Teukolsky-type system from 
those of solutions to the transformed  Regge--Wheeler-type system.
\end{enumerate}

We give a rough statement of the main result in Section \ref{first-result}.

The spin $\pm2$ system arises in the study of linearized Einstein--Maxwell equations on a fixed Reissner--Nordstr\"om solution and it is composed by the equations satisfied by certain gauge-invariant combinations of the null-decomposed linearized Weyl curvature tensor and linearized Maxwell $2$-form. Establishing decay estimates for solutions to this system can thus be seen as a first step in establishing the linear stability of the Reissner--Nordstr\"om family. 
Indeed, the results of this paper, combined with \cite{Giorgi5}, are used as a key step to prove the decay of the linearized metric of the perturbation for $|Q| \ll M$ in \cite{Giorgi6}.

\subsection{The Teukolsky-type equations}
Just as in the vacuum case, the original approach to linear stability of the Reissner--Nordstr{\"o}m spacetime are the metric pertubations, leading to a generalization of the Regge--Wheeler and Zerilli equations. See for example \cite{Chandra-RN}, where the two pairs of one-dimensional wave equations which govern the odd and the even-parity perturbations of the Reissner--Nordstr{\"o}m black hole are derived directly from a treatment of its metric perturbations. See also references in \cite{Chandra}.

Another widely studied approach in the physics community is the Newman--Penrose formalism. See for example \cite{Chandra-RN2} and \cite{Bicak}.

These results rely on the derivation of the equations in separated forms, with a decomposition in modes. This is enough to prove that there are no exponentially growing modes, but does not prove boundedness or decay of the solution. To prove boundedness and decay, a physical space analysis has to be done.

In this paper, we use the formalism of null frames and derive the Teukolsky equations for Reissner--Nordstr{\"o}m spacetime in physical space. Suppose that $(\MM, \g, \F)$ is a solution to the Einstein--Maxwell equations such that the manifold $\MM$ can be foliated by $2$-spheres. Define then the symmetric traceless $2$-covariant $S$-tensors $\a$ and $\ff$ defined relative to a null frame $\{ e_3, e_4, e_A \}_{A=1,2}$ as 
\beaa
\a_{AB}&=& \W(e_4, e_A, e_4, e_B), \qquad \ff_{AB}= \DDs_2 \bF_{AB} + \rhoF \chih_{AB}
\eeaa
where $\W$ is the Weyl curvature, $\DDs_2\bF_{AB}=- \D_{(A}\bF_{B)}+ \frac 1 2 \g_{AB} \divv \bF$ with $\D$ the Levi-Civita connection of $\g$, $\bF=\F(e_A, e_4)$, $\rhoF=\frac 1 2 \F(e_3, e_4)$ and $\chih$ is the traceless part of the $S$-tensor $\chi_{AB}=\g(\D_A e_4, e_B)$.

One may wonder what is the heuristic reason to define such a $\ff$ as above. 
The particular combination of electromagnetic and Ricci coefficients which constitutes the definition of $\ff$ appears in the right hand side of the Bianchi identity for $\a$\footnote{See equation \eqref{nabb-3-a}.}, and represents the additional part with respect to the vacuum Einstein equation in Schwarzschild. Since from the linear stability of Schwarzschild \cite{DHR} we know that the equation for $\a$ plays a central role, the presence of such quantity on the right hand side of the Bianchi identity for $\a$ motivates the definition of $\ff$.

In the gravitational perturbation of Schwarzschild, the Bianchi and Ricci identities imply that the extreme curvature component $\a$, which is a gauge-invariant quantity, verifies the Teukolsky equation. A fundamental good property of this equation is that it decouples, at the linear level, from any other curvature components. In Schwarzschild, it can be schematically written as:
\beaa
\Box_{\g_{M}} \a+c_1(r)\nabb_4\a+c_2(r)\nabb_3\a+\tilde{V}_1(r)\,\a=0
\eeaa
where $\Box_{\g_M}$ is the d'Alembertian of the Schwarzschild metric, $\nabb_4, \nabb_3$ are the projection on the sphere of $\D_{e_4}, \D_{e_3}$ respectively, and $c_i, \tilde{V}_1$ are smooth functions of an area radius function $r$.

In the case of the Einstein--Maxwell equations, the Bianchi and Maxwell's equations imply that $\a$ satisfy a Teukolsky-type equation in Reissner-Nordstr{\"o}m spacetime (derived in Appendix \ref{computations-appendix}), which can be schematically written as:
\bea\label{schematica}
\Box_{\g_{M, Q}} \a+c_1(r)\nabb_4\a+c_2(r)\nabb_3\a+\tilde{V}_1(r)\,\a=Q \cdot c_3(r)\nabb_4\ff
\eea
where $\Box_{\g_{M, Q}}$ is the d'Alembertian the Reissner-Nordstr{\"o}m metric,  $c_i, \tilde{V}_1$ are smooth functions of $r$, $Q$ is the charge of the Reissner-Nordstr{\"o}m metric $ \g_{M,Q}$,  and $\ff$ is a new gauge-invariant quantity which depends on the electromagnetic tensor.  Observe that the Teukolsky-type equation in Reissner--Nordstr\"om for $\a$ is coupled with the new quantity $\ff$. This makes the previous analysis of the Teukolsky equation in Schwarzschild not applicable to this case.

It is remarkable that the quantity $\ff$ verifies itself a Teukolsky-type equation in Reissner--Nordstr{\"o}m spacetime (derived in Appendix \ref{computations-appendix}), which can be schematically written as:
\bea\label{schematicff}
\Box_{\g_{M, Q}} \ff+d_1(r)\nabb_4\ff+d_2(r)\nabb_3\ff+\tilde{V}_2(r)\,\ff=- Q \cdot  d_3(r) \nabb_3\a
\eea
where $d_i, \tilde{V}_2$ are smooth functions of $r$. Observe that equation \eqref{schematicff} for $\ff$ is coupled back to the curvature component $\a$. 

 The two equations above form what we shall call \textit{the generalized Teukolsky system} for spin $+2$ in Reissner--Nordstr{\"o}m spacetime, which governs the coupled gravitational and electromagnetic perturbation of Reissner-Nordstr{\"o}m. Completely analogous equations are verified by the spin $-2$ components $\aa$ and $\underline{\ff}$, which is obtained after swapping $\nabb_3$ and $\nabb_4$ derivatives. The interaction of these two quantities cannot be decoupled, as opposed to the analogous system in the metric perturbations, see \cite{Moncrief1}, \cite{Moncrief2}.

The main result of this paper concerns the quantities $(\a, \ff)$ verifying such a system.

\subsubsection{Previous works on boundedness and decay for Teukolsky-type equations}

We briefly review here some previous works on quantitative decay for Teukolsky-type equations in different backgrounds. Recall that the $s=0$ Teukolsky equation reduces to the scalar linear wave equation, while the $s=\pm1$ Teukolsky equation is verified by the one-form extreme component of the electromagnetic tensor governed by the Maxwell's equations.

\begin{itemize}
\item {\bf{The case $s=0$ in vacuum:}} Numerous advances on boundedness and decay for solutions to the Cauchy problem for the scalar wave equation on Kerr have been obtained in recent years. From an early result for boundedness of the solutions to the scalar wave equation in Schwarzschild \cite{waveSch}, to more robust techniques introduced in \cite{redshift} and \cite{rp}, there are now complete results on boundedness and decay in Schwarzschild (see \cite{Blue}), in Kerr  for very small angular momentum (see \cite{smalla}, \cite{inventio}, \cite{tataru}, \cite{lectures}, \cite{and}, \cite{Yakov}), and in the full subextremal range of Kerr parameters $|a| < M$ (see \cite{kerr}). 
\item {\bf{The case $s=0$ in electrovacuum:}} The main techniques introduced in the case of scalar wave equations in vacuum spacetimes can be easily generalized to electrovacuum solutions.  In Reissner--Nordstr{\"o}m spacetime, more general results about the wave equation on spherically symmetric, stationary spacetimes can be applied (\cite{stefanos1}, \cite{stefanos2}, \cite{mos}). Some properties of the scalar wave equation in the Kerr family have been extended to the Kerr-Newman family (see \cite{civin}). 
\item {\bf{The case $s=\pm 1$ in vacuum:}} The boundedness and polynomial decay for the Maxwell equations in Schwarzschild, verifying the spin $\pm1$ Teukolsky equation, has been proved in \cite{BlueMax} and  \cite{Federico}. Similarly, results in Kerr spacetime with $|a| \ll M$ are obtained in \cite{andkerr} and \cite{ma1}. 
\item {\bf{The case $s=\pm 1$ in electrovacuum:}} For the Maxwell equations in Reissner-Nordstr{\"o}m spacetime, more general results about spherically symmetric spacetimes can be applied (see \cite{sterbenz}). 
\item {\bf{The case $s=\pm 2$ in vacuum:}} The spin $\pm2$ Teukolsky equation in Schwarzschild has been studied in \cite{DHR} in the context of linearized gravity. Boundedness and decay statements for the Teukolsky equations in Kerr spacetimes with $|a|\ll M$ are proved in \cite{ma2} and \cite{TeukolskyDHR}. 
\end{itemize}

\subsection{The main result and first comments on the proof}\label{first-result}
A rough version of our main result is the following.

\begin{theorem}(Rough version)\label{main-theorem-1-rough} Smooth solutions to the generalized Teukolsky system of spin $\pm2$ on Reissner-Nordstr{\"o}m exterior spacetimes $(\MM, \g_{M,Q})$ with $|Q|\ll M$ arising from smooth initial data which is prescribed on a Cauchy hypersurface $\Sigma_0$ and finite when measured in a higher order and weighted Sobolev norm satisfy energy boundedness, integrated local energy decay and a hierarchy of $r$-weighted energy estimates. 
\end{theorem}

We also obtain higher order versions of the above theorem concerning $\nabb_T$ and $\nabb$ derivatives, where $T$ is the asymptotically timelike Killing field of $(\MM, g_{M,Q})$ and $\nabb$ denotes derivatives in directions tangent to spheres. We also remark that this hierarchy of estimates is such that, using a pigeonhole principle argument as introduced in \cite{rp}, one obtains for $\de>0$ the pointwise decay estimates, for a time function $\tau >0$ 
\beaa
|r^{\frac{5+\de}{2}} \a | \leq C \tau^{-1+\frac{\de}{2}}, \qquad |r^{\frac{5+\de}{2}} \ff | \leq C \tau^{-1+\frac{\de}{2}} \\
|r \aa | \leq C \tau^{-1+\frac{\de}{2}}, \qquad |r^2 \underline{\ff} | \leq C \tau^{-1+\frac{\de}{2}}
\eeaa
where $C$ is some constant depending on an appropriate Sobolev norm of the data. 

The precise statement will be given as Theorem \ref{main-theorem-1} in Section \ref{section-main-theorem}.

\subsubsection{The Chandrasekhar transformation and the Regge-Wheeler system}

A direct analysis of equations of the form \eqref{schematica} or \eqref{schematicff} is prevented by the presence of the first order terms. 
We rely instead on the introduction of two derived quantities $\qf$ and $\qf^\F$, obtained from $\a$ and $\ff$ respectively, which verify a system of what we shall call \textit{generalized Regge--Wheeler equations}, coupled together.
Similar derived quantities were introduced in \cite{DHR} in the Schwarzschild case, and in \cite{TeukolskyDHR} in the Kerr case.

As observed in the case of the Regge--Wheeler-type equation obtained in Kerr in \cite{TeukolskyDHR}, the complete decoupling of the equation is not necessary in the derivation of the estimates. A new important feature appearing in the Einstein--Maxwell equations, which is not present in the vacuum case, are the estimates involving coupling terms of curvature and electromagnetic tensor, which are independent quantities. In addition to those, the coupling of the Regge--Wheeler equations involve lower order terms, as in Kerr (\cite{TeukolskyDHR}, \cite{ma2}).  In order to take into account this whole structure in the estimates, the two equations have to be considered as one system, together with transport estimates for the lower order terms.

We define the derived quantities $\qf$ and $\qf^\F$ as (defined in \eqref{quantities} and \eqref{quantities-2})
\bea\label{relation-sqf}
\begin{split}
\qf&=\frac{1}{\kab}\nabb_3\left(\frac{r}{\kab}\nabb_3(r^3 \kab^2 \a)\right) \\
\qf^\F &= \frac{1}{\kab}\nabb_3(r^3 \kab \ \ff)
\end{split}
\eea
where $\kab:=\tr\chib$ is the trace of the null second fundamental form. They correspond to physical space versions of transformations first considered by Chandrasekhar \cite{Chandra}. Similar physical versions of the Chandrasekhar transformation were first introduced in \cite{DHR} (see Section \ref{relation-DHR} for a comparison of the derived quantities). 
This transformation has the remarkable property of turning the Teukolsky-type equations \eqref{schematica} and \eqref{schematicff} into a system of generalized Regge--Wheeler equations.

A lenghty computation, carried out in Appendix \ref{derivation-Regge}, reveals that $\qf$ and $\qf^\F$ 
verify a system of the following schematic form:
 \begin{equation}\label{systemRW}
\begin{cases}
\Box_{\g_{M, Q}}\mathfrak{q}- V_1(r)\mathfrak{q} = Q \cdot \left(r^{-1} \lapp\mathfrak{q}^\F + \lot(\qf^\F, \a, \ff)\right), \\
\Box_{\g_{M, Q}}\mathfrak{q}^\F- V_2(r) \mathfrak{q}^\F = Q \cdot \left(-r^{-3} \mathfrak{q}+ \lot(\ff)\right)
\end{cases}
\end{equation}
where $V_i$ are positive on the exterior and $\lot$ denotes lower order with respect to derivatives. An analogous system holds in the spin $-2$ case. 
See Propositions \ref{wave-qfF} and \ref{wave-qf}  for the exact form.

 In the case of $Q=0$, as in Schwarzschild spacetime, the system reduces to the first equation, with trivial right hand side. We therefore recover the Regge--Wheeler equations obtained in the context of linear stability of Schwarzschild in \cite{DHR}. 

We emphasize that the particular structure of the coupling terms on the right hand side allows the estimates to be derived as in this paper. In particular, the difference in sign in the highest terms on the right hand side of the system allows the cancellation for the most troubling terms at the photon sphere. See Remark \ref{signs-different}.

We outline here the procedure for the proof of boundedness and decay for the Teukolsky system in the case of small charge. 
\subsubsection{Estimates for the two equations}
We write system \eqref{systemRW} in the following concise form:
\beaa
\begin{cases}
\Big(\square_\g-V_1\Big)\qf= \M_1[\qf, \qf^\F]:=Q\c \C_1[\qf^\F]+Q  \c \Lbb_1[\qf^\F]+Q^2 \c \Lbb_1[\qf], \\
\Big(\square_\g-V_2\Big)\qf^{\F}=\M_2[\qf, \qf^\F]:=Q\c  \C_2[\qf] +Q^2 \c \Lbb_2[\qf^\F]
\end{cases}
\eeaa      
The terms $\C$s and $\Lbb$s are respectively the coupling and the lower order terms. In particular:
\begin{itemize}
 \item The terms $\C_1$ and $\C_2$ represent the coupling between the Weyl curvature and the Ricci curvature. In the wave equation for $\qf$ the coupling term $\C_1=\C_1[\qf^\F]$ is an expression in terms of $\qf^\F$, while in the wave equation for $\qf^\F$ the coupling term $\C_2=\C_2[\qf]$ is an expression in terms of $\qf$. 
 \item The terms $\Lbb_1$ and $\Lbb_2$ collect the lower order terms: in particular $\Lbb_1[\qf]$ are lower order terms with respect to $\qf$ (i.e. they contain $\a$ or one derivative of $\a$), while $\Lbb_1[\qf^\F]$ and $\Lbb_2[\qf^\F]$ are lower terms with respect to $\qf^\F$ (i.e. they contain $\ff$). 
 The index $1$ or $2$ denotes if they appear in the first or in the second equation.
 \end{itemize}
 
 We first derive separated estimates for the two equations of the form 
 \beaa
\Big(\Box_\g-V_i \Big) \Psi_i&=& \M_i
\eeaa
for $i=1,2$, $\Psi_1=\qf$, $\Psi_2=\qf^\F$, keeping the right hand side as it is in the computations.  We apply to both equations separately the standard procedures used to derive energy and Morawetz estimates. Using the Morawetz vector field as multiplier, the redshift estimates and the Dafermos--Rodnianski $r^p$-method, we derive Morawetz estimates and $r^p$-weighted estimates, as in \eqref{final-estimate-Psi}, as well as higher derivative estimates, as in \eqref{final-estimate-Psi-higher}, for the two separated equations. This is done in Section \ref{separate-estimates}.
 
 Notice that these separated estimates will contain on the right hand side terms involving $\M_1$ and $\M_2$ that at this stage are not controlled.  In particular $\M_1$ and $\M_2$ contain both the coupling terms $\C$ and the lower order terms $\Lbb$.

\subsubsection{Estimates for the coupling terms}

In deriving estimates for the coupling terms, the main obstacle we encounter is to obtain, due to trapping at the photon sphere, the $T$-energy estimate. We observe that the structure of the right hand side in the two equations of the system is not symmetric. In particular, the coupling term $\C_1[\qf^\F]$ in the first equation involves up to two derivatives of $\qf^\F$, while the coupling term $\C_2[\qf]$ in the second equation contains $0$th-order derivative of $\qf$. Regularity considerations mean that one must try and obtain the $T$-energy estimate for $\qf$ and the $T$-energy estimate for $\nabb_T \qf^\F$, $\nabb \qf^\F$ at the same time. 

 In order to take into account the difference in the presence of derivatives, we consider the $0$th-order Morawetz and $r^p$-weighted estimate for the first equation and the $1$st-order estimate for the second equation, and we add them together.  This operation will create a combined estimate, where the Morawetz bulks on the left hand side of each equations ought to absorb the coupling term on the right hand side of the other equation. This means that we should control, at the top order, spacetime integrals of $C_1 \lapp \qf^\F \nabb_T \qf + C_2 \nabb_T {\nabb^2}_T \qf^\F +C_3 \nabb \qf \nabb_T \nabb \qf^\F$ by the initial energies of $\qf$ and $\nabb_T \qf^\F, \nabb \qf^\F$. Here $C_i$ are freely chosen constants coming from choosing $C_i T$ as multiplier. 
 
 Since the Morawetz bulk energies are degenerate at the photon sphere for the angular and the $T$ derivative, obtaining such bounds for any constants would be impossible. However, we show that, after integrating by parts and using the wave equation for $\qf^\F$, a clever choice of constants (normalized relative to the photon sphere) actually allows these troublesome terms to be expressed in terms of either bulk integrals concerning derivatives controlled without degeneracy at the photon sphere, or pure boundary terms which can be absorbed into the energy. Such a choice of constant and the special structure of the coupling terms $\C_1[\qf^\F]$ and $\C_2[\qf]$ imply a cancellation of problematic terms in the trapping region. This is done in Section \ref{coupling-subsection}.

\subsubsection{Estimates for the lower order terms}

The lower order terms $\Lbb_1[\qf]$, $\Lbb_1[\qf^\F]$ and $\Lbb_2[\qf^\F]$ contain expressions in $\a$ and $\ff$, which are not contained in the Morawetz estimates, and for this reason we need transport estimates for $\a$ and $\ff$.  This is achieved by viewing the relations \eqref{relation-sqf} between the derived quantities and the original quantities as transport equations along the null cones generated by the vectorfields $e_3$ and $e_4$. Using these estimates, we invert the transformation theory so as to upgrade the estimates on the generalized Regge--Wheeler system to the spin $\pm2$ system. 
 This is done in Section \ref{section-lower-order}.

Summing the separated estimates and absorbing the coupling terms and the lower order terms on the right hand side we obtain a combined estimate for the system as in the Main Theorem.

\subsubsection{The full subextremal range $|Q|<M$}

During the preparation of this paper, boundedness and decay for the gauge-invariant quantities governing the gravitational and electromagnetic perturbations of Reissner-Nordstr\"om have been obtained by the author in the full subextremal range $|Q|<M$ (see \cite{Giorgi7}). The extension to the full subextremal range has been possible by considering a different system of Regge-Wheeler equations as the one here used.

The system analyzed in \cite{Giorgi7} is composed of the Regge-Wheeler-type equation for $\qf^\F$ here derived (i.e. the second equation in \eqref{systemRW}) and of another Regge-Wheeler equation obtained from a gauge-invariant quantity $\pf$ of spin 1 introduced in \cite{Giorgi5}. The gauge-invariant quantity $\qf$ can be expressed in terms of $\qf^\F$ and the new quantity $\pf$, and therefore the dependence on $\qf$ can be altogether eliminated. Through this elimination, the system of equations becomes diagonalizable and presents right hand sides which are symmetric. This is in contrast with the system \eqref{systemRW} here analyzed, which has non symmetric right hand sides, and for which the analysis can be obtained for very small $Q$ only. Such symmetry is used in \cite{Giorgi7} to define a combined energy-momentum tensor of the system which allows for a cancellation of the highest order terms, without recurring to smallness of the charge.

\subsection{Outline of the paper}

The paper is organized as follows. 

In Section \ref{VEeqDNGsec}, we introduce the general framework of null frames, and we write the Einstein--Maxwell equations in such null frames. 
In Section \ref{RN-generale}, we introduce the Reissner--Nordstr{\"o}m metric $(\MM, \g_{M,Q})$ and recall its main properties. 
In Section \ref{linearized-equations}, we derive the linearization of the Einstein-Maxwell equations around the Reissner--Nordstr{\"o}m solution. 

In Section \ref{section-gauge-invariant}, we describe the gauge-invariant quantities $\a$ and $\ff$ which verify the generalized Teukolsky system, presented in Section \ref{spin2}, together with the generalized Regge--Wheeler system verified by $\qf$ and $\qf^\F$. 

In Section \ref{transformation-theory}, we present the Chandrasekhar transformation, relating the generalized Teukolsky system and the generalized Regge--Wheeler system.
In Section \ref{statement}, we define the main weighted energies and bulks used in the estimates and state the theorem.
In Section \ref{separate-estimates}, the estimates for the separated equations are carried out, and in Section \ref{lot-absorbing} the estimates for the coupling and the lower order terms are derived.
Finally, in Section \ref{proof-of-theorem-this-is-the-end}, we summarize the previous steps into the proof of the Main Theorem.

In Appendix \ref{computations-appendix}, we collect the computations in the derivation of the generalized Teukolsky system, and in Appendix \ref{derivation-Regge}, we show the derivation of the generalized Regge--Wheeler system through the Chandrasekhar transformation. 

\bigskip

{\bf{Acknowledgements}}  The author is grateful to Sergiu Klainerman and Mu-Tao Wang for comments and suggestions.

\section{The Einstein--Maxwell equations in null frames}
\label{VEeqDNGsec}
In this section, we review the general form of the Einstein-Maxwell equations \eqref{Einsteineq} and \eqref{Max} written with respect to a local null frame of a Lorentzian manifold. In this section, we will derive the main equations in their full generality.

\subsection{Preliminaries} \label{sec:genmfld}
Let $(\MM, \g)$ be a $3+1$-dimensional Lorentzian manifold, and let $\D$ be the covariant derivative associated to $\g$.

\subsubsection{Local null frames}\label{local-null-frames}
Suppose that the Lorentzian manifold $(\MM, \g)$ can be foliated by spacelike $2$-surfaces $(S,\slashed{g}) $, where $\slashed{g}$ is the pullback of the metric $\g$ to $S$.  To each point of $\MM$, we can associate a null frame $\mathscr{N}=\left\{e_A, e_3, e_4\right\}$, with $\{ e_A \}_{A=1,2}$ being tangent vectors to $(S,\slashed{g}) $, such that the following relations hold
\bea\label{relations-null-frame}
\begin{split}
\g\left(e_3,e_3\right) &= 0, \qquad \g\left(e_4,
e_4 \right) = 0, \qquad \g\left(e_3,e_4\right) = -2\\ 
\g\left(e_3,
e_A \right) &= 0 \ \ \ , \ \ \  \g\left(e_4, e_A\right) = 0 \ \ \ , \ \ \ \g \left(e_A, e_B \right) = \slashed{g}_{AB} \, .
\end{split}
\eea

\subsubsection{S-tensor algebra}\label{tensor-algebra}
In Section \ref{sec:rccc}, we will express the Ricci coefficients, curvature and electromagnetic components with respect to a null frame $\mathscr{N}$ associated to a foliation of surfaces $S$. These objects are therefore $S$-tangent tensors. We will use the standard notations for operations on $S$-tangent tensor. See for example Section 3.1.3. of \cite{DHR}. 

We recall the definition of the projected covariant derivatives and the angular operator on $S$-tensors. We denote $\nabb_3=\nabb_{e_3}$ and $\nabb_4=\nabb_{e_4}$ the projection to $S$ of the spacetime covariant derivatives $\D_{e_3}$ and $\D_{e_4}$ respectively.

We define the following angular operator on $S$-tensors (see \cite{Ch-Kl}). Let $\xi$ be an arbitrary one-form and $\th$ an arbitrary symmetric traceless $2$-tensor on $S$. 
\begin{itemize}
\item $\nabb$ denotes the covariant derivative associated to the metric $\slashed{g}$ on $S$.
\item $\DDd_1$ takes $\xi$ into the pair of functions $(\divv \xi, \curll \xi)$, where $$\divv \xi=\slashed{g}^{AB} \nabb_A \xi_B, \qquad \curll\xi=\slashed{\ep}^{AB}\nabb_A \xi_B$$
\item $\DDs_1$ is the formal $L^2$-adjoint of $\DDd_1$, and takes any pair of functions $(\rho, \sigma)$ into the one-form $-\nabb_A \rho+\slashed{\ep}_{AB} \nabb^B \sigma$.
\item  $\DDd_2$ takes $\th$ into the one-form $(\divv \th)_C=\slashed{g}^{AB}\nabb_A \th_{BC}$.
\item $\DDs_2$ is the formal $L^2$-adjoint of $\DDd_2$, and takes $\xi$ into the symmetric traceless two tensor $$(\DDs_2\xi)_{AB}=-\frac 12 \left( \nabb_B\xi_A+\nabb_A\xi_B-(\divv \xi)\slashed{g}_{AB}\right)$$
\end{itemize}

We recall the relations between the angular operators and the laplacian $\lapp$ on $S$:
\bea\label{angular-operators}
\begin{split}
\DDd_1 \DDs_1&=-\lapp_0, \qquad \DDs_1 \DDd_1= -\lapp_1+K, \\
\DDd_2 \DDs_2&= -\frac 1 2 \lapp_1-\frac 1 2 K, \qquad \DDs_2 \DDd_2= -\frac 1 2 \lapp_2+K
\end{split}
\eea
where $\lapp_0$ and $\lapp_1$ are the laplacian on scalars and on $1$-form respectively, and $K$ is the Gauss curvature of the surface $S$.

\subsection{Ricci coefficients, curvature and electromagnetic components} \label{sec:rccc}
We now define the Ricci coefficients, curvature and electromagnetic components associated to the metric $\g$ with respect to the null frame $\mathscr{N}=\left\{e_A, e_3, e_4\right\}$,  where the indices $A,B$ take values $1,2$. 

\subsubsection{Ricci coefficients}\label{ricci-coefficients}

 We define the Ricci coefficients associated to the metric $\g$ with respect to the null frame $\mathscr{N}$ in the following way (see \cite{Ch-Kl}):
   \bea\label{def1}
   \begin{split}
   \chi_{AB}:&=\g(\D_A  e_4, e_B), \qquad  \chib_{AB}:=\g(\D_A  e_3, e_B)\\
   \eta_A:&=\frac 1 2 \g(\D_3 e_4, e_A), \qquad \etab_A:=\frac 1 2 \g(\D_4 e_3, e_A),\\
    \xi_A:&=\frac 1 2 \g(\D_4 e_4, e_A), \qquad \xib_A:=\frac 1 2 \g(\D_3 e_3, e_A)\\
     \om:&=\frac 1 4 \g(\D_4 e_4, e_3), \qquad \omb:=\frac 1 4\g(\D_3 e_3, e_4)\\
   \ze_A:&=\frac 1 2 \g( \D_A e_4, e_3), \qquad 
   \end{split}
   \eea
   It is natural to decompose the $2$-tensor $\chi_{AB}$ into its tracefree part $\chih_{AB}$, a symmetric traceless 2-tensor on $S$, and its trace $\ka:=\tr\chi$. In particular we write $ \chi_{AB}=\frac 1 2 \ka \ \slashed{g}_{AB}+\chih_{AB}$, with $\slashed{g}^{AB} \chih_{AB}=0$ and $\ka=\slashed{g}^{AB} \chi_{AB}$. Similarly for $\chib_{AB}$.

It follows from \eqref{def1} that we have the following relations for the commutators of the null frame:
\bea\label{commutators}
\begin{split}
\big[e_3, e_A\big]&= \nabb_3 e_A+(\eta_A-\ze_A) e_3+\xib_A e_4-\chib_{AB} e^B, \\
\big[e_4, e_A\big]&= \nabb_4 e_A+(\etab_A+\ze_A) e_4+\xi_A e_3-\chi_{AB} e^B, \\
\big[e_3, e_4\big]&= -2\om e_3+2\omb  e_4+2(\eta^A -\etab^A) e_A
\end{split}
\eea

\subsubsection{Curvature components}
Let $\W$ denote the Weyl curvature of $\g$ and let $\dual \W$ denote the Hodge dual on $(\MM, \g)$ of $\W$. We define the null curvature components in the following way (see \cite{Ch-Kl}):  
   \bea\label{def3}
\begin{split}
\a_{AB}:&=\W(e_A, e_4, e_B, e_4), \qquad \aa_{AB}:= \W(e_A, e_3, e_B, e_3) \\
\b_A:&=\frac 1 2 \W(e_A, e_4, e_3, e_4), \qquad  \bb_{A} :=\frac 1 2 \W(e_A, e_3, e_3, e_4) \\
\rho:&=\frac 14 \W(e_3, e_4, e_3, e_4) \qquad \sigma:=\frac 1 4  \dual \W (e_3, e_4, e_3, e_4)
\end{split}
\eea
The remaining components of the Weyl tensor are given by 
\beaa
\W_{AB34}&=&2\sigma \ep_{AB}, \qquad \W_{ABC3}=\ep_{AB} \dual \bb_C, \qquad \W_{ABC4}=\ep_{AB} \dual \b_C, \\
\W_{A3B4}&=&-\rho \de_{AB}+\sigma \ep_{AB}, \qquad \W_{ABCD}=-\ep_{AB}\ep_{CD} \rho
\eeaa

\subsubsection{Electromagnetic components}
Let $\F$ be a $2$-form in $(\MM, \g)$, and let $\dual \F$ denote the Hodge dual on $(\MM, \g)$ of $\F$. We define the null electromagnetic components in the following way (see \cite{Zipser} and \cite{Federico}):
\bea\label{decomposition-F}
\begin{split}
\bF_{A}:&=\F(e_A, e_4), \qquad \bbF_{A}:= \F(e_A, e_3) \\
 \rhoF:=& \frac 1 2 \F(e_3, e_4), \qquad \sigmaF:=\frac 1 2 \dual \F(e_3, e_4) 
\end{split}
\eea
The only remaining component of $\F$ is given by $\F_{AB}=-\ep_{AB}\sigmaF$.

\begin{remark} The notation in both \cite{Zipser} and \cite{Federico} differs to ours. In those previous works, the extreme components of $\F$ are denoted by $\alpha(F)$, as opposed to $\beta(F)$. By using our notation, we want to stress the fact that $\bF$ and $\bbF$ are not gauge-invariant quantities in the gravitational and electromagnetic perturbations of Reissner-Nordstr{\"o}m, leading to major differences in the treatment of the estimates. See Section  \ref{section-gauge-invariant}.
\end{remark}

\subsection{The Einstein-Maxwell equations} \label{nseq}
If $(\MM, \g)$ satisfies the Einstein-Maxwell equations

\bea \label{Einstein-Maxwell-eq}
\R_{\mu\nu}&=&2 \F_{\mu \lambda} {\F_\nu}^{\lambda}- \frac 1 2 \g_{\mu\nu} \F^{\alpha\beta} \F_{\alpha\beta}, \label{Einstein-1}\\
\D_{[\alpha} \F_{\beta\gamma]}&=&0, \qquad \D^\alpha \F_{\alpha\beta}=0. \label{Maxwell}
\eea
the Ricci coefficients, curvature and electromagnetic components defined in \eqref{def1}, \eqref{def3} and \eqref{decomposition-F} satisfy a system of equations, which is presented in this section.

\subsubsection{The null structure equations} \label{sec:nse}

The first equations for $\chi$ and $\chib$ are given by
\beaa
\nabb_3 \chib_{AB}+\chib_{A}^C \chib_{CB}+2\omb \chib_{AB}  &=& 2\nabb_B \xib_A+2\eta_B \xib_A+2\etab_A \xib_B-4\ze_B\xib_A+ \R_{A33B}, \\
 \nabb_4 \chi_{AB}+\chi_{A}^C \chi_{CB}+2\om \chi_{AB}  &=& 2\nabb_B \xi_A+2\eta_B \xi_A+2\eta_A \xi_B+4\ze_B\xi_A+\R_{A44B}
\eeaa
and separating in the symmetric traceless part and in the trace part, we obtain
\bea\label{nabb-3-chibh-general}
\begin{split}
\nabb_3 \chibh+\kab \ \chibh+2\omb \chibh&=&-2\DDs_2\xib -\aa+2(\eta+\etab-2\ze)\hot \xib,  \\
\nabb_4 \chih+\ka \ \chih+2\om \chih&=&-2\DDs_2\xi -\a +2(\eta+\etab+2\ze)\hot \xi
\end{split}
\eea
and
\bea\label{nabb-3-kab-general}
\begin{split}
\nabb_3\kab+\frac 1 2 \kab^2+2\omb \kab=2\divv\xib-(\chibh, \chibh )+2(\eta+\etab-2\ze)\c \xib-2 (\bbF,  \bbF), \\
\nabb_4\ka+\frac 1 2 \ka^2+2\om \ka=2\divv\xi-(\chih, \chih )+2(\eta+\etab+2\ze)\c \xi-2 (\bF,  \bF)
\end{split}
\eea

The second equations for $\chi$ and $\chib$ are given by 
\beaa
\nabb_4\chib_{AB}  &=&  2\nabb_B \etab_A +2\om \chib_{AB} - \chi_B^C \chib_{AC}  +2(\xi_B \xib_A+ \etab_B \etab_A) + \R_{A34B}, \\
\nabb_3\chi_{AB}  &=&  2\nabb_B \eta_A +2\omb \chi_{AB} - \chib_B^C \chi_{AC}  +2(\xib_B \xi_A+ \eta_B \eta_A) + \R_{A43B}, \\
\eeaa
and separating in the symmetric traceless part and in the trace part, we obtain
\bea\label{nabb-3-chih-general}
\begin{split}
\nabb_3\chih+\frac 12  \kab \,\chih-2 \omb \chih   &=&  -2 \slashed{\mathcal{D}}_2^\star \eta-\frac 1 2 \ka \chibh  + (\eta\hot \eta)+ (\xib\hot \xi) -(\bF \hot \bbF) \\
\nabb_4\chibh+\frac 12  \ka \,\chibh -2 \om \chibh &=&  -2 \slashed{\mathcal{D}}_2^\star \etab -\frac 1 2 \kab \chih  + (\etab\hot \etab)+ (\xib\hot \xi) -(\bF \hot \bbF),
\end{split}
\eea
and
\bea\label{nabb-3-ka-general}
\begin{split}
\nabb_3 \ka +\frac 12  \ka\, \kab-2\omb\,\ka &=&2 \divv \eta  -(\chih,  \chibh)+2(\xi,\xib) +2(\eta,  \eta)+2\rho \\
\nabb_4 \kab +\frac 12  \ka\,\kab-2\om\,\kab  &=&2 \divv \etab -(\chih,  \chibh)+2(\xi,\xib) +2(\etab,  \etab)+2\rho 
\end{split}
\eea
while the antisymmetric part is given by 
\bea\label{curl-eta-general}
\slashed{\curl}\eta &=& -\frac 1 2\chi \wedge \chib  +\sigma,  \qquad \slashed{\curl}\etab = \frac 1 2\chi \wedge \chib   -\sigma
\eea
The equations for $\ze$ are given by 
\beaa
\nabb_3 \ze&=&-2 \nabb \omb -\chib \c (\ze+\eta)+2\omb(\ze-\eta)+\chi \c \xib+2\om \xib - \frac 1 2 \R_{A334}, \\
-\nabb_4 \ze&=&-2 \nabb \om -\chi \c (-\ze+\etab)+2\om(-\ze-\etab)+\chib \c \xi+2\omb \xi - \frac 1 2 \R_{A443}
\eeaa
and therefore reducing to
\bea\label{nabb-3-ze-general}
\begin{split}
\nabb_3 \ze&=&-2 \nabb \omb -\chib \c (\ze+\eta)+2\omb(\ze-\eta)+\chi \c \xib+2\om \xib - \bb +\sigmaF{\ep_{A}}^C\bbF_C -\rhoF\bbF, \\
\nabb_4 \ze&=&2 \nabb \om +\chi \c (-\ze+\etab)+2\om(\ze+\etab)-\chib \c \xi-2\omb \xi -\b -\sigmaF \ep \c \bF -\rhoF\bF 
\end{split}
\eea
The equations for $\xi$ and $\xib$ are given by 
\beaa
\nabb_4\xib- \nabb_3 \etab &=&4\om\xib+ -\chib\c(\eta-\etab) -\frac 1 2 \R_{A334}, \\
\nabb_3\xi- \nabb_4 \eta &=&4\omb\xi+ \chi\c(\eta-\etab) -\frac 1 2 \R_{A443}
\eeaa
and therefore reducing to
\bea\label{nabb-4-xib-general}
\begin{split}
\nabb_4\xib- \nabb_3 \etab &=& -\chib\c(\eta-\etab)+4\om\xib -\bb +\sigmaF \ep \c \bbF -\rhoF\bbF, \\
\nabb_3\xi-  \nabb_4 \eta &=& \chi\c(\eta-\etab)+4\omb\xi +\b +\sigmaF \ep \c \bF +\rhoF\bF \label{nabb-3-xi-general}
\end{split}
\eea
The equation for $\om$ and $\omb$ is given by
\beaa
\nabb_4 \omb+\nabb_3\om&=& 4\om\omb+\xi\c\xib +\ze\c(\eta-\etab) -\eta\c\etab+\frac 1 4 \R_{3434}
\eeaa
and therefore reducing to
\bea\label{nabb-4-omb-nabb-3-om}
\nabb_4 \omb+\nabb_3\om&=& 4\om\omb +\ze\c(\eta-\etab)+\xi\c\xib -\eta\c\etab+\rho + \rhoF^2 -\sigmaF^2
\eea
The spacetime equations that generate Codazzi equations are
\beaa
\nabb_C \chib_{AB} + \ze_B \chib_{AC} &=& \nabb_B \chib_{AC} + \ze_C \chib_{AB} + \R_{A3CB}, \\
\nabb_C \chi_{AB} - \ze_B \chi_{AC} &=& \nabb_B \chi_{AC} - \ze_C \chi_{AB} + \R_{A4CB}
\eeaa
and, taking the trace in $C,A$ we obtain
\bea\label{codazzi-general}
\begin{split}
\divv \chibh_B&=&(\chibh\c \ze)_B-\frac 1 2 \kab \ze_B+\frac 1 2 (\nabb_B \kab) +\bb_B+\sigmaF{\ep_{B}}^C\bbF_C -\rhoF\bbF_B, \\
\divv \chih_B&=&-(\chih\c \ze)_B+\frac 1 2 \ka \ze_B+\frac 1 2 (\nabb_B \ka) -\b_B+\sigmaF{\ep_{B}}^C\bF_C +\rhoF\bF_B
\end{split}
\eea
The spacetime equation that generates Gauss equation is 
\beaa
\gS^{AC} \gS^{BD} \R_{ABCD} = 2K + \frac 1 2 \ka\kab - \chih \c \chibh
\eeaa
therefore reducing to
\bea\label{Gauss-general}
K=- \frac 1 4 \ka\kab + \frac 1 2 (\chih, \chibh)-\rho+\rhoF^2 -\sigmaF^2  
\eea

\subsubsection{The Maxwell equations}\label{Maxwell-general}
We derive the null decompositions of Maxwell equations \eqref{Maxwell}. 

The equation $\D_{[\alpha} \F_{\beta\gamma]}=0$ gives three independent equations. The first one is given by
\bea\label{Maxwell-1}
\begin{split}
\nabb_3 \bF_A-\nabb_4 \bbF_A&= -\left(\frac 1 2 \kab-2\omb\right) \bF_A+\left(\frac 1 2 \ka -2\om\right) \bbF_A+ 2\nabb_A\rhoF+2 (\eta_A+\etab_A)\rhoF\\
&-2(\eta^B-\etab^B)\ep_{AB} \sigmaF+(\chih \c \bbF)_A-(\chibh \c \bF)_A
\end{split}
\eea
The second and third equation are given by
\bea\label{nabb-3-sigmaF-general}
\begin{split}
\nabb_3 \sigmaF+\kab \ \sigmaF&=& -\slashed{\curl}\bbF-(\ze-\eta)\wedge \bbF+\xib \wedge \bF, \\
 \nabb_4 \sigmaF+\ka \ \sigmaF&=& \slashed{\curl}\bF+(\ze+\etab)\wedge \bF+\xi \wedge \bbF
 \end{split}
\eea
The equation $\D^\mu \F_{\mu\nu}=\slashed{\g}^{BC} \D_B \F_{C\nu}-\frac 1 2 \D_4 \F_{3\nu}-\frac 1 2 \D_3\F_{4\nu}=0$ gives three more independent equations. The first one is given by
\bea\label{Maxwell-2}
\begin{split}
\nabb_3 \bF_A+\nabb_4 \bbF_A&=-\left(\frac 1 2 \kab -2\omb \right)\bF_A-\left(\frac 1 2 \ka -2\om \right)\bbF_A+2(\eta_A-\etab_A)\rhoF\\
&-2\ep_{AC}\nabb^C \sigmaF-2(\eta^B+\etab^B)\ep_{AB} \sigmaF+ (\chibh \c \bF)_A+ (\chih \c \bbF)_A
\end{split}
\eea
Summing and subtracting \eqref{Maxwell-1} and \eqref{Maxwell-2} we obtain respectively
\bea
\nabb_3 \bF_A+\left(\frac 1 2 \kab-2\omb\right) \bF_A&=& -\DDs_1(\rhoF, \sigmaF) +2\eta_A\rhoF-2\eta^B\ep_{AB} \sigmaF+(\chih \c \bbF)_A, \label{nabb-3-bF-general}\\
\nabb_4 \bbF_A+\left(\frac 1 2 \ka -2\om \right)\bbF_A&=&\DDs_1(\rhoF, -\sigmaF)-2\etab_A\rhoF-2\etab^B\ep_{AB} \sigmaF+ (\chibh \c \bF)_A
\eea
The last two equations are given by 
\bea\label{nabb-4-rhoF-general}
\begin{split}
\nabb_4 \rhoF+ \ka \rhoF&=& \divv\bF+(\ze+\etab) \c \bF-\xi \c \bbF, \\
\nabb_3 \rhoF+ \kab \rhoF&=& -\divv\bbF+(\ze-\eta) \c \bbF-\xib \c \bF
\end{split}
\eea

\subsubsection{The Bianchi equations} \label{bieq}
The Bianchi identities for the Weyl curvature are given by 
\beaa
 \D^\a \W_{\a\b\gamma\delta}&=&\frac 1 2 (\D_\gamma \R_{\b\delta}-\D_{\delta}\R_{\b\gamma})=:J_{\b\gamma\delta} \\
 \D_{[\sigma }\W_{\gamma\delta] \a\b}&=&\g_{\delta \b}J_{\a\gamma\sigma}+\g_{\gamma \a}J_{\b\delta\sigma}+\g_{\sigma \b}J_{\a\delta\gamma}+\g_{\delta \a}J_{\b\sigma\gamma}+\g_{\gamma \b}J_{\a\sigma\delta}+\g_{\sigma \a}J_{\b\gamma\delta}:= \tilde{J}_{\sigma\gamma\delta\a\b}
\eeaa 
The Bianchi identities for $\a$ and $\aa$ are given by
\beaa
\nabb_3\a_{AB}+\frac 1 2 \kab\,\a_{AB}-4\omb \a_{AB}&=&-2 (\DDs_2\, \b)_{AB} -3(\chih_{AB} \rho+\dual\, \chih_{AB} \sigma)+((\ze+4\eta)\hat{\otimes} \b )_{AB}+\\
&&+\frac 1 2 (\tilde{J}_{3A4B4}+\tilde{J}_{3B4A4}+J_{434}\slashed{\g}_{AB})\\
\nabb_4\aa_{AB}+\frac 1 2 \ka\,\aa_{AB}-4\om \aa_{AB}&=&2 (\DDs_2\, \bb)_{AB} -3(\chibh_{AB} \rho+\dual\, \chibh_{AB} \sigma)-((-\ze+4\etab)\hat{\otimes} \bb )_{AB}+\\
&&+\frac 1 2 (\tilde{J}_{4A3B3}+\tilde{J}_{4B3A3}+J_{343}\slashed{\g}_{AB})
\eeaa
Using that $\tilde{J}_{3A4B4}=-\slashed{\g}_{A B}J_{4 34}+2J_{B A4}$, it is reduced to
\bea\label{Bianchi-non-linear1}
\begin{split}
\nabb_3\a_{AB}+\frac 1 2 \kab\,\a_{AB}-4\omb\a_{AB}&=&-2 (\DDs_2\, \b)_{AB} -3(\chih_{AB} \rho+\dual\, \chih_{AB} \sigma)+((\ze+4\eta)\hat{\otimes} \b )_{AB}+\\
&&+J_{B A4}+J_{AB4}-\frac 1 2 \slashed{g}_{A B}J_{4 34}, \\
\nabb_4\aa_{AB}+\frac 1 2 \ka\,\aa_{AB}-4\om \aa_{AB}&=&2 (\DDs_2\, \bb)_{AB} -3(\chibh_{AB} \rho+\dual\, \chibh_{AB} \sigma)-((-\ze+4\etab)\hat{\otimes} \bb )_{AB}+\\
&&+J_{B A3}+J_{AB3}-\frac 1 2 \slashed{g}_{A B}J_{3 43}\end{split}
\eea 
The Bianchi identities for $\b$ and $\bb$ are given by 
\bea\label{Bianchi-b-1}
\begin{split}
\nabb_4 \b_A+ 2\ka  \b_A+2\om \b_A &=&\divv\a_A +((2\ze+\etab)\c \a)_A+3(\xi_A \rho+\dual \xi_A \sigma) -J_{4A4}, \label{nabb-4-b-general}\\
\nabb_3 \bb_A+ 2\kab  \bb_A+2\omb \bb_A &=&-\divv\aa_A +((2\ze-\eta)\c \aa)_A-3(\xib_A \rho+\dual \xib_A \sigma) +J_{3A3}
\end{split}
\eea
and 
\bea\label{Bianchi-b-2}
\begin{split}
\nabb_3 \b_A+\kab \b_A-2\omb\, \b_A &=&\DDs_1(-\rho, \sigma)_A+2(\chih \c\bb)_A +\xib \c \a+3(\eta_A \rho +\dual\eta_A\,   \sigma) + J_{3A4}, \\
\nabb_4 \bb_A+\ka \bb_A-2\om\, \bb_A &=&\DDs_1(\rho, \sigma)_A+2(\chibh \c\b)_A-\xi \c \aa -3(\etab_A \rho -\dual\etab_A\,   \sigma) - J_{4A3} 
\end{split}
\eea
The Bianchi identity for $\rho$ is given by 
\bea\label{nabb-4-rho-general}
\begin{split}
\nabb_4 \rho+\frac 3 2 \ka \rho&=&\divv\b+(2\etab+\ze)\c\b  -\frac 1 2 (\chibh \c \a )-2\xi \c \bb -\frac 1 2 J_{434}, \\
\nabb_3 \rho+\frac 3 2 \kab \rho&=&-\divv\bb-(2\eta-\ze)\c\bb  +\frac 1 2 (\chih \c \aa )+2\xib \c \b -\frac 1 2 J_{343}
\end{split}
\eea
The Bianchi identity for $\sigma$ is given by 
\beaa
\nabb_4 \sigma+\frac 3 2 \ka \sigma&=&-\slashed{\curl}\b-(2\etab+\ze)\wedge\b  +\frac 1 2 \chibh \wedge \a  -\frac 1 2 \dual J_{434}, \\
\nabb_3 \sigma+\frac 3 2 \kab \sigma&=&-\slashed{\curl}\bb-(2\eta-\ze)\wedge\b  -\frac 1 2 \chih \wedge \aa  +\frac 1 2 \dual J_{343}
\eeaa
and writing $\dual J_{434}=\frac 1 2 J_{4\mu\nu}{\ep^{\mu\nu}}_{34}=- J_{4AB}\ep_{AB}=(J_{AB4}-J_{BA4})\ep^{AB}$, we obtain
\bea\label{nabb-4-sigma-general}
\begin{split}
\nabb_4 \sigma+\frac 3 2 \ka \sigma&=&-\slashed{\curl}\b-(2\etab+\ze)\wedge\b  +\frac 1 2 \chibh \wedge \a  -\frac 1 2 (J_{AB4}-J_{BA4})\ep^{AB}, \\
\nabb_3 \sigma+\frac 3 2 \kab \sigma&=&-\slashed{\curl}\bb-(2\eta-\ze)\wedge\bb  -\frac 1 2 \chih \wedge \aa  +\frac 1 2 (J_{AB3}-J_{BA3})\ep^{AB}
\end{split}
\eea

\section{The Reissner-Nordstr{\"o}m spacetime}\label{RN-generale}

In this section, we introduce the Reissner-Nordstr{\"o}m exterior metric, as well as relevant background structure. We collect here standard coordinate transformations relevant to the study of Reissner-Nordstr{\"o}m spacetime. See for example \cite{Exact}. 

 We first fix in Section \ref{diff-structure} an ambient manifold-with-boundary $\MM$ on which we define the Reissner-Nordstr{\"o}m exterior metric $ \g_{M, Q}$ with parameters $M$ and $Q$ verifying $|Q| < M$. We shall then pass to more convenient sets of coordinates, like the double null coordinates, the outgoing and ingoing Eddington-Finkelstein coordinates. Finally we will show how these sets of coordinates relate to the standard form of the metric as given in \eqref{RNintro}. 

We will follow closely Section 4 of \cite{DHR}, where the main features of the Schwarzschild metric and differential structure are easily extended to the Reissner-Nordstr{\"o}m solution.

\subsection{Differential structure and metric}\label{diff-structure}

We define in this section the underlying differential structure and metric in terms of the Kruskal coordinates. 

\subsubsection{Kruskal coordinate system}

Define the manifold with boundary
\begin{align} \label{SchwSchmfld}
\mathcal{M} := \mathcal{D} \times S^2 := \left(-\infty,0\right] \times \left(0,\infty\right) \times S^2
\end{align}
with coordinates $\left(U,V,\theta^1,\theta^2\right)$.
We will refer to these coordinates as \emph{Kruskal coordinates}.
 The boundary 
\[
\mathcal{H}^+:=\{0\}  \times \left(0,\infty\right) \times S^2
\]
will be referred to as the \emph{horizon}. 
We denote by $S^2_{U,V}$ the $2$-sphere $\left\{U,V\right\} \times S^2 \subset \mathcal{M}$ in $\mathcal{M}$.

\subsubsection{The Reissner-Nordstr{\"o}m metric}

We define the Reissner-Nordstr{\"o}m metric on $\mathcal{M}$ as follows.
Fix two parameters $M>0$ and $Q$, verifying $|Q|<M$. Let the function $r : \mathcal{M} \rightarrow \left[M+\sqrt{M^2-Q^2},\infty\right)$ be 
given implicitly as a function of the coordinates $U$ and $V$ by
\beaa
-UV = \frac{4r_{+}^4}{(r_{+}-r_{-})^2} \Big|\frac{r-r_{+}}{r_{+}}\Big| \Big|\frac{r_{-}}{r-r_{-}}\Big|^{\left(\frac{r_{-}}{r_{+}}\right)^2} \exp\Big(\frac{r_{+}-r_{-}}{r_{+}^2} r\Big) \, ,
\eeaa
where 
\bea\label{definiion-rpm}
r_{\pm}=M\pm \sqrt{M^2-Q^2}
\eea
 We will also denote 
\bea\label{def-rH}
r_{\mathcal{H}}=r_+=M+\sqrt{M^2-Q^2}
\eea
Define also
\beaa
\Up_K \left(U,V\right) &=& \frac{r_{-}r_{+}}{4r(U,V)^2} \Big( \frac{r(U,V)-r_{-}}{r_{-}}\Big)^{1+\left(\frac{r_{-}}{r_{+}}\right)^2}\exp\Big(-\frac{r_{+}-r_{-}}{r_{+}^2} r(U,V)\Big) \\
 \gamma_{AB} &=& \textrm{standard metric on $S^2$} \, .
\eeaa
Then the Reissner-Nordstr{\"o}m metric $\g_{M, Q}$ with parameters $M$ and $Q$ is defined to be the metric:
\begin{align} \label{sskruskal}
\g_{M, Q} = -4 \Up_K \left(U,V\right) d{U} d{V} +   r^2 \left(U,V\right) \gamma_{AB} d{\theta}^A d{\theta}^B.
\end{align}
Note that the horizon $\mathcal{H}^+=\partial\mathcal{M}$ 
is a null hypersurface with respect to $\g_{M, Q}$. We will use the standard
spherical coordinates $(\theta^1,\theta^2)=(\theta, \phi)$, in which case
the metric $\gamma$ takes the explicit form
\begin{equation}
\label{gammaexplicit}
\gamma= d\theta^2+\sin^2\theta d\phi^2.
\end{equation}
The above metric \eqref{sskruskal} can be extended to define the maximally-extended Reissner-Nordstr{\"o}m solution on the ambient manifold $(-\infty, \infty) \times (\infty, \infty) \times S^2$. In this paper, we will only consider the manifold-with-boundary $\MM$, corresponding to the exterior of the spacetime.

 The Reissner-Nordstr{\"o}m family
  of spacetimes $(\mathcal{M},\g_{M, Q})$ is the unique electrovacuum spherically symmetric spacetime. It is a static and asymptotically flat spacetime. The parameter $Q$ may be interpreted as the charge of the source. This metric clearly reduces to Schwarzschild spacetime when $Q=0$, therefore $M$ can be interpreted as the mass of the source.

Using definition \eqref{sskruskal}, the metric $\g_{M, Q}$ is manifestly smooth in the whole domain. We will now describe different sets of coordinates for which smoothness breaks down, but which are nevertheless useful for computations.

\subsubsection{Double null coordinates $u$, $v$}
\label{EFdnulldef}

We define another double null coordinate system that covers  the interior of $\mathcal{M}$, modulo the degeneration of the 
angular coordinates. This coordinate system, 
$\left(u,v,\theta^1, \theta^2\right)$, is called  \emph{double null coordinates} and are defined via the relations
\begin{align} \label{UuVv}
U = -\frac{2r_{+}^2}{r_{+}-r_{-}}\exp\left(-\frac{r_{+}-r_{-}}{4r_{+}^2} u\right) \ \ \ \textrm{and} \ \ \ \ V = \frac{2r_{+}^2}{r_{+}-r_{-}}\exp \left(\frac{r_{+}-r_{-}}{4r_{+}^2} v\right) \, .
\end{align}
Using (\ref{UuVv}), we obtain the Reissner-Nordstr{\"o}m metric on the interior of $\MM$ in $\left(u,v,\theta^1, \theta^2\right)$-coordinates:
\begin{align} \label{ssef}
\g_{M, Q} =  - 4 \Up \left(u,v\right) \, d{u} \, d{v} +   r^2 \left(u,v\right) \gamma_{AB} d{\theta}^A d{\theta}^B  
\end{align}
with
\begin{align}
\label{officialOmegadef}
\Up:= 1-\frac{2M}{r}+\frac{Q^2}{r^2}
\end{align}
and the function $r: \left(-\infty,\infty\right) \times \left(-\infty,\infty\right) \rightarrow \left(M+\sqrt{M^2-Q^2},\infty\right)$ defined implicitly via the relations between $(U,V)$ and $(u,v)$. In $\left(u,v,\theta^1, \theta^2\right)$-coordinates, the horizon $\mathcal{H}^+$ can still be formally parametrised by $\left(\infty, v,\theta^2,\theta^2\right)$ with $v \in \mathbb{R}$, $\left(\theta^1,\theta^2\right) \in S^2$.

Note that $u, v$ are  regular optical functions. Their corresponding   null geodesic generators are
\bea\label{definition-L-Lb}
\Lb:=-g^{ab}\pr_a v  \pr_b=\frac{1}{\Up} \pr_u,         \qquad  L:=-g^{ab}\pr_a u  \pr_b=\frac{1}{\Up} \pr_v,
\eea
They verify
\beaa
g(L, L)=g(\Lb,\Lb)=0, \quad  g(L, \Lb)=-2\Up^{-1},\qquad D_L L=D_\Lb \Lb=0.
\eeaa

\subsubsection{Standard coordinates $t$, $r$}
Recall the form of the metric \eqref{ssef} in double null coordinates. Setting $$t=u+v$$ we may rewrite the above metric in coordinates $(t, r, \th, \phi)$ in the usual form \eqref{RNintro}:
\bea\label{RNintro-1}
 \g_{M, Q}=-\Up(r) dt^2 +\Up(r)^{-1}dr^2 +r^2(d\theta^2+\sin^2\theta d\phi^2),
\eea
which covers the interior of $\MM$. Observe that  
\beaa
\Up(r):=1-\frac{2M}{r}+\frac{Q^2}{r^2}=\frac{(r-r_{-})(r-r_{+})}{r^2}
\eeaa
where $r_{-}$ and $r_+$ are defined in \eqref{definiion-rpm}.

The photon sphere of Reissner-Nordstrom corresponds to the hypersurface in which null geodesics are trapped. It is given by  
\bea\label{Up'-photonsphere}
\Up'(r)-\frac{2}{r} \Up(r)=\frac{1}{r^3}(r^2-3M r+2Q^2)=0.
\eea
 In particular, the photon sphere is realized at the hypersurface given by $\{ r=r_P \}$ where $r_P$ is given by 
\bea\label{def-rP}
r_P=\frac{3M+\sqrt{9M^2-8Q^2}}{2}
\eea

The null vectors $L$ and $\Lb$ defined in \eqref{definition-L-Lb}, in $(t, r)$ coordinates can be written as
\beaa
\Lb=\Up^{-1}\pr_t -\pr_r,         \qquad  L=\Up^{-1}\pr_t +\pr_r,
\eeaa

\subsubsection{ Ingoing  Eddington-Finkelstein coordinates $v$, $r$}  
We define another coordinate system that covers the interior of $\MM$. This coordinate system, $(v, r, \th^1, \th^2)$ is called \emph{ingoing  Eddington-Finkelstein coordinates} and makes use of the above defined functions $v$ and $r$.  The Reissner-Nordstr{\"o}m metric on the interior of $\MM$ in $\left(v,r, \th, \phi\right)$-coordinates is given by 
\bea
\label{Schw:EF-coordinates}
\g_{M,Q}=-\Up(r) dv^2 + 2 dv dr + r^2(d\theta^2+\sin^2\theta d\phi^2). 
\eea

\subsubsection{Outgoing Eddington-Finkelstein coordinates $u$, $r$}
We define another coordinate system that covers the interior of $\MM$. This coordinate system, $(u, r, \th^1, \th^2)$ is called \emph{outgoing  Eddington-Finkelstein coordinates} and makes use of the above defined functions $u$ and $r$.  The Reissner-Nordstr{\"o}m metric on the interior of $\MM$ in $\left(u,r, \th, \phi\right)$-coordinates is given by 
\bea
\label{Schw:EF-coordinates-out}
\g_{M,Q}=-\Up(r) du^2 - 2 du  dr + r^2(d\theta^2+\sin^2\theta d\phi^2). 
\eea

Observe that the Reissner-Nordstr{\"o}m metric is foliated by spheres $S$, as appears from any coordinate system written above. The spheres $S$ can be obtained as intersections of levels of coordinate functions.

\subsection{Null frames: Ricci coefficients and curvature components }\label{null-frames-ricci}

We define in this section three normalized null frames associated to Reissner-Nordstr{\"o}m. 
We  can use the null geodesic generators $L, \Lb$ defined in \eqref{definition-L-Lb}  to define the following          canonical null pairs.
\begin{enumerate}

\item The          null frame $(e_3^*, e_4^*) $        for which $e_3$ is geodesic  (which is regular towards the future along the event horizon) is given by
\bea
\label{eq:regular-nullpair}
e_3^*=\Lb,\,\, \qquad e_4^*=\Up L.
\eea
All  Ricci  coefficients  vanish  except,
\beaa
\ka=\frac{2\Up}{r}=\frac 2 r \left(1-\frac{2M}{r}+\frac{Q^2}{r^2} \right),\,\,\quad  \kab=-\frac{2}{r},\,\,\quad  \om=-\frac{M}{r^2} + \frac{Q^2}{r^3},\,\, \quad \omb=0.
\eeaa
\item The null frame $(e_3, e_4)$  for which $e_4$ is geodesic is given by     
\bea\label{outgoing-null-pair}
e_3=\Up \Lb,\,\, \qquad  e_4 =L
\eea
All  Ricci  coefficients  vanish except,
\beaa
\ka=\frac{2}{r},\qquad \kab=-\frac{2\Up}{r}=-\frac 2 r \left(1-\frac{2M}{r}+\frac{Q^2}{r^2} \right), \qquad \om=0, \qquad \omb =\frac{M}{r^2} -\frac{Q^2}{r^3}
\eeaa
\item The symmetric null frame $(e_3^s, e_4^s)$ is given by 
\bea\label{symmetric-null-pair}
e_4^s =\Omega  L , \qquad e_3^s=\Omega\Lb .
\eea
where $\Omega=\sqrt{\Up}$. 
In this case,
\beaa
\ka=-\kab= \frac{2\Omega}{r}, \qquad \om=-\omb =-\frac{M}{2\Omega r^2} +\frac{Q^2}{2\Omega r^3}  
\eeaa
\end{enumerate}

Observe that all null frames above defined, \eqref{eq:regular-nullpair}, \eqref{outgoing-null-pair} and \eqref{symmetric-null-pair}, verify conditions \eqref{relations-null-frame} of local null frames. 
They verify the relations:
\bea\label{relations-between-frames}
e_3=\Up e_3^*, \qquad e_4=\Up^{-1}e_4^*
\eea
We denote $\mathscr{N}=\{ e_3, e_4, e_A \}$ the null frame with $e_4$ geodesic, and $\mathscr{N}^*=\{ e_3^*, e_4^*, e_A \}$ the null frame which is regular towards the horizon.  In particular, $\mathscr{N}^*$ extends regularly to a non-vanishing null frame on $\mathcal{H}^+$.

The curvature and electromagnetic components which are non-vanishing do not depend on the particular null frame. They are given by 
\beaa
\rhoF=\frac{Q}{r^2}, \qquad \rho  =  -\frac{2M}{r^3}+\frac{2Q^2}{r^4} 
\eeaa
 We also have that
\begin{align} \label{GCdef}
K = \frac{1}{r^2}
\end{align}
for the Gauss curvature of the round $S^2$-spheres.

\subsubsection{Foliation $\Sigma_\tau$}\label{foliation}
For all values $t \in \mathbb{R}$, the hypersurfaces $\widetilde{\Sigma}_\tau=\{ t= \tau \}$  are spacelike. For polynomial decay following the method of \cite{rp} and \cite{mos}, we will require hypersurfaces $\Sigma_\tau$ which connect the event horizon and null infinity. We define such a foliation in the following way. 

 Recall the definitions of $r_{\mathcal{H}}$ and $r_P$ given by \eqref{def-rH} and \eqref{def-rP}.
We divide the exterior of Reissner-Nordstr{\"o}m spacetime  $\MM$  in   the following  regions:
             \begin{enumerate}
             \item The  red shift region   $\MM_{red}:=\{ r_{\mathcal{H}} \leq r \leq \frac{11}{10}r_{\mathcal{H}} \}$            
             \item The trapping region  $\MM_{trap} :=\{     \frac{5}{6} r_P\le r \le  \frac{7}{6}r_P\}$
            
             \item  The far-away region $\MM_{far}:=\{  r  \ge R_0 \}$ with  $R_0$  a fixed  number  $R_0\gg 2r_P$.             
             \end{enumerate}

   For fixed $R$ we denote by $\MM_{\le R}$ and  $\MM_{\ge R} $ the   regions  defined by $r\le R$ and $r\ge R$.               
             
                We foliate  $\MM$    by  hypersurfaces  $\Si_\tau $ which are:
             \begin{enumerate}
             \item  Incoming null    in $\MM_{red}$,        
              with $e_3^*$ as null incoming   generator (which is regular up to horizon). We denote  this portion 
              $\Si_{red}$. This is realized by a portion of $\{ v=\text{const}\}$ in the ingoing Eddington-Finkelstein coordinates.
             \item   Strictly spacelike  in $\MM_{trap} $. We denote  this  portion   by $\Si_{trap}$. This is realized by a portion of $\widetilde{\Sigma}_\tau$. 
                          \item     Outgoing null    in  $\MM_{far}$
             with $e_4$ as null  outgoing generator.   We denote this portion by $\Si_{far}$. This is realized by a portion of $\{ u= \text{const}\}$ in the outgoing Eddington-Finkelstein coordinates. 
                      \end{enumerate}

 We denote $\MM(\tau_1, \tau_2)\subset \MM$  the spacetime region   in  the past of $\Si(\tau_2)$ and in the future of $\Si(\tau_1)$. We also denote 
 \beaa
 \mathcal{H}^+(\tau_1, \tau_2)=\MM(\tau_1, \tau_2) \cap \mathcal{H}^+, \qquad  \mathcal{I}^+(\tau_1, \tau_2)=\MM(\tau_1, \tau_2) \cap \mathcal{I}^+
 \eeaa

\begin{figure}[h]
\centering
\def\svgwidth{0.35\textwidth} 
\includegraphics{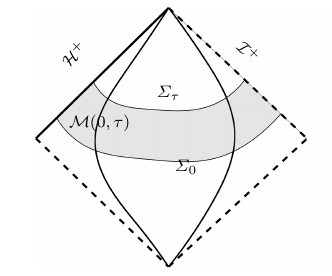}
\caption{Foliation $\Sigma_\tau$ in the Penrose diagram of Reissner-Nordstr{\"o}m spacetime}
\end{figure}

\subsection{Killing fields of the Reissner-Nordstr{\"o}m metric}

We discuss now the Killing fields associated to the metric $\g_{M, Q}$. Notice that the Reissner-Nordstr{\"o}m metric possesses the same symmetries as the ones possessed by Schwarzschild spacetime.

We define the vectorfield $T$ to be the timelike Killing vector field $\pr_t$ of the $(t, r)$ coordinates in \eqref{RNintro-1}, which in double null coordinates is given by 
\beaa
T=\frac 1 2 (\pr_u+\pr_v)
\eeaa
 The vector field extends to a smooth Killing field on the horizon $\mathcal{H}^+$, which is moreover null and tangential to the null generator of $\mathcal{H}^+$. 

In terms of the null frames defined in Section \ref{null-frames-ricci}, the Killing vector field $T$ can be written as 
\bea\label{definition-T}
T=\frac 12 (\Up e_3^*+e_4^*)=\frac 1 2 (e_3+\Up e_4)
\eea
Notice that at on the horizon, $T$ corresponds up to a factor with the null vector of $\mathscr{N}^*$ frame, $T=\frac 1 2 e_4^*$. 

We can also define a basis of angular momentum operator $\Omega_i$, $i=1,2,3$. Fixing standard spherical coordinates on $S^2$, we have 
\beaa
\Omega_1=\partial_\phi, \qquad \Omega_2=-\sin \phi \partial_\th-\cot \th \cos \phi \partial_\phi, \qquad \Omega_3=\cos \phi \partial_\th-\cot \th \sin \phi \partial_\phi
\eeaa
 The Lie algebra of Killing vector fields of $\g_{M,Q}$ is then generated by $T$ and $\Omega_i$, for $i=1,2,3$.

\subsection{Reissner-Nordstr{\"o}m background operators and commutation identities} \label{sec:commutation}
In this section, we specialize the operators discussed in Section \ref{tensor-algebra} to the Reissner-Nordstr{\"o}m metric.

Adapting the commutation formulae \eqref{commutators} to the Reissner-Nordstr{\"o}m metric, we obtain the following commutation formulae. For projected covariant derivatives for $\xi = \xi_{A_1...A_n}$ any $n$-covariant $S^2_{u,s}$-tensor in Reissner-Nordstr{\"o}m metric $\left(\mathcal{M},\g_{M,Q}\right)$ we have
\begin{align}\label{commutation-formulas}
\nabb_3 \slashed{\nabla}_B \xi_{A_1...A_n} - \slashed{\nabla}_B \slashed{\nabla}_3 \xi_{A_1...A_n} &= - \frac{1}{2} \kab  \slashed{\nabla}_B \xi_{A_1...A_n} \, , \nonumber \\
\slashed{\nabla}_4 \slashed{\nabla}_B \xi_{A_1...A_n} - \slashed{\nabla}_B \slashed{\nabla}_4 \xi_{A_1...A_n} &= - \frac{1}{2} \ka  \slashed{\nabla}_B \xi_{A_1...A_n} \, , \\
\slashed{\nabla}_3 \slashed{\nabla}_4 \xi_{A_1...A_n} - \slashed{\nabla}_4 \slashed{\nabla}_3 \xi_{A_1...A_n} &=-2\om \slashed{\nabla}_3 \xi_{A_1...A_n}+2\omb \slashed{\nabla}_4 \xi_{A_1...A_n} \, . \nonumber
\end{align}
In particular, we have
\begin{align}\label{commutator-rnabla}
\left[\slashed{\nabla}_4, r \slashed{\nabla}_A \right] \xi = 0 \ \ \ , \ \ \ \left[\slashed{\nabla}_3, r \slashed{\nabla}_A \right] \xi = 0  \, .
\end{align}
We summarize here the commutation formulae for the angular operators defined in Section \ref{tensor-algebra}. Let $\rho, \sigma$ be scalar functions, $\xi$ be a $1$-tensor and $\th$ be a symmetric traceless $2$-tensor in Reissner-Nordstr{\"o}m manifold. Then:
\bea
\left[ \nabb_4, \DDd_1\right]\xi &=&-\frac 1 2 \ka \DDd_1\xi, \qquad \left[ \nabb_3, \DDd_1\right]\xi =-\frac 1 2 \kab \DDd_1\xi\label{commutator-nabb-4-divv} \\
\left[ \nabb_4, \DDs_1\right](\rho,\sigma)&=&-\frac 1 2 \ka\DDs_1(\rho,\sigma), \qquad \left[ \nabb_3, \DDs_1\right](\rho,\sigma)=-\frac 1 2 \kab\DDs_1(\rho,\sigma)\label{commutator-nabb-4-DDs-1}, \\
\left[ \nabb_4, \DDd_2\right]\th &=&-\frac 1 2 \ka \DDd_2\th, \qquad \left[ \nabb_3, \DDd_2\right]\th =-\frac 1 2 \kab \DDd_2\th\label{commutator-nabb-4-DDd_2} \\
\left[ \nabb_4, \DDs_2\right]\xi&=&-\frac 1 2 \ka\DDs_2\xi, \qquad
\left[ \nabb_3, \DDs_2\right]\xi=-\frac 1 2 \kab\DDs_2\xi\label{commutator-nabb-4-DDs}
\eea

\section{The linearized Einstein-Maxwell equations}\label{linearized-equations}
We collect here the equations for linearized gravitational and electromagnetic perturbation of Reissner-Nordstr{\"o}m metric. 
Recall that in Reissner-Nordstr{\"o}m metric the following Ricci coefficients, curvature and electromagnetic components vanish:
\beaa
&\chih, \quad \chibh, \quad \eta, \quad \etab, \quad \ze, \quad \xi, \quad \xib \\
&\a, \quad \b, \quad \sigma, \quad \bb, \quad \aa \\
&\bF, \quad  \sigmaF, \quad \bbF
\eeaa
In particular, in writing the linearization of the equations of Section \ref{VEeqDNGsec}, we will neglet the quadratic terms, i.e. product of terms above which vanish in Reissner-Nordstr{\"o}m background. 

We remark that these equations are linear once one understands the quantities of order $\ep$ to be the unknowns, and the quantities of order 1 to be the known Reissner-Nordstr\"om values upon a gauge choice. 

\subsection{Linearised null structure equations}
The linearization of equations \eqref{nabb-3-chibh-general}-\eqref{Gauss-general} are the following:
\bea
\nabb_3 \chibh+\left( \kab +2\omb\right) \chibh&=&-2\DDs_2\xib -\aa, \label{nabb-3-chibh} \\
\nabb_4 \chih+\left(\ka+2\om\right) \ \chih&=&-2\DDs_2\xi -\a , \label{nabb-4-chih}\\
\nabb_3\chih+\left(\frac 12  \kab -2 \omb\right) \chih   &=&  -2 \slashed{\mathcal{D}}_2^\star \eta-\frac 1 2 \ka \chibh   \label{nabb-3-chih}\\
\nabb_4\chibh+\left(\frac 12  \ka-2\om\right) \,\chibh &=&  -2 \slashed{\mathcal{D}}_2^\star \etab -\frac 1 2 \kab \chih \label{nabb-4-chibh}, 
\eea
\bea
\nabb_3 \ze+ \left(\frac 1 2 \kab-2\omb \right)\ze&=&2 \DDs_1( \omb, 0)-\left(\frac 1 2 \kab+2\omb \right)\eta +\left(\frac 1 2 \ka +2\om \right)\xib - \bb  -\rhoF\bbF, \label{nabb-3-ze} \\
\nabb_4 \ze+ \left(\frac 1 2 \ka-2\omb \right)\ze&=&-2 \DDs_1( \omb, 0)+\left(\frac 1 2 \ka+2\om \right)\etab -\left(\frac 1 2 \kab +2\omb \right)\xi - \b  -\rhoF\bF, \label{nabb-4-ze} \\
\nabb_4\xib-\nabb_3 \etab&=&-\frac 1 2 \kab\left(\eta-\etab \right)+4\om  \xib -\bb-\rhoF \bbF , \label{nabb-4-xib}\\
\nabb_3\xi-\nabb_4 \eta&=&\frac 1 2 \ka\left(\eta-\etab \right)+4\omb  \xi +\b+\rhoF \bF , \label{nabb-3-xi}
\eea
\bea
\nabb_3 \kab +\frac 1 2 \kab^2+2\omb \ \kab&=&2\divv \xib, \label{nabb-3-kab}\\
\nabb_4 \ka +\frac 1 2 \ka^2+2\om \ \ka&=&2\divv \xi, \label{nabb-4-ka}\\
\nabb_3 \ka+\frac 1 2 \ka\kab-2\omb \ka&=& 2\divv \eta +2\rho, \label{nabb-3-ka}\\
\nabb_4 \kab+\frac 1 2 \ka\kab-2\om \kab&=& 2\divv \etab +2\rho, \label{nabb-4-kab}
\eea
\bea
\divv \chibh&=&-\frac 1 2 \kab \ze-\frac 1 2 \DDs_1( \kab, 0) +\bb-\rhoF\bbF, \label{Codazzi-chib} \\
\divv \chih&=&\frac 1 2 \ka \ze-\frac 1 2 \DDs_1(\ka, 0) -\b +\rhoF\bF \label{Codazzi-chi}
\eea
\bea
\nabb_4\omb+\nabb_3\om&=& 4\om\omb+\rho+\rhoF^2,\label{nabb-4-omb} \\
\curll \eta&=&\sigma, \label{curl-eta}\\
\curll \etab&=& -\sigma \label{curl-etab}
\eea
\bea
K&=&-\frac 1 4 \ka\kab-\rho+\rhoF^2 \label{Gauss}
\eea

\subsection{Linearised Maxwell equations}
The linearisation of equations \eqref{Maxwell-1}-\eqref{nabb-4-rhoF-general} are the following:
\bea
\nabb_3 \bF+\left(\frac 1 2 \kab-2\omb\right) \bF&=& -\DDs_1(\rhoF, \sigmaF) +2\rhoF \eta, \label{nabb-3-bF}\\
\nabb_4 \bbF+\left(\frac 1 2 \ka-2\om\right) \bbF&=&\DDs_1(\rhoF, -\sigmaF)-2\rhoF \etab \label{nabb-4-bbF}
\eea
\bea
\nabb_3\rhoF+ \kab \rhoF&=& -\divv\bbF \label{nabb-3-rhoF}\\
\nabb_4\rhoF+ \ka \rhoF&=&\divv\bF\label{nabb-4-rhoF}
\eea
\bea
\nabb_3\sigmaF+\kab \ \sigmaF&=& -\slashed{\curl}\bbF\label{nabb-3-sigmaF}\\
\nabb_4\sigmaF+\ka \ \sigmaF&=& \slashed{\curl}\bF\label{nabb-4-sigmaF}
\eea

\subsection{Linearised Bianchi identities}
The linearisation of the equations \eqref{Bianchi-non-linear1}-\eqref{nabb-4-sigma-general} are the following:
\bea
 \nabb_3\a+\left(\frac 1 2 \kab-4\omb \right)\a&=&-2 \DDs_2\, \b-3\rho\chih -2\rhoF \ \left(\DDs_2\bF +\rhoF\chih \right), \label{nabb-3-a}\\
 \nabb_4\aa+\left(\frac 1 2 \ka-4\om\right) \aa&=&2 \DDs_2\, \bb-3\rho\chibh +2 \rhoF \ \left(\DDs_2\bbF -\rhoF\chibh \right) , \label{nabb-4-aa}
 \eea
 \bea
 \nabb_3 \b+\left(\kab -2\omb\right) \b &=&\DDs_1(-\rho, \sigma) +3\rho  \eta +\rhoF\left(\DDs_1(-\rhoF, \sigmaF)  - \ka\bbF-\frac 1 2 \kab \bF    \right) , \label{nabb-3-b}\\
 \nabb_4 \bb+\left(\ka-2\om\right) \bb &=&\DDs_1(\rho, \sigma) -3 \rho \etab +\rhoF\left(\DDs_1(\rhoF, \sigmaF)  - \kab\bF-\frac 1 2 \ka \bbF    \right) , \label{nabb-4-bb}
 \eea
 \bea
 \nabb_3 \bb+ \left( 2\kab  +2\omb\right) \bb &=&-\divv\aa-3\rho \xib  +\rhoF\left(\nabb_3\bbF+2\omb \bbF+2\rhoF \ \xib\right), \label{nabb-3-bb}\\
 \nabb_4 \b+ \left(2\ka+2\om\right) \b &=&\divv\a +3\rho\xi +\rhoF\ \left(\nabb_4\bF+2\om \bF-2\rhoF \xi\right) \label{nabb-4-b}
 \eea
\bea
\nabb_3\rho+\frac 3 2 \kab \rho&=&-\kab\rhoF^2-\divv\bb-\rhoF \ \divv\bbF \label{nabb-3-rho}, \\
\nabb_4\rho+\frac 3 2 \ka \rho&=&-\ka\rhoF^2 +\divv\b+\rhoF \ \divv\bF, \label{nabb-4-rho}
\eea
\bea
\nabb_3 \sigma+\frac 3 2 \kab \sigma&=&-\slashed{\curl}\bb - \rhoF \ \slashed{\curl}\bbF, \label{nabb-3-sigma}\\
\nabb_4 \sigma+\frac 3 2 \ka \sigma&=&-\slashed{\curl}\b - \rhoF  \ \slashed{\curl}\bF \label{nabb-4-sigma}
\eea

\section{Gauge-invariant quantities}\label{section-gauge-invariant}
The Teukolsky equations we shall consider are wave equations for gauge-invariant quantities for linear gravitational and electromagnetic perturbations of Reissner-Nordstr{\"o}m spacetime. In the context of spin $\pm2$ Teukolsky-type equations, we will consider  in particular two quantities which are gauge-invariant: $\a$ and  $\ff$ for spin $+2$, and $\aa$ and $\underline{\ff}$ for spin $-2$. 

In order to identify the gauge-invariant quantities we consider null frame transformations, i.e. linear transformations which take null frames into null frames.

We recall here a generalisation of Lemma 2.3.1. of \cite{stabilitySchwarzschild}.
\begin{lemma}\label{null-frame-transformations}[Lemma 2.3.1. in \cite{stabilitySchwarzschild}] A linear null transformation can be written in the form 
\beaa
e_4'&=&e^a \left(e_4+f^A e_A \right), \\
e_3'&=& e^{-a} \left(e_3+\underline{f}^A e_A \right), \\
e_A'&=& {O_A}^B e_B+\frac 1 2 \underline{f}_A e_4+\frac 1 2 f_A e_3
\eeaa
where $a$ is a scalar function, $f$ and $\underline{f}$ are $S_{u,s}$-tensors and ${O_A}^B$ is an orthogonal transformation of $(S_{u,s}, \slashed{g})$, i.e. ${O_A}^B \slashed{g}_{BC}=\slashed{g}_{AC}$.
\end{lemma}
Observe that the identity transformation is given by $a=f_A=\underline{f}_A=0$ and ${O_A}^B=\de_A^B$. Therefore, a linear perturbation of a null frame is one for which $a=f_A=\underline{f}_A=O(\ep)$ and ${O_A}^B=\de_A^B+O(\ep)$.

We recall here Proposition 2.3.4 of \cite{stabilitySchwarzschild}. We write the transformations for some of the Ricci coefficients and curvature components under a general null transformation of this type.

\begin{proposition}\label{prop:transformations1}[Proposition 2.3.4 of \cite{stabilitySchwarzschild}]
Under a general transformation of the type given in Lemma \ref{null-frame-transformations}, the linear transformations of Ricci coefficients and curvature components are the following: 
\bea
     \chih'&=&\chih -\DDs_2 f, \qquad  \chibh'=\chibh-\DDs_2 \fb  \label{chiprime}
   \eea
\bea
\a'&=&\a, \qquad \aa'= \aa  \label{aprime} \label{aaprime}\\
\b'&=&\b+\frac 3 2 \rho f , \qquad \bb'=\bb-\frac 3 2 \rho \fb \\
\rho'&=& \rho
\eea
 \bea
 \bF'&=&\bF +f \rhoF,  \qquad  \bbF'=\bbF - \fb \rhoF \label{bbfprime}\label{bfprime} \\
 \rhoF'&=&\rhoF.
 \eea
\end{proposition}

The extreme curvature components $\a$ and $\aa$ are gauge-invariant, as implied by Proposition \ref{prop:transformations1}, equation \eqref{aprime}. 

These curvature components are extremely important because in the case of the Einstein vacuum equation they verify a decoupled wave equation, the celebrated Teukolsky equation, first discovered in the Schwarzschild case in \cite{Bardeen} and generalized to the Kerr case in \cite{Teukolsky}.

In the case of Einstein-Maxwell equation, the presence of electromagnetic perturbations modifies the decoupling of the Teukolsky equations for $\a$ and $\aa$. Indeed, they do not verify a wave equation which is decoupled by all other curvature components, but instead a Teukolsky-type equation with a non-trivial right hand side. 

By Proposition \ref{prop:transformations1}, equation \eqref{bfprime}, the extreme electromagnetic components $\bF$ and $\bbF$ are not gauge-invariant if $\rhoF$ is not zero in the background. In the case of the Maxwell equations in Schwarzschild, the components $\bF$ and $\bbF$ are gauge-invariant, and satisfy a spin $\pm1$ Teukolsky equation. In \cite{Federico}, the author proves a boundedness and decay statement for solutions $\bF$ and $\bbF$ of the spin $\pm1$ Teukolsky equation. In the case of coupled gravitational and electromagnetic perturbations of Reissner-Nordstr{\"o}m spacetime, the spin $\pm1$ Teukolsky equations verified by $\bF$ and $\bbF$ (derived in Proposition \ref{Teukolsky-bF}) cannot be used, since they are not gauge-invariant.
It turns out that we will make use of different gauge-invariant quantities related to the electromagnetic components $\bF$ and $\bbF$. 

We define the following symmetric traceless $2$-tensors
\beaa
 \ff:=\DDs_2 \bF+\rhoF \chih, \qquad \underline{\ff}:=\DDs_2 \bbF-\rhoF \chibh
 \eeaa
 
 We summarize in the following lemma a fundamental property of the 2-tensors $\ff$ and $\underline{\ff}$.
 
 \begin{lemma}\label{ff-is-gaugeinv}
  The tensors $\ff$ and $\underline{\ff}$ are invariant to any linear null transformation. 
 \end{lemma}
 \begin{proof} Consider any linear null transformation described by $f$. Using \eqref{chiprime} and \eqref{bfprime}, we obtain
 \beaa
 \ff'&=&\DDs_2\bF'+\rhoF' \chih'=\DDs_2(\bF+f\rhoF )+\rhoF (\chih-\DDs_2 f)\\
 &=&\DDs_2\bF+\rhoF' \chih+\rhoF \DDs_2 f-\rhoF \DDs_2 f=\ff
 \eeaa
 and similarly for $\ffb$. 
 \end{proof}
 According to Lemma \ref{ff-is-gaugeinv}, we say that $\ff$ and $\underline{\ff}$ are gauge-invariant.  
 
 Notice that the quantities $\ff$ and $\ffb$ appear in the Bianchi identities for $\a$ and $\aa$. The equations \eqref{nabb-3-a} and \eqref{nabb-4-aa} can therefore be rewritten as 
 \bea
  \nabb_3\a+\left(\frac 1 2 \kab-4\omb \right)\a&=&-2 \DDs_2\, \b-3\rho\chih -2\rhoF \ \ff, \label{nabb-3-a-ff}\\
 \nabb_4\aa+\left(\frac 1 2 \ka-4\om\right) \aa&=&2 \DDs_2\, \bb-3\rho\chibh +2 \rhoF \ \ffb\label{nabb-4-aa-ff} 
 \eea
 By this rewriting of the Bianchi identities \eqref{nabb-3-a} and \eqref{nabb-4-aa} it is clear why $\ff$ and $\ffb$ shall appear on the right hand side of the wave equation verified by $\a$ and $\aa$. 

The quantities $\ff$ and $\ffb$ verify themselves Teukolsky-type equations, which are coupled with $\a$ and $\aa$ respectively. The equations for $\a$ and $\ff$ and for $\aa$ and $\ffb$ constitute the generalized spin $\pm 2$ Teukolsky system obtained in Section \ref{spin2}.

\section{Generalized spin $\pm2$-Teukolsky and Regge-Wheeler systems}\label{spin2}

In this section, we will introduce a generalization of the celebrated spin $\pm 2$ Teukolsky equations and the Regge-Wheeler equation, and explain the connection between them.

\subsection{Generalized spin $\pm2$ Teukolsky system}

The generalized spin $\pm2$ Teukolsky system concern symmetric traceless $2$-tensors in Reissner-Nordstr{\"o}m spacetime, which we denote $(\a, \ff)$ and $(\aa, \underline{\ff})$ respectively.

\begin{definition} Let $\a$ and $\ff$ be two smooth symmetric traceless $S^2_{u,s}$ 2-tensors defined
on a subset $\mathcal{D}\subset \mathcal{M}$. We say that $(\alpha, \ff)$ satisfy the {\bf generalized Teukolsky system of spin ${\bf +2}$} if they satisfy the following coupled system of PDEs:
\beaa
\Box_\g \a&=& -4\omb \nabb_4\a+2\left( \ka+2\om\right) \nabb_3\a+\left( \frac 1 2 \ka\kab-4\rho +4\rhoF^2+2\om\,\kab-10\omb\ka-8\om\omb -4\nabb_4\omb  \right)\,\a\\
&& +4\rhoF\left(\nabb_4\ff+ \left( \ka  +2\om \right) \ff\right), \\
\Box_\g (r\ff)&=& -2\omb \nabb_4(r\ff)+\left( \ka+2\om\right) \nabb_3(r\ff)+\left(-\frac 1 2 \ka\kab-3\rho+\om\kab-3\omb\ka-2\nabb_4\omb\right) r\ff \\
&&-r\rhoF \left(\nabb_3\a+\left(\kab-4\omb\right)\a\right)
\eeaa
where $\Box_\g=\Box_{\g_{M, Q}}=\g^{\mu\nu} \D_{\mu}\D_{\nu}$ denotes the wave operator in Reissner-Nordstr{\"o}m spacetime, and the coefficients (in terms of $\ka$, $\kab$, $\om$, $\omb$, $\rho$, $\rhoF$) are the corresponding background values in Reissner-Nordstr\"om\footnote{The generalized Teukolsky system as defined here is gauge-invariant. In particular, the coefficients of the equations can be evaluated in any gauge.}. 

Let $\aa$ and $\underline{\ff}$ be two smooth symmetric traceless $S^2_{u,s}$ 2-tensor defined
on a subset $\mathcal{D}\subset \mathcal{M}$. We say that $(\aa, \underline{\ff})$ satisfy the {\bf generalized Teukolsky system of spin ${\bf -2}$} if they satisfy
\beaa
\Box_\g \aa&=& -4\om \nabb_3\aa+2\left( \kab+2\omb\right) \nabb_4\aa+\left( \frac 1 2 \ka\kab-4\rho +4\rhoF^2+2\omb\,\ka-10\om\kab-8\om\omb -4\nabb_3\om  \right)\,\aa\\
&& -4\rhoF\left(\nabb_3\underline{\ff}+ \left( \kab  +2\omb \right) \underline{\ff}\right), \\
\Box_\g (r \underline{\ff})&=& -2\om \nabb_3(r\underline{\ff})+\left( \kab+2\omb\right) \nabb_4(r\underline{\ff})+\left(-\frac 1 2 \ka\kab-3\rho+\omb\ka-3\om\kab-2\nabb_3\om\right) r\underline{\ff} \\
&&+r\rhoF \left(\nabb_4\aa+\left(\ka-4\om\right)\aa\right)
\eeaa
\end{definition}

\begin{remark} When the electromagnetic tensor vanishes, i.e. if $\bF=\bbF=\rhoF=\sigmaF=0$, the generalized Teukolsky system of spin $\pm2$ reduces to the first equation, since $\ff=0$. Moreover, the first equation reduces to the Teukolsky equation of spin $\pm2$ in Schwarzschild.
\end{remark}

Given a linear gravitational and electromagnetic perturbations of Reissner-Nordstr{\"o}m spacetime, the quantities $(\a, \ff)$ and $(\aa, \underline{\ff})$ given by the gauge-invariant quantities defined in Section \ref{section-gauge-invariant} satisfy the generalized Teukolsky system of spin $\pm 2$ respectively. The equations are implied by the linearized Einstein-Maxwell equations in Section \ref{linearized-equations}, as showed in Proposition \ref{squarea} and Proposition \ref{squareff} in Appendix A.

\subsection{Generalized Regge-Wheeler system}

The other generalized system to be defined here is the generalized Regge-Wheeler system, to be satisfied again by symmetric traceless $S_{u,s}$ tensors $(\qf, \qf^\F)$.

\begin{definition} \label{def:rwe}
Let $\qf$ and $\qf^\F$ be two smooth symmetric traceless $S^2_{u,s}$ $2$-tensor on $\mathcal{D}$.
We say that $(\qf, \qf^\F)$ satisfy the {\bf generalized Regge--Wheeler system for spin $+2$} if they satisfy the following coupled system of PDEs:
\bea\label{finalsystem}
\begin{split}
\Box_\g\qf+\left( \ka\kab-10\rhoF^2\right)\ \qf&= \rhoF\Bigg(4r\lapp_2\qf^{\F}-4r\kab \nabb_4(\qf^{\F})-4r\ka\nabb_3( \qf^\F) + r\left(6\ka\kab+16\rho+8\rhoF^2\right)\qf^{\F}\Bigg)\\
&+\rhoF(l.o.t.)_1, \\
\Box_\g\qf^{\F}+\left( \ka\kab+3\rho\right)\ \qf^{\F}&=\rhoF \Bigg(-\frac {1}{ r} \qf\Bigg) +\rhoF^2 (l.o.t.)_2
\end{split}
\eea
where $(l.o.t.)_1$ and $(l.o.t.)_2$ are lower order terms with respect to $\qf$ and $\qf^\F$. Schematically\footnote{The exact form of the equations is obtained in Proposition \ref{wave-qfF} and Proposition \ref{wave-qf}.} $\nabb_3^{\leq 2}(l.o.t)=\qf$.

Let $\underline{\qf}$ and $\underline{\qf}^\F$ be two smooth symmetric traceless $S^2_{u,s}$ $2$-tensor on $\mathcal{D}$.
We say that $(\underline{\qf}, \underline{\qf}^\F)$ satisfy the {\bf generalized Regge--Wheeler system for spin $-2$} if they satisfy the following coupled system of PDEs:
\small
\bea\label{finalsystem-spin-2}
\begin{split}
\Box_\g \underline{\qf}+\left( \ka\kab-10\rhoF^2\right)\ \underline{\qf}&= -\rhoF\Bigg(4r\lapp_2\underline{\qf}^{\F}-4r\ka \nabb_3(\underline{\qf}^{\F})-4r\kab\nabb_4( \underline{\qf}^\F) + r\left(6\ka\kab+16\rho+8\rhoF^2\right)\underline{\qf}^{\F}\Bigg)\\
&-\rhoF(l.o.t.)_1, \\
\Box_\g \underline{\qf}^{\F}+\left( \ka\kab+3\rho\right)\ \qf^{\F}&=-\rhoF \Bigg(-\frac {1}{ r} \underline{\qf}\Bigg) +\rhoF^2 (l.o.t.)_2
\end{split}
\eea
\normalsize
\end{definition}

Observe that the coupled system of PDEs in \eqref{finalsystem-spin-2} is obtained from the system \eqref{finalsystem} by interchanging $\nabb_3$ with $\nabb_4$ and underlined quantities with non-underlined ones. 
Observe that the generalized Regge-Wheeler system of spin $\pm 2$ differ from each other only by the sign in front of the right hand side, since in interchanging the $e_3$ and $e_4$, the quantity $\rhoF=\frac 1 2 \F(e_3, e_4)$ changes sign.   The analysis of the two system will be completely analogous.

In Section \ref{transformation-theory}, we will show that given a solution $(\a, \ff)$ and $(\aa, \underline{\ff})$ of the spin $\pm 2$ Teukolsky equations, respectively, we can derive two solutions $(\qf,\qf^\F)$ and $(\underline{\qf},\underline{\qf}^\F)$, respectively, of the generalized Regge-Wheeler system. In view of the above considerations, it follows that we can associate to a linear gravitational and electromagnetic perturbations of Reissner-Nordstr{\"o}m, solutions to the Regge-Wheeler system. 

In the context of the proof of Main Theorem, we will do estimates directly at the level of the tensorial system \eqref{finalsystem}.

\subsection{Well-posedness}
For completeness, we state here a standard well-posedness theorem for the generalized Teukolsky systems.

\begin{proposition}[Well-posedness for the generalized Teukolsky system]\label{theorem-teuk} Let $(\a_0, \a_1) \in H^j_{loc}(\Sigma_0) \times H^{j-1}_{loc}(\Sigma_0)$ and $(\ff_0, \ff_1)\in H^j_{loc}(\Sigma_0) \times H^{j-1}_{loc}(\Sigma_0)$ be symmetric traceless $2$-tensor along $\Sigma_0$, with $j \geq 1$. Then there exists a unique pair $(\a, \ff)$ of symmetric traceless $2$-tensors on $\MM(0, \tau)$, for any $\tau > 0$, satisfying the generalized Teukolsky system of spin $+2$, with $(\a, \ff) \in H^j_{loc}(\Sigma_\tau) \times H^j_{loc}(\Sigma_\tau)$, $(n_{\Sigma_\tau}\a, n_{\Sigma_\tau}\ff) \in H^{j-1}_{loc}(\Sigma_\tau) \times H^{j-1}_{loc}(\Sigma_\tau)$, such that $(n_{\Sigma_0}\a, n_{\Sigma_0}\ff)|_{\Sigma_0} =(\a_1, \ff_1)$. \end{proposition}
A similar result holds for the generalized Teukolsky system for spin $-2$.

\section{Chandrasekhar transformation into Regge-Wheeler}\label{transformation-theory}

We now describe a transformation theory relating solutions of the generalized
Teukolsky system to solutions of the generalized Regge-Wheeler system. We emphasize that a physical space version of the Chandrasekhar transformation was first introduced in \cite{DHR}, for the Schwarzschild spacetime.

We introduce the following operators for a $n$-rank $S$-tensor $\Psi$:
\bea \label{operators}
\underline{P}(\Psi)&=&\frac{1}{\kab} \nabb_3(r \Psi), \qquad  P(\Psi)=\frac{1}{\ka} \nabb_4(r \Psi) 
\eea 
Observe that the operators $P$ and $\underline{P}$ above preserve the signature\footnote{As in \cite{Ch-Kl}, the signature of a tensor is given by the number of contraction with $e_4$ minus the number of contraction with $e_3$ in its definition. In the definition of $\underline{P}$, the derivative $\nabb_3$ lowers the signature or $\Psi$ of 1, and the division by $\kab$ (which has itself signature -1) raises the signature of 1.} of the tensor $\Psi$ as well as its rank. These operators can be thought of as rescaled derivatives: $\underline{P}$ is a rescaled version of $\nabb_3$, while $P$ is a rescaled version of $\nabb_4$. 

Recall that in a spherically symmetric spacetime $\nabb_3r=\frac 1 2 \kab r$ and $\nabb_4r=\frac 1 2 \ka r$, therefore
\bea
\underline{P}(\Psi)&=& \frac 1 2  r \Psi+\frac{1}{\kab} r\nabb_3( \Psi), \qquad P(\Psi)= \frac 1 2  r \Psi+\frac{1}{\ka} r\nabb_4( \Psi) \label{explicit-operators}
\eea

Given a solution  $(\alpha, \ff)$ of the generalized Teukolsky system of spin $+2$, we can define the following \emph{derived} quantities for $(\a, \ff)$:
\bea\label{quantities}
\begin{split}
\psi_0 &= r^2 \kab^2 \a, \\
\psi_1&=\underline{P}(\psi_0), \\
\psi_2&=\underline{P}(\psi_1)=\underline{P}(\underline{P}(\psi_0))=:\qf, \\
\psi_3&= r^2 \kab \ \ff, \\
\psi_4&= \underline{P}(\psi_3)=:\qf^\F
\end{split}
\eea
Similarly, given a solution  $(\aa, \underline{\ff})$ of the generalized Teukolsky system of spin $-2$, we can define the following \emph{derived} quantities for $(\aa, \underline{\ff})$:
\bea\label{quantities-2}
\begin{split}
\underline{\psi}_0 &= r^2 \ka^2 \aa, \\
\underline{\psi}_1&=P(\underline{\psi}_0), \\
\underline{\psi}_2&=P(\underline{\psi}_1)=P(P(\underline{\psi}_0))=:\underline{\qf}, \\
\underline{\psi}_3&= r^2 \ka \ \underline{\ff}, \\
\underline{\psi}_4&= P(\underline{\psi}_3)=:\underline{\qf}^\F
\end{split}
\eea
These quantities are again symmetric traceless $S$ $2$-tensors. 

\begin{remark} Observe that, even if $\ff$ is a symmetric traceless $2$-tensor, we apply the Chandrasekhar transformation only once to obtain the quantity $\qf^\F$ which verifies a Regge-Wheeler-type equation (as opposed to $\a$, for which the Chandrasekhar transformation is applied twice). This is because $\ff$ by definition is constructed from the one-form $\bF$, which verifies a spin $+1$ Teukolsky-type equation. 
\end{remark}

The following  proposition
is proven in Appendix B. 
\begin{proposition} \label{prop:rwt1}
Let $(\alpha, \ff)$ be a solution of the generalized Teukolsky system of spin $+2$ on $D^+(\Sigma_0)$. Then the symmetric traceless tensors $(\qf, \qf^\F)$ as defined through \eqref{quantities} satisfy the generalized Regge-Wheeler system of spin $+2$ on $D^+(\Sigma_0)$. 

Similarly, let $(\aa, \underline{\ff})$ be a solution of the generalized Teukolsky system of spin $-2$ on $D^+(\Sigma_0)$. Then the symmetric traceless tensors $(\underline{\qf}, \underline{\qf}^\F)$ as defined through \eqref{quantities-2} satisfy the generalized Regge-Wheeler system of spin $-2$ on $D^+(\Sigma_0)$. 
\end{proposition}
\begin{proof} In Lemma \ref{squareP}, we compute the wave equation verified by a derived quantity of the form $\underline{P}(\Psi)$. We use this lemma to derive the wave equation for $\qf^\F$ in Proposition \ref{wave-qfF}, from the Teukolsky equation for $\ff$. The wave equation for $\qf$ in Proposition \ref{wave-qf} is obtained from the Teukolsky equation for $\a$. See Appendix B.
\end{proof}

The fact that the derived quantities $(\qf, \qf^\F)$ satisfy the generalized Regge-Wheeler system, together with the transport relations \eqref{quantities}, will be the key to estimating the generalized Teukolsky equations.

\subsection{Relation with the higher order quantities defined in Schwarzschild}\label{relation-DHR}
In the linear and non-linear stability of Schwarzschild, similar derived quantities are defined.

In the work of linear stability of Schwarzschild \cite{DHR}, the transformation theory is defined in the following way (see Section 7.3 in \cite{DHR}):
\beaa
\psi:&=&-\frac{1}{2r\Omega^2}\nabb_3\left(r\Omega^2 \a\right), \qquad P:=\frac{1}{r^3\Omega} \nabb_3\left( \psi r^3 \Omega\right)
\eeaa
and the quantity that verifies the Regge-Wheeler equation is $\Psi=r^5P$, given therefore by
\bea\label{PsiDHR}
\Psi=\frac{r^2}{\Omega} \nabb_3\left( -\frac{r^2}{2\Omega}\nabb_3\left(r\Omega^2 \a\right)  \right)
\eea
Recall that in double null coordinates used in \cite{DHR}, $\kab=\tr\chib=-\frac 2 r \Omega$. Substituting $\Omega=-\frac r 2 \kab$ in \eqref{PsiDHR}, we obtain
\beaa
\Psi=-2\frac{r}{\kab} \nabb_3\left( \frac{r}{\kab}\nabb_3\left(\frac{1}{4}r^3\kab^2 \a\right)  \right)=-\frac 1 2 r \qf
\eeaa
which relate the $\Psi$ in \cite{DHR} to our $\qf$.

In the work of non-linear stability of Schwarzschild in \cite{stabilitySchwarzschild}, the derived quantities are defined in the following way at the linear level (see Appendix in \cite{stabilitySchwarzschild})\footnote{In the derivation of the non-linear terms, a different definition of the derived quantities is used, which is fundamental in the treatment of the non-linear terms. The two definitions coincide at the linear level.}:
\beaa
\Phi_0&=&r^2\kab^2 \a, \\
\Phi_1&=&\underline{P}(\Phi_0), \\
\Phi_2&=&\underline{P}(\Phi_1)
\eeaa 
In particular, in the case of vanishing electromagnetic tensor, i.e. $\ff=\ffb=0$, the quantities defined in \eqref{quantities} coincide with the above. The Regge-Wheleer system reduces to the equation:
\beaa
\Big(\square_{\g_M}+ \ka\kab\Big)\Phi_2&=& O(\ep^2)
\eeaa
which is the main equation used in \cite{stabilitySchwarzschild} to derive decay estimates for $\qf$, and subsequently for $\a$ and all other curvature and connection coefficients quantities.

\subsection{Relation with the linear stability of Reissner-Nordstr{\"o}m spacetime}
We will now finally relate the equations presented above to the full system of linearized
gravitational and electromagnetic perturbations of Reissner-Nordstr{\"o}m spacetime in the context of linear stability of Reissner-Nordstr{\"o}m.

Consider a solution to the linearized Einstein-Maxwell equations  around Reissner-Nordstr{\"o}m spacetime, as presented in Section \ref{linearized-equations}.
Then, the quantities $\a_{AB}$, $\aa_{AB}$ and the quantities $\ff_{AB}= \DDs_2 \bF_{AB}+\rhoF \chih_{AB}$ and $\underline{\ff}_{AB}=\DDs_2 \bbF_{AB} -\rhoF \chibh_{AB}$ verify the generalized Teukolsky system of spin $\pm2$. 
We obtain the following theorem.

\begin{theorem} \label{prop:relfull}
Let $\a$, $\aa$, $\ff$, $\underline{\ff}$ be the curvature components of a solution to the linearized Einstein-Maxwell equations around Reissner-Nordstr{\"o}m spacetime as in Section \ref{linearized-equations}. Then $(\a, \ff)$ satisfy the generalized Teukolsky system of spin $+2$, and $(\aa, \underline{\ff})$ satisfy the generalized Teukolsky system of spin $-2$. 
\end{theorem} 
\begin{proof} See Proposition \ref{squarea} and Proposition \ref{squareff} in Appendix A for the derivation of the generalized Teukolsky system for $(\a, \ff)$. 
\end{proof}

Using Proposition \ref{prop:rwt1}, we can therefore associate to any solution to the  linearized Einstein-Maxwell equations around Reissner-Nordstr{\"o}m spacetime, two symmetric traceless $2$-tensors which verify the generalized Regge-Wheeler system of spin $\pm 2$. The result in this paper will therefore imply boundedness and decay results for $\qf$, $\qf^\F$, $\a$, $\ff$, and $\underline{\qf}$, $\underline{\qf}^\F$, $\aa$, $\ffb$.

\section{Energy quantities and statements of the main theorem}\label{statement}
We first give the definitions of weighted energy quantities in Section \ref{definition-norms}, and then we provide the precise statement of the Main Theorem in Section \ref{section-main-theorem}.

\subsection{Definition of weighted energies}\label{definition-norms}
We define in this section a number of weighted energies. Recall the notation described in Section \ref{foliation}. 
Recall the vectorfield $T$ defined by \eqref{definition-T} in terms of the null frames $(e_3, e_4)$ and $(e_3^*, e_4^*)$. We define in addition the following vectorfield:
\beaa
R=\frac 1 2(-\Up e_3^*+e_4^*)=\frac 1 2 (-e_3+\Up e_4)
\eeaa
 Note that in coordinates $T=\partial_t$ and $R=\Up \partial_r$, and
 \bea
g(T,T)=-g(R, R) = -\Up, \qquad g(T, R) = 0, \qquad T(r)=0, \qquad R(r)=\Up \label{R(r)}
 \eea
Notice that $T$ and $R$ are both parallel to $e_4^*$ at the horizon. To control all the derivatives at the horizon, we define the following modified $T$ and $R$.        Let $\th$  a smooth bump  function   equal $1$ on  $\MM_{red}$,  vanishing for  $r\ge \frac{12}{10}r_{\mathcal{H}}$ and define the vectorfields 
      \bea
       \label{eq:Rc-Tc}
       \begin{split}
  \Rbrev&:=  \th  \frac 1 2 ( e^*_4-e^*_3) +(1-\th)        \Up^{-1} R\\
  \Tbrev&:=\th  \frac 1 2 ( e^*_4+e^*_3) +(1-\th)        \Up^{-1} T
  \end{split}
      \eea 
      We also introduce the following notation:
      \beaa
     \check{\nabb_4}(\Psi):=\nabb_4(\Psi)+\frac 1 r \Psi 
      \eeaa
which will be used in Section \ref{section-rp-estimates} in the context of the $r^p$-estimates. 
We also denote $\nabb_4^*=\nabb_{e_4^*}$,  $\nabb_3^*=\nabb_{e_3^*}$.

\subsubsection{Weighted energies for $\qf$, $\qf^\F$, $\underline{\qf}$, $\underline{\qf}^\F$}

The energies in this section will in general be applied to $\Psi=\qf$, $\qf^\F$, $\underline{\qf}$ or $\underline{\qf}^\F$. Let $p$ be a free parameter, which will eventually take the values $\de \leq p \leq 2-\de$.

    We introduce the following  weighted energies for $\Psi$. To simplify the notations, we suppress the volume form, and write $\int_{\MM} d\operatorname{vol}_{\MM}=\int_r \int_t \int_0^{\pi} \int_{0}^{2\pi} r^2 \sin\th dt dr d\th d\phi $ for the spacetime integrals and  $\int_{\Sigma} d \operatorname{vol}_{\Sigma}=\int r^2 \sin\th dr d\th d\phi $ for integrals along spacelike hypersurfaces. 
    \begin{enumerate}
    \item Energy quantities on $\Si_\tau$:
    \begin{itemize}
    \item Basic energy quantity\footnote{The definition of the energy is divided in three regions because the foliation $\Si_\tau$ is defined to be null near $\mathcal{H}^+$ and null near $\mathcal{I}^+$, and therefore along null cones not all derivatives can be controlled. }
\bea
\label{def:basic-energy}
\begin{split}
E[\Psi](\tau):&= \int_{\Si_{red}}    |\nabb_3^*\Psi|^2 +|\nabb\Psi|^2 + r^{-2}|\Psi|^2\\
&+ \int_{\Si_\tau \setminus (\Si_{red} \cup \Si_{far})}    |\nabb_4 \Psi|^2  + |\nabb_3\Psi|^2 +|\nabb\Psi|^2 + r^{-2}|\Psi|^2 \\
&+ \int_{\Si_{far}}    |\nabb_4 \Psi|^2   +|\nabb\Psi|^2 + r^{-2}|\Psi|^2
\end{split}
\eea
Notice that the energy $E[\Psi](\tau)$ contains regular derivative close to the horizon. 
\item Weighted energy quantity in the far away region
\beaa
 \bsplit
  E_{p\,; \,R}[\Psi](\tau):&=  \int_{\Si_{\ge  R}(\tau)}  r^p |\check{\nabb}_4\Psi|^2
  \end{split}
  \eeaa
  \item Weighted energy quantity
  \beaa
  \bsplit
  E_{p}[\Psi](\tau):&=E[\Psi](\tau)+ E_{p\,; \,R}[\Psi](\tau)
  \end{split}
  \eeaa

\end{itemize}
\item Energy quantities on the event horizon $\mathcal{H}^+$:
\beaa
E_{\mathcal{H}^+}[\Psi](\tau_1, \tau_2):&=& \int_{\mathcal{H}^+(\tau_1, \tau_2)} |\nabb_4^*\Psi|^2+|\nabb \Psi|^2+|\Psi|^2  
\eeaa

\item Energy quantities on null infinity $\mathcal{I}^+$:
\beaa
E_{\mathcal{I}^+, p}[\Psi](\tau_1, \tau_2):&=& \int_{\mathcal{I}^+(\tau_1, \tau_2)} |\nabb_3\Psi|^2+r^p|\nabb \Psi|^2+r^{p-2}|\Psi|^2  
\eeaa

    \item  Weighted spacetime bulk energies in $\MM(\tau_1, \tau_2)$:
    \begin{itemize}
    \item Basic Morawetz bulk 
    \bea
    \label{def:Mor-bulk}
    \begin{split}
\Mor[\Psi](\tau_1, \tau_2):= &\int_{\MM(\tau_1, \tau_2)} \frac{M^2}{r^3} |R(\Psi)|^2+ \frac{M}{r^4} |\Psi|^2 \\
&+ \int_{\MM(\tau_1, \tau_2)}\frac{(r^2-3M r +2Q^2)^2}{r^5}\left( |\nabb \Psi|^2+\frac{M^2}{r^2} |T\Psi|^2 \right)     
\end{split}
\eea
Notice that the Morawetz bulk $\Mor[\Psi](\tau_1, \tau_2)$ is degenerate at the photon sphere $r=r_P$. 
\item Red-shifted Morawetz bulk
 \beaa
\Mor_{\mathcal{H}^+}[\Psi](\tau_1, \tau_2):= \int_{\MM(\tau_1, \tau_2)} |\Rbrev(\Psi)|^2 +|\Tbrev \Psi|^2    +|\nabb\Psi|^2  + M^{-2} |\Psi|^2
     \eeaa
\item Improved Morawetz bulk
\beaa
\bsplit
\Morr[\Psi](\tau_1, \tau_2)&:=\Mor[\Psi](\tau_1, \tau_2)+\int_{\MM_{far}(\tau_1, \tau_2)} r^{-1-\de}  |\nabb_3(\Psi)|^2
 \end{split}
\eeaa
\item  Weighted bulk norm in the far away region 
\bea\label{Morawetz-far-away}
\MM_{p\,; \,R}[\Psi](\tau_1, \tau_2):&=\int_{\MM_{\ge  R}(\tau_1, \tau_2) }  r^{p-1}  \left( p | \check{\nabb}_4(\Psi) |^2 +(2-p)  ( |\nabb \Psi|^2+   r^{-2}  |\Psi|^2)\right)  
\eea
\item Weighted bulk norm
\beaa
\MM_{p}[\Psi](\tau_1, \tau_2):&=\Morr[\Psi](\tau_1, \tau_2)+\Mor_{\mathcal{H}^+}[\Psi](\tau_1, \tau_2)+\MM_{p\,; \,R}[\Psi](\tau_1, \tau_2)
\eeaa
\end{itemize}
    \end{enumerate}
    
    \subsubsection{Weighted energies for $\a$, $\psi_1$, $\ff$ and $\aa$, $\underline{\psi}_1$, $\underline{\ff}$}
    The quantities in this section will be applied to $\a$, $\psi_1$, $\ff$ and to $\aa$, $\underline{\psi}_1$, $\underline{\ff}$.

    \begin{enumerate}
    \item Weighted energy quantities on $\Si_\tau$:
    \begin{itemize}
    \item Basic energy quantities
     \beaa
      E[\ff](\tau):&=&  \int_{\Si_{red}}    |\nabb_3^*\ff|^2 +|\nabb\ff|^2 + r^{-2}|\ff|^2\\
      &&+\int_{\Si_\tau \setminus (\Si_{red} \cup \Si_{far})}|\ff|^2+|T\ff|^2+|R\ff|^2+|\nabb \ff|^2\\
      &&+ \int_{\Si_{far}}    |\nabb_4 \ff|^2   +|\nabb\ff|^2 + r^{-2}|\ff|^2\\
    E[\psi_1](\tau):&=& \int_{\Si_{red}}    |\nabb_3^*\psi_1|^2 +|\nabb\psi_1|^2 + r^{-2}|\psi_1|^2\\
    &&+ \int_{\Si_\tau \setminus (\Si_{red} \cup \Si_{far})} |\psi_1|^2+ |T\psi_1|^2+|R\psi_1|^2\\
    &&+ \int_{\Si_{far}}    |\nabb_4 \psi_1|^2   +|\nabb\psi_1|^2 + r^{-2}|\psi_1|^2\\
     E[\a](\tau):&=& \int_{\Si_{red}}    |\nabb_3^*\a|^2 +|\nabb\a|^2 + r^{-2}|\a|^2\\
     &&+ \int_{\Si_\tau \setminus (\Si_{red} \cup \Si_{far})} |\a|^2+ |T\a|^2+|R\a|^2+|\nabb\a|^2\\
     &&+ \int_{\Si_{far}}    |\nabb_4 \a|^2   +|\nabb\a|^2 + r^{-2}|\a|^2
    \eeaa
    and similarly for $\underline{\ff}$, $\underline{\psi_1}$ and $\aa$. 
\item Weighted energy quantities in the far away region
     \beaa
    E_{p\,; \,R}[\ff](\tau):&=&\int_{\Si_{\ge  R}(\tau)}r^{2+p}|\ff|^2+r^{p+4}|\check{\nabb}_4\ff|^2+r^{4+p}|\nabb \ff|^2\\
    E_{p\,; \,R}[\psi_1](\tau):&=& \int_{\Si_{\ge  R}(\tau)} r^p |\psi_1|^2+r^{p+2}|\check{\nabb}_4\psi_1|^2+r^{2+p}|\nabb\psi_1|^2\\
     E_{p\,; \,R}[\a](\tau):&=& \int_{\Si_{\ge  R}(\tau)} r^{2+p} |\a|^2+r^{p+4}|\check{\nabb}_4\a|^2+r^{4+p} |\nabb\a|^2
    \eeaa
   \item Weighted energy quantities
  \beaa
  \bsplit
  E_{p}[\ff](\tau):&=E[\ff](\tau)+ E_{p\,; \,R}[\ff](\tau), \\
  E_{p}[\psi_1](\tau):&=E[\psi_1](\tau)+ E_{p\,; \,R}[\psi_1](\tau), \\
  E_{p}[\a](\tau):&=E[\a](\tau)+ E_{p\,; \,R}[\a](\tau)
  \end{split}
  \eeaa
 
\end{itemize}

\item Energy quantities on the event horizon $\mathcal{H}^+$:
\beaa
E_{\mathcal{H}^+}[\ff](\tau_1, \tau_2):&=& \int_{\mathcal{H}^+(\tau_1, \tau_2)} |\nabb_4^*\ff|^2+|\nabb \ff|^2+|\ff|^2 \\
E_{\mathcal{H}^+}[\psi_1](\tau_1, \tau_2):&=& \int_{\mathcal{H}^+(\tau_1, \tau_2)} |\nabb_4^*\psi_1|^2+|\nabb \psi_1|^2+|\psi_1|^2  \\
E_{\mathcal{H}^+}[\a](\tau_1, \tau_2):&=& \int_{\mathcal{H}^+(\tau_1, \tau_2)} |\nabb_4^*\a|^2+|\nabb \a|^2+|\a|^2 
\eeaa
  
    \item Non-degenerate bulk norms on $\MM(\tau_1, \tau_2)$:
    \begin{itemize}
    \item Basic non-degenerate bulk norms
     \beaa
      \widehat{\MM}[\ff](\tau_1, \tau_2)&=&\int_{\MM(\tau_1, \tau_2)}|\ff|^2+|T\ff|^2+|R\ff|^2+|\nabb \ff|^2\\
    \widehat{\MM}[\psi_1](\tau_1, \tau_2)&=& \int_{\MM(\tau_1, \tau_2)} |\psi_1|^2+ |T\psi_1|^2+|R\psi_1|^2\\
     \widehat{\MM}[\a](\tau_1, \tau_2)&=& \int_{\MM(\tau_1, \tau_2)} |\a|^2+ |T\a|^2+|R\a|^2+|\nabb\a|^2
    \eeaa
    and similarly for $\underline{\ff}$, $\underline{\psi_1}$ and $\aa$.
    \item Weighted bulk norms in the far away region
\beaa
\widehat{\MM}_{p\,; \,R}[\ff](\tau_1, \tau_2):&=&\int_{\MM_{\ge  R}(\tau_1, \tau_2) } r^{1+p}|\ff|^2+r^{3+p}|\nabb_3\ff|^2+r^{p+3}|\check{\nabb}_4\ff|^2+r^{3+p}|\nabb\ff|^2, \\
\widehat{\MM}_{p\,; \,R}[\psi_1](\tau_1, \tau_2):&=&\int_{\MM_{\ge  R}(\tau_1, \tau_2) } r^{-1+p}|\psi_1|^2+r^{1+p}|\nabb_3\psi_1|^2+r^{p+1}|\check{\nabb}_4\psi_1|^2+r^{1+p}|\nabb\psi_1|^2 \nn\\
\widehat{\MM}_{p\,; \,R}[\a](\tau_1, \tau_2):&=&\int_{\MM_{\ge  R}(\tau_1, \tau_2) }r^{1+p}|\a|^2+r^{p+3}|\nabb_3\a|^2+r^{p+3}|\check{\nabb}_4\a|^2+r^{3+p}|\nabb\a|^2\\
\widehat{\MM}_{\de\,; \,R}[\underline{\ff}](\tau_1, \tau_2):&=&\int_{\MM_{\ge  R}(\tau_1, \tau_2) } r^{1-\de}|\underline{\ff}|^2+r^{1-\de}|\nabb_3\underline{\ff}|^2+r^{3+\de}|\nabb_4\underline{\ff}|^2+r^{3-\de}|\nabb\underline{\ff}|^2, \\
\widehat{\MM}_{\de\,; \,R}[\underline{\psi}_1](\tau_1, \tau_2):&=&\int_{\MM_{\ge  R}(\tau_1, \tau_2) } r^{-1-\de}|\underline{\psi}_1|^2+r^{-1-\de}|\nabb_3\underline{\psi}_1|^2+r^{1+\de}|\nabb_4\underline{\psi}_1|^2+r^{1-\de}|\nabb\underline{\psi}_1|^2 \nn\\
\widehat{\MM}_{\de\,; \,R}[\aa](\tau_1, \tau_2):&=&\int_{\MM_{\ge  R}(\tau_1, \tau_2) }r^{-1-\de}|\aa|^2+r^{-1-\de}|\nabb_3\aa|^2+r^{3+\de}|\nabb_4\aa|^2+r^{3-\de}|\nabb\aa|^2
\eeaa
\item Weighted bulk norms
  \beaa
  \bsplit
  \widehat{\MM}_{p}[\ff](\tau):&=\widehat{\MM}[\ff](\tau)+ \widehat{\MM}_{p\,; \,R}[\ff](\tau), \\
  \widehat{\MM}_{p}[\psi_1](\tau):&=\widehat{\MM}[\psi_1](\tau)+ \widehat{\MM}_{p\,; \,R}[\psi_1](\tau), \\
  \widehat{\MM}_{p}[\a](\tau):&=\widehat{\MM}[\a](\tau)+ \widehat{\MM}_{p\,; \,R}[\a](\tau)\\
  \widehat{\MM}_{\de}[\underline{\ff}](\tau):&=\widehat{\MM}[\underline{\ff}](\tau)+ \widehat{\MM}_{\de\,; \,R}[\underline{\ff}](\tau), \\
  \widehat{\MM}_{\de}[\underline{\psi}_1](\tau):&=\widehat{\MM}[\underline{\psi}_1](\tau)+ \widehat{\MM}_{\de\,; \,R}[\underline{\psi}_1](\tau), \\
  \widehat{\MM}_{\de}[\aa](\tau):&=\widehat{\MM}[\aa](\tau)+ \widehat{\MM}_{\de\,; \,R}[\aa](\tau)
  \end{split}
  \eeaa

\end{itemize}
    \end{enumerate}

\subsubsection{Weighted energy for the inhomogeneous term}
We introduce the following norm for the right hand side $\M$ of the Regge-Wheeler equation in the far away region:
\bea\label{definition-norm-M}
\II_{p\,; \,R}[ \M](\tau_1,\tau_2) :&=& \int_{\MM_{\ge R}(\tau_1,\tau_2)  } r^{1+p} |\M|^2
\eea

\subsubsection{Higher order energies}
   
   To estimate higher order energies we also introduce the following notation, motivated by the fact that the Regge-Wheeler system commutes with $T$ and the angular momentum operators $\Omega_i$.  We define 
   \begin{enumerate}
\item Higher derivative energies for $n\geq 1$:
   \beaa
   E^{n, T, \nabb}[\Psi]&=&\sum_{i+j \leq n} E[T^i(r\nabb_A)^j\Psi], \qquad  E^{n, T, \nabb}_{p\,; \,R}[\Psi]= \sum_{i+j \leq n}E_{p\,; \,R}[T^i(r\nabb_A)^j\Psi], \\
   E^{n, T, \nabb}_{p}[\Psi]&=& E^{n, T, \nabb}[\Psi]+E^{n, T, \nabb}_{p\,; \,R}[\Psi]
   \eeaa
   and
   \beaa
  E_p^{n, T, \nabb}[\ff]&=&\sum_{i+j \leq n} E_p[T^i(r\nabb_A)^j\ff] , \qquad  E_p^{n, T, \nabb}[\psi_1]=\sum_{i+j \leq n} E_p[T^i(r\nabb_A)^j\psi_1]\\
   E_p^{n, T, \nabb}[\a]&=&\sum_{i+j \leq n} E_p[T^i(r\nabb_A)^j\a] 
   \eeaa
   Similarly for the energies at the horizon and at null infinity. 
   \item Higher derivative Morawetz bulks:
   \beaa
   \Mor^{n, T, \nabb}[\Psi]&=&\sum_{i+j \leq n} \Mor[T^i(r\nabb_A)^j\Psi], \qquad  \Morr^{n, T, \nabb}[\Psi]=\sum_{i+j \leq n} \Morr[T^i(r\nabb_A)^j\Psi], \\
   \MM^{n, T, \nabb}_{p\,; \,R}[\Psi]&=&\sum_{i+j \leq n}\MM_{p\,; \,R}[T^i(r\nabb_A)^j\Psi], \qquad \widehat{\MM}^{n, T, \nabb}_p[\Psi]=\sum_{i+j \leq n} \widehat{\MM}_p[T^i(r\nabb_A)^j\Psi]\\
   \MM^{n, T, \nabb}_{p}[\Psi]&=& \Morr^{n, T, \nabb}[\Psi]+\MM^{n, T, \nabb}_{p\,; \,R}[\Psi], 
   \eeaa
   and
   \beaa
    \widehat{\MM}^{n, T, \nabb}_{p}[\ff]&=& \sum_{i+j \leq n}\widehat{\MM}_p[T^i(r\nabb_A)^j\ff], \qquad \widehat{\MM}^{n, T, \nabb}_{p}[\psi_1]= \sum_{i+j \leq n}\widehat{\MM}_p[T^i(r\nabb_A)^j\psi_1]\\
     \widehat{\MM}^{n, T, \nabb}_{p}[\a]&=& \sum_{i+j \leq n}\widehat{\MM}_p[T^i(r\nabb_A)^j\a]
   \eeaa
   Similarly for the red-shifted Morawetz bulks.
   \item Higher derivative norm for $\MM$:
   \beaa
   \II^{n, T, \nabb}_{p\,; \,R}[ \M]&=&\sum_{i+j \leq n} \II_{p}[ T^i(r\nabb_A)^j\M] 
   \eeaa
 \end{enumerate}

\subsection{Precise statement of the Main Theorem}\label{section-main-theorem}
We are now ready to state the boundedness and decay theorem for solutions $(\a, \ff)$ of the generalized Teukolsky system.

In the estimates below we will denote $\AA\les \BB$ if there exists an universal constant $C$ such that $\AA \le C \BB$.

\begin{theorem}\label{main-theorem-1}(Spin $+2$) Let $(\a_0, \a_1) \in H^j_{loc}(\Sigma_0) \times H^{j-1}_{loc}(\Sigma_0)$ and $(\ff_0, \ff_1)\in H^j_{loc}(\Sigma_0) \times H^{j-1}_{loc}(\Sigma_0)$ be as in the well-posedness Proposition \ref{theorem-teuk}, and let $\a$, $\psi_1$, $\qf$, $\ff$, $\qf^\F$ be as defined by \eqref{quantities}.  Then the following estimates hold, for all $\de \le p \le 2-\de$ and  for any $\tau>0$: 
\begin{enumerate}
\item  energy and red-shifted boundedness, degenerate integrated local energy decay and $r^p$ hierarchy of estimates for $\qf$ and $\qf^\F$:
       \bea\label{first-estimate-main-theorem-1}
  \begin{split}
     & E_p[\qf](\tau)  + E^{1, T, \nabb}_p[\qf^\F](\tau) + E_{\mathcal{H}^+}[\qf](0, \tau)  + E_{\mathcal{H}^+}^{1, T, \nabb}[\qf^\F](0,\tau)  +E_{\mathcal{I}^+, p}[\qf](0, \tau) +E^{1, T, \nabb}_{\mathcal{I}^+, p}[\qf^\F](0, \tau)            \\
     &+\MM_p[\qf](0,\tau) +\MM^{1, T, \nabb}_p[\qf^\F](0,\tau)\\
      &\les E_{p}[\qf](0)+E^{1, T, \nabb}_p[\qf^\F](0)+E[\ff](0)+E_p[\psi_1](0)+E_p[\a](0)
      \end{split}
       \eea
       \item energy boundedness, integrated local energy decay and $r^p$ hierarchy of estimates for $\a$, $\psi_1$, $\ff$: 
    \bea\label{first-estimate-a-f-main-theorem-1}
    \begin{split}
  &  E_p[\a](\tau)+E_p[\psi_1](\tau)+E_p^{1, T, \nabb}[\ff](\tau)+ E_{\mathcal{H}^+}[\a](0, \tau)+E_{\mathcal{H}^+}[ \psi_1](0, \tau)+E_{\mathcal{H}^+}^{1, T, \nabb}[\ff](0, \tau)\\
  &+\widehat{\MM}_p[\a](0, \tau)+ \widehat{\MM}_p[\psi_1](0, \tau)+ \widehat{\MM}^{1, T, \nabb}_p[\ff](0, \tau)\\
&\les E_{p}[\qf](0)+E^{1, T, \nabb}_p[\qf^\F](0)+E_p[\a](0)+ E_p[\psi_1](0)+ E_p^{1, T, \nabb}[\ff](0)
\end{split}
  \eea
 \item higher order energy and integrated decay estimates for $\qf$ and $\qf^\F$, for any integer $n \geq 1$:
   \bea\label{higher-estimate-main-theorem}
   \begin{split}
   & E_p^{n, T, \nabb}[\qf](\tau) + E_p^{n+1, T, \nabb}[\qf^\F](\tau)       +  \MM_{p}^{n, T, \nabb}[\qf](0, \tau) +  \MM_{p}^{n+1, T, \nabb}[\qf^\F](0, \tau)   \\
&\les E_p^{n, T, \nabb}[\qf](0)+E_p^{n+1, T, \nabb}[\qf^\F](0)+E_p^{n, T, \nabb}[\ff](0)+E_p^{n, T, \nabb}[\psi_1](0)+E_p^{n, T, \nabb}[\a](0)
\end{split}
   \eea
   \item higher order energy and integrated decay estimates for $\a$, $\psi_1$ and $\ff$, for any integer $n \geq 1$:
   \bea\label{higher-estimate-main-theorem-a-f}
   \begin{split}
  & E^{n, T, \nabb}_p[\a](\tau)+E^{n, T, \nabb}_p[\psi_1](\tau)+E^{n+1, T, \nabb}_p[\ff](\tau)\\
  &+\widehat{\MM}^{n, T, \nabb}_p[\a](0, \tau)+ \widehat{\MM}^{n, T, \nabb}_p[\psi_1](0, \tau)+ \widehat{\MM}^{n+1, T, \nabb}_p[\ff](0, \tau)\\
&\les E_p^{n, T, \nabb}[\qf](0)+E_p^{n+1, T, \nabb}[\qf^\F](0)+E_p^{n+1, T, \nabb}[\ff](0)+E_p^{n, T, \nabb}[\psi_1](0)+E_p^{n, T, \nabb}[\a](0)
\end{split}
   \eea
    \item polynomial decay for the energy:
\bea\label{polynomial-decay-theorem}
\begin{split}
&E_{\de}[\a](\tau)+E_{\de}[\psi_1](\tau)+E^{1, T, \nabb}_{\de}[\ff](\tau) + E_{\de}[\qf](\tau)+ E_{\de}^{1, T, \nabb}[\qf^\F](\tau)\\
&\les \frac{1}{\tau^{2-\de}}\mathbb{D}_{2, 2-\de}[\a, \psi_1, \qf, \ff, \qf^\F](0)
\end{split}
\eea
where 
\beaa
\mathbb{D}_{2, 2-\de}[\a, \psi_1, \qf, \ff, \qf^\F](0)&=&E_{2-\de}^{2, T, \nabb}[\a](0)+E_{2-\de}^{2, T, \nabb}[\psi_1](0)+E_{2-\de}^{3, T, \nabb}[\ff](0)\\
&&+ E_{2-\de}^{2, T, \nabb}[\qf](0)+E_{2-\de}^{3, T, \nabb}[\qf^\F](0)
\eeaa

\end{enumerate}
\end{theorem}

As an example of the pointwise estimate which follow from the above theorem, we point out the pointwise estimate for the solutions to the generalized Teukolsky system $(\a, \ff)$. 

\begin{corollary} Let $(\a_0, \a_1)$ and $(\ff_0, \ff_1)$ be smooth and of compact support. Then the solution $(\a, \ff)$ satisfy
\beaa
|r^{\frac{5+\de}{2}} \a| \leq C\tau^{-\frac{2-\de}{2}}, \qquad |r^{\frac{5+\de}{2}} \ff| \leq C\tau^{-\frac{2-\de}{2}}
\eeaa
where $C$ depends on an appropriate higher Sobolev norm.
\end{corollary}

A similar theorem holds for the spin $-2$ case, with the corresponding energies defined above. It implies the following pointwise estimate.

\begin{corollary} Let $(\aa_0, \aa_1)$ and $(\underline{\ff}_0, \underline{\ff}_1)$ be smooth and of compact support. Then the solution $(\aa, \underline{\ff})$ satisfy
\beaa
|r \aa| \leq C\tau^{-\frac{2-\de}{2}}, \qquad |r^2\underline{\ff}| \leq C\tau^{-\frac{2-\de}{2}}
\eeaa
where $C$ depends on an appropriate higher Sobolev norm.
\end{corollary}

Observe that the $r$-weights in these estimates, even though they do not seem to come from a standard application of the $r^p$-method, are consistent with the decay for $\a$ in Schwarzschild (see \cite{DHR} and \cite{stabilitySchwarzschild}). This is because one does not apply the $r^p$-estimates directly to the quantities $\a$ and $\ff$, but first needs to apply it to $\qf$ and $\qf^\F$ and then integrate in the $e_3$ direction. As a result of those transport estimates, exemplified in Lemma \ref{general-lemma-transport}, one obtains the above $r$-weights for $\a$, $\ff$, $\aa$, $\underline{\ff}$.

\subsection{The logic of the proof}
The remainder of the paper concerns the proof of Main Theorem. 
We outline here the main steps.
\begin{enumerate}
\item We consider the two equations composing the system \eqref{finalsystem} separately and derive two separated integrated local energy decay for them, with right hand side which is not controlled at this stage. We use standard techniques in order to derive the equations: we apply the vector field method to obtain energy and Morawetz estimates. We also improve the estimates with the red-shift vector field. Finally, we use the $r^p$-method of Dafermos and Rodnianski to obtain integrated decay in the far-away region. We also obtain higher order estimates by commuting the equations with the Killing vector fields. 
This is done in Section \ref{separate-estimates}. 
\item We analyze the right hand side of the separated estimates, and separate it into coupling terms and lower order terms. We derive estimates for the coupling terms on the right hand side and transport estimates for the lower order terms on the right hand side. Our goal is to absorb the norms of these inhomogeneous terms on the right hand side with the Morawetz bulks of the estimates on the left hand side, using the smallness of the charge.  This is done in Section \ref{lot-absorbing}.

The coupling terms are particularly problematic because of the degeneracy of the Morawetz bulks at the photon sphere. We present a cancellation of the higher order terms at the photon sphere which allows to close the estimates. This is done in Section \ref{coupling-subsection}. 

The lower order terms are treated using enhanced transport estimates, which make also use of Bianchi identities. This is done in Section \ref{section-lower-order}.

\item We put together the above estimates to prove the Main Theorem. This is done in Section \ref{proof-of-theorem-this-is-the-end}.
\end{enumerate}

We will show all the details of the proof for the spin $+2$ system. We outline the proof of the case spin $-2$ and the main differences with the spin $+2$ in Section \ref{spin-2-all}.

\section{Estimates for the Regge-Wheeler equations separately}\label{separate-estimates}
In this section, we prove estimates for the two equations separately, keeping the right hand side of the equations unchanged. 
Schematically, we write the two equations in  \eqref{finalsystem} as a general expression:
\bea\label{general-equation}
\Big(\Box_\g-V_i \Big) \Psi_i&=& \M_i
\eea
where $\M_i$ are whatever expressions we have on the left hand side of the equation. The two equations comprising the system are obtained by
\bea\label{potentials}
\Psi_1&=&\qf, \ \underline{\qf}, \qquad V_1=-\ka\kab+10 \rhoF^2=\frac{1}{r^2}\left(4-\frac{8M}{r}+\frac{14Q^2}{r^2}\right), \\
\Psi_2&=&\qf^\F, \ \underline{\qf}^\F \qquad V_2=-\ka\kab-3\rho=\frac{1}{r^2}\left(4-\frac{2M}{r}-\frac{2Q^2}{r^2}\right)
\eea
Observe that $\Psi_1$ and $\Psi_2$ are symmetric traceless $2$-tensors, therefore the operator $\Box_\g$ refers to the wave operator applied to a $2$-tensor. We will analyze the general equation without worrying for now of the right hand side $\M$.

The goal of this section is to prove the following two estimates, for $\qf$ and $\qf^\F$ respectively:
    \bea\label{estimate1}
       \begin{split}
      &  E_p[\qf](\tau)      + E_{\mathcal{H}^+}[\qf](0, \tau)  +E_{\mathcal{I}^+, p}[\qf](0, \tau)      +\MM_p[\qf](0,\tau)\\
        &\les E_{p}[\qf](0)+\II_{p\,; \,R}[ \M_1[\qf, \qf^\F]](0,\tau) \\
        &- \int_{\MM(0,\tau)} \left((r-r_P)R(\qf)+ \Lambda T(\qf )+\frac 1 2  w\qf\right)\c \M_1[\qf, \qf^\F]
        \end{split}
  \eea
  and
      \bea\label{estimate2}
       \begin{split}
    &  E^{1, T, \nabb}_p[\qf^\F](\tau)    + E_{\mathcal{H}^+}^{1, T, \nabb}[\qf^\F](0,\tau)  +E^{1, T, \nabb}_{\mathcal{I}^+, p}[\qf^\F](0, \tau)          +\MM^{1, T, \nabb}_p[\qf^\F](0,\tau)\\
      &\les E^{1, T, \nabb}_p[\qf^\F](0)+  \II^{1, T, \nabb}_{p\,; \,R}[\M_2[\qf, \qf^\F]]( 0,\tau)  \\
       &- \int_{\MM(0,\tau)} \left((r-r_P)R(\qf^\F)+ \Lambda T(\qf^\F )+\frac 1 2  w\qf^\F\right)\c \M_2[\qf, \qf^\F]\\
        &- \int_{\MM(0,\tau)} \left((r-r_P)R(T\qf^\F)+ \Lambda T(T\qf^\F )+\frac 1 2  w T\qf^\F\right)\c T(\M_2[\qf, \qf^\F])\\
        &- \int_{\MM(0,\tau)} \left((r-r_P)R(r\nabb_A\qf^\F)+ \Lambda T(r\nabb_A\qf^\F )+\frac 1 2  w r\nabb_A\qf^\F\right)\c r\nabb_A(\M_2[\qf, \qf^\F])
   \end{split}
  \eea
  where $w$ is defined in \eqref{definitionw} and $\Lambda$ is chosen big enough in Section \ref{subsection-energy-estimates}. 
 Identical estimates hold for $\underline{\qf}$ and $\underline{\qf}^\F$.

 In Section \ref{lot-absorbing}, we will then combine the two estimates above to a unique estimate for the whole system.

\subsection{Preliminaries}
We collect in this section some preliminaries lemmas in order to apply the vector field method to equation \eqref{general-equation}.

Consider the   energy-momentum tensor associated to the wave equation \eqref{general-equation}:
\bea\label{defenergymomentum}
\begin{split}
 \QQ_{\mu\nu}:&=\D_\mu  \Psi \c \D _\nu \Psi 
          -\frac 12 \g_{\mu\nu} \left(\D_\la \Psi\c\D^\la \Psi + V_i\Psi \c \Psi\right)\\
          &=\D_\mu  \Psi \c \D _\nu \Psi 
          -\frac 12 \g_{\mu\nu} \LL_i[\Psi]
          \end{split}
 \eea
   where $\c$ denotes the scalar product induced by $\g$ and $\D$ is the spacetime covariant derivative. The second line of \eqref{defenergymomentum} defines the Lagrangian $\LL_i[\Psi]$. Notice that 
   \bea\label{components-Q}
    \QQ_{33}=|\nabb_3 \Psi|^2, \qquad \QQ_{44}=|\nabb_4 \Psi|^2, \qquad \QQ_{34}=|\nabb \Psi|^2+V_i |\Psi|^2
   \eea
\begin{lemma}
The divergence of the energy-momentum tensor $\QQ$ verifies
 \bea\label{divergenceQ}
 \D^\nu\QQ_{\mu\nu}
  &=& \D_\mu  \Psi \c\M_i + \D^\nu  \Psi ^A\R_{ A   B   \nu\mu}\Psi^B-\frac 1 2 \D_\mu V_i \Psi\c \Psi
 \eea
and its trace is given by 
\bea\label{trace-Q}
\g^{\mu\nu}\QQ_{\mu\nu}&=&-\LL_i(\Psi)- V_i|\Psi|^2
\eea
\end{lemma}

We will make use of the following standard computation.
\begin{proposition}
\label{prop:qf-tensorial'}
Consider  a $S$-tensor $\Psi$ verifying the spacetime equation \eqref{general-equation}.
Let $X= a(r) e_3+b(r) e_4$ be a vectorfield, $w$ a scalar  function and $M$  a one form. Defining
 \bea\label{definition-of-P}
 \PP_\mu^{(X, w, M)}[\Psi]&=&\QQ_{\mu\nu} X^\nu +\frac 1 2  w \Psi \D_\mu \Psi -\frac 1 4\Psi^2   \pr_\mu w +\frac 1 4 \Psi^2 M_\mu,
  \eea
 then,  
  \bea
  \label{le:divergPP-gen}
  \begin{split}
  \D^\mu \PP_\mu^{(X, w, M)}[\Psi]&= \frac 1 2 \QQ  \c\piX+\left( - \frac 1 2 X( V_i ) -\frac 1 4   \Box_\g  w \right)|\Psi|^2+\frac 12  w \LL_i[\Psi] +\frac 1 4  \D^\mu (\Psi^2 M_\mu)      \\
  &+  \left(X( \Psi )+\frac 1 2   w \Psi\right)\c \M_i[\Psi] 
   \end{split}
 \eea
 
\end{proposition}

 {\bf Notation.} 
For convenience    we  introduce the notation,
 \bea
 \label{eq:modified-div}
 \EE[X, w, M](\Psi)&:=&   \D^\mu  \PP_\mu^{(X, w, M)}[\Psi]  -   \left(X( \Psi )+\frac 1 2   w \Psi\right)\c \M_i   
  \eea
  Thus equation  \eqref{le:divergPP-gen} becomes
   \bea
   \label{le:divergPP-gen-EE}
   \EE[X, w, M](\Psi)&=& \frac 1 2 \QQ  \c\piX+\left( - \frac 1 2 X( V_i ) -\frac 1 4   \Box_\g  w \right)|\Psi|^2+\frac 12  w \LL_i[\Psi] +\frac 1 4  \Db^\mu (\Psi^2 M_\mu)  
 \eea
When   $M=0$ we   simply write  $\EE[X, w](\Psi)$.

\subsubsection{Main identities for vectorfields $X=a(r) e_3+ b(r) e_4$}

 Recall that the Ricci coefficients defined with respect to the null frame $\mathscr{N}$ given by \eqref{outgoing-null-pair} have the following values:
 \bea\label{value-ka-om}
 \ka=\frac{2}{r}, \qquad \kab=-\frac{2\Up}{r}, \qquad \om=0, \qquad \omb=\frac{M}{r^2}-\frac{Q^2}{r^3}
 \eea
 In particular, we have
 \bea\label{formulas-christoffel}
 \begin{split}
\D_3 e_3 &=\left(-\frac{2M}{r^2}+\frac{2Q^2}{r^3}\right) e_3 , \quad \D_3 e_4=\left(\frac{2M}{r^2}-\frac{2Q^2}{r^3}\right) e_4, \quad  \D_A e_B = - \frac{\Up}{2r}\slashed{g}_{AB} e_4+\frac{1}{2r} \slashed{g}_{AB} e_3\\
 &\D_4 e_A =\D_4 e_4=\D_4 e_3= \D_3 e_A =0, \qquad  \D_A  e_4 = \frac 1 r  e_A, \qquad \D_A e_3=-\frac{\Up}{r} e_A.
 \end{split}
\eea
 Moreover, 
 \bea\label{derivatives-r}
 e_3(r)=-\Up, \qquad e_4(r)=1
 \eea
and
 \bea\label{Omegab'}
\Up'=\frac{2M}{r^2}-\frac{2Q^2}{r^3}
 \eea

Since we will make large use of vectorfields of the form $X= a(r) e_3+b(r) e_4$, we summarize here the general computation of its deformation tensor.

\begin{lemma}\label{general-deformation} Let $X=a(r) e_3+b(r) e_4$ a vectorfield. The component of its deformation tensor verify
\beaa
\piX_{44}&=& -4a', \qquad \piX_{34}= 2 \Up a'-2 b'+\left(\frac{4M}{r^2}-\frac{4Q^2}{r^3} \right)a, \qquad \piX_{3A}= \piX_{4A}=0,\\
\piX_{33}&=&  4\Up b'+\left(-\frac{8M}{r^2}+\frac{8Q^2}{r^3} \right)b,  \qquad \piX_{AB}= \left(-\frac{2\Up}{r}a+\frac{2}{r}b\right) \slashed{g}_{AB}
\eeaa
\end{lemma}
\begin{proof} Recall that the deformation tensor of a vector field $V$ is defined as $\,^{(V)}\pi_{\a\b}:=\mathcal{L}_V \g_{\a\b}=\D_\a V_\b+\D_\b V_\a$. Using \eqref{formulas-christoffel}, we easily compute the components of $\piX$. 
\end{proof}

 We define the following vectorfields:
\begin{itemize}
\item the Morawetz vector field $Y=f(r) R$
\item the $r^p$-hierarchy vector field $Z=l(r) e_4$
\end{itemize}

As a corollary of Lemma \ref{general-deformation}, we compute the deformation tensor of all the vectorfields we will use in the estimates. 
\begin{corollary}\label{lemma:componentspiR} The components of the deformation tensor of $T$ all vanish identically. The only components of the deformation tensor of $R$ which do not vanish are the following:
\beaa
\piR_{34}&=&-\frac{4M}{r^2}+\frac{4Q^2}{r^3} , \qquad \piR_{AB} =\frac {2\Up}{ r}\slashed{g}_{AB}
\eeaa
The components of the deformation tensor of $Y=f(r)R$  which do not vanish are the following
\beaa
\piY_{33}=2\Up^2 f' , \quad \piY_{44}= 2 f', \quad \piY_{34}= \left(-\frac{4M}{r^2}+\frac{4 Q^2}{r^3}  \right)f  -2 f' \Up , \quad \piY_{AB}=\frac{2\Up}{r}f\slashed{g}_{AB}
\eeaa
The components of the deformation tensor of $Z=l(r) e_4$ which do not vanish are the following 
\beaa
\piZ_{34}= -2 l', \qquad \piZ_{33}=  4\Up l'+\left(-\frac{8M}{r^2}+\frac{8Q^2}{r^3} \right)l, \qquad \piZ_{AB}=\frac{2}{r}l \slashed{g}_{AB}
\eeaa
\end{corollary}

In order to apply Proposition \ref{prop:qf-tensorial'}, we compute the general expression for $ \QQ\c \piX$ for $X=a(r) e_3+b(r) e_4$.

 \begin{lemma}
 \label{le:QQcpidX}
 Let $\QQ$ be the energy momentum tensor as defined in \eqref{defenergymomentum}. For $X=a(r) e_3+b(r) e_4$, we have
 \beaa
 \QQ\c \piX  &=& \left( \Up a'- b'+\left(\frac{2M}{r^2}-\frac{2Q^2}{r^3} -\frac{2\Up}{r}\right)a +\frac 2 r b\right)|\nabb \Psi|^2+ \left( \Up b'+\left(-\frac{2M}{r^2}+\frac{2Q^2}{r^3} \right)b \right)|\nabb_4 \Psi|^2 \\
 &&-a'|\nabb_3 \Psi|^2+\left(\frac{2\Up}{r}a-\frac{2}{r}b\right) \LL_i[\Psi]+\left(\Up a'- b'+\left(\frac{2M}{r^2}-\frac{2Q^2}{r^3} \right)a \right)V_i |\Psi|^2
 \eeaa
  or equivalently
 \beaa
\QQ\c \piX  &=& \left( \Up a'- b'+\left(\frac{2M}{r^2}-\frac{2Q^2}{r^3} \right)a \right)|\nabb \Psi|^2+ \left( \Up b'+\left(-\frac{2M}{r^2}+\frac{2Q^2}{r^3} \right)b \right)|\nabb_4 \Psi|^2 -a'|\nabb_3 \Psi|^2\\
 &&+\left(-\frac{2\Up}{r}a+\frac{2}{r}b\right)\nabb_3 \Psi \c \nabb_4 \Psi+\left(\Up a'- b'+\left(\frac{2M}{r^2}-\frac{2Q^2}{r^3} +\frac{2\Up}{r}\right)a -\frac{2}{r}b\right)V_i |\Psi|^2
 \eeaa
 \end{lemma} 
 \begin{proof} The computation is easily deduced using \eqref{trace-Q} and \eqref{components-Q}. 
 \end{proof}

 \begin{corollary}\label{QQcpiY} For $Y=f(r) R$, we have 
 \beaa
 \QQ \c \piY&=& 2f\left( \frac{1}{r}-\frac{3M}{r^2}+\frac{2Q^2}{r^3}  \right)|\nabb \Psi|^2+2 f'|R \Psi|^2+\left(-\frac{2\Up}{r}f-\Up f'\right) \LL_i[\Psi]\\
 &&+\left(-\frac{2M}{r^2}+\frac{2Q^2}{r^3}\right) fV_i |\Psi|^2
 \eeaa
 and for $Z=l(r) e_4$ we have
 \beaa
 \QQ \c \piZ&=& \left(- l' +\frac 2 r l\right)|\nabb \Psi|^2+ \left( \Up l'+\left(-\frac{2M}{r^2}+\frac{2Q^2}{r^3} \right)l \right)|\nabb_4 \Psi|^2-\frac{2}{r}l \LL_i[\Psi]- l' V_i |\Psi|^2
 \eeaa
 \end{corollary}
  
 We recall here an useful lemma.
  \begin{lemma}
\label{calculation:square-radial}
If  $f=f(r)$ then
\beaa
\Box_\g f(r) &=& r^{-2} \pr_r \left( r^2 \Up \pr_rf\right) 
\eeaa
\end{lemma}
\begin{proof}
Using Lemma \ref{square-RN}, and the values \eqref{value-ka-om} and \eqref{derivatives-r}, we  deduce   for a function $f=f(r)$, 
 \beaa
 \Box_\g  f &=& -\frac 1 2  (\nabb_3 (f') - \nabb_4 (\Up f')) +(\omb-\frac 1 2 \kab ) f' -(\om-\frac 1 2 \ka)  \Up f'\\
 &=& \frac 1 2  (2\Up f'' + \Up' f') +\left(\frac{M}{r^2}-\frac{Q^2}{r^3}+\frac{\Up}{r} \right) f' +\frac {  \Up}{r} f'
 \eeaa
 and using \eqref{Omegab'}, we finally obtain
  \beaa
 \Box_\g  f  &=& \Up f''  +\left(\frac{2M}{r^2}-\frac{2Q^2}{r^3}+\frac{2\Up}{r} \right) f' =r^{-2} \pr_r \left( r^2 \Up \pr_rf\right) 
 \eeaa
 as desired.
 \end{proof}

 \subsection{Morawetz estimates}

We will derive now the main identity to apply the vector field method to $Y=f(r)R$. We choose $w$ as a function of $f$.
 \begin{proposition}
  \label{prop-ident-Mor1}
 Let  $Y=f(r) R$  and   $w=  r^{-2}\Up\pr_r \left( r^2 f\right)  $.  Then for a solution $\Psi$ to the equation \eqref{general-equation}, we obtain
 \beaa
 \EE[Y, w](\Psi)&=& f\left(\frac{r^2-3M r +2 Q^2 }{r^3}\right)|\nabb \Psi|^2+ f'|R \Psi|^2+\left( -\frac 1 2 \pr_r\left(\Up V_i\right)f-\frac 1 4  r^{-2} \pr_r \left( r^2 \Up \pr_rw\right) \right)|\Psi|^2 
 \eeaa
  \end{proposition}
 \begin{proof}
By definition \eqref{le:divergPP-gen-EE} and according to    Corollary \ref{QQcpiY}  we have, 
  \beaa
 \EE[Y, w](\Psi)&=& f\left( \frac{1}{r}-\frac{3M}{r^2}+\frac{2Q^2}{r^3}  \right)|\nabb \Psi|^2+ f'|R \Psi|^2+\left(\frac 12  w -\frac{\Up}{r}f-\frac{\Up}{2} f'\right) \LL_i[\Psi]\\
 &&+\left( - \frac 1 2 Y( V_i )+\left(-\frac{M}{r^2}+\frac{Q^2}{r^3}\right) fV_i -\frac 1 4   \Box_\g  w \right)|\Psi|^2 
 \eeaa
 and observe that the choice of $w$ cancels out the term involving $\LL_i[\Psi]$. Since the potentials are scalar functions of $r$ only, writing $Y(V_i)=\Up f\pr_rV_i$, we obtain
  \beaa
 \EE[Y, w](\Psi)&=& f\left( \frac{1}{r}-\frac{3M}{r^2}+\frac{2Q^2}{r^3}  \right)|\nabb \Psi|^2+ f'|R \Psi|^2\\
 &&+\left( \left(-\frac {\Up}{2} \pr_rV_i+\left(-\frac{M}{r^2}+\frac{Q^2}{r^3}\right) V_i \right)f-\frac 1 4   \Box_\g  w \right)|\Psi|^2 
 \eeaa
 Observing that $-\frac{M}{r^2}+\frac{Q^2}{r^3}=-\frac 12 \Up'$ and using Lemma \ref{calculation:square-radial} to write $\Box_\g  w$, we get the desired expression.
 \end{proof}
 
\begin{remark} Observe that the coefficient of the first term of $\EE[fR, w](\Psi)$ is given by $\frac{r^2-3M r +2 Q^2 }{r^3}$, whose zero gives the $r$-value of photon sphere in Reissner-Nordstr{\"o}m, $r_P=\frac{3M + \sqrt{9M^2-8Q^2}}{2}$. This term has therefore a degeneracy in the trapping region.
\end{remark}

We separate the expression $\EE[Y, w](\Psi)$ into a part involving and one not involving the potential $V_i$:
\beaa
\EE[Y, w](\Psi)&=&\EE_0[Y, w](\Psi)+\EE_{V_i}[Y, w](\Psi), \\
\EE_0[Y, w](\Psi)&=& f\left(\frac{r^2-3M r +2 Q^2 }{r^3}\right)|\nabb \Psi|^2+ f'|R \Psi|^2-\frac 1 4  r^{-2} \pr_r \left( r^2 \Up \pr_rw\right)  |\Psi|^2, \\
\EE_{V_i}[Y, w](\Psi)&=&-\frac 1 2 \pr_r\left(\Up V_i\right)f |\Psi|^2
\eeaa
 The goal is   to  show that there  exist choices of $f, w$ verifying the   condition
  $w= r^{-2}\Up\pr_r \left( r^2 f\right)  $     of Proposition \ref{prop-ident-Mor1}   which make  $ \EE[Y, w](\Psi)$            positive definite.

\subsubsection{Construction of the function $f(r)$}
    
     We  first look  for choices     of $f$ and $w$ such that   the coefficients appearing in $\EE_0[Y, w](\Psi)$ are  non-negative. We immediately observe the following facts.
     \begin{enumerate}
      \item Observe that $r^2-3M r +2 Q^2 $ is non-negative for $r\geq r_P$ and non-positive for $r\leq r_P$. Therefore, the coefficient of $|\nabb \Psi|^2$ being non-negative implies that $f\leq 0$ in $r<r_P$ and $f \geq 0$ in $r>r_P$, therefore $f=0$ at $r=r_P$. 
 \item    The coefficient of $|R\Psi|^2$ being non-negative implies that the function $f$ must be  increasing  as a function of $r$. 
 \item To ensure that the coefficient of $|\Psi|^2$ is non-negative we need to choose $w$ such that the following holds: $\pr_r \left( r^2 \Up \pr_rw\right)\leq 0$.  
 \end{enumerate}
 
Once $w$ is chosen, the function $f$ can be determined by the condition $w=  r^{-2}\Up\pr_r \left( r^2 f\right)  $. In particular we   are naturally led to  define
 \bea
 \label{eq:John2}
  r^2 f= \int_{r_P} ^r  \frac{r^2}{\Up} w. 
   \eea
 Recalling that $\Up \geq 0$ in the exterior region, to ensure that 1. is satisfied we only need a non-negative $w$.  
We define $w$ based on the argument given in \cite{Stogin} and applied in \cite{stabilitySchwarzschild}. 
 \begin{lemma}
 \label{le:John2'}
 
  There exists $r_*:=2M +\sqrt{4M^2-3 Q^2} > r_P$  such   that, defining the function $w$ as
 \bea \label{definitionw}
 w=\begin{cases}
 & \frac{2}{r_*}  \Up(r_*)=: w_*  \qquad  \quad  \mbox{if}  \quad    r< r_*\\
  & \frac{2}{r}  \Up(r)=\frac{2}{r} (1-\frac{2M}{r}+\frac{Q^2}{r^2})  \qquad \qquad \qquad\mbox{if} \, \,\,\quad    r\ge r_*
 \end{cases}
 \eea
 then $w$ is               $C^1$ and  non-negative. In addition, in the subextremal Reissner-Nordstr{\"o}m spacetime for $|Q|<M$, we have $\pr_r \left( r^2 \Up \pr_rw\right)\leq 0$.
  \end{lemma}
\begin{proof}
Defining $w$ as in \eqref{definitionw}, we have that for $r> r_*$,   
   \bea
   \label{eq:sign-rV}
   \pr_r w=2\pr_r( \frac 1 r \Up)  = 2\left(- r^{-2}\Up+ r^{-1}\Up'\right)=-2\frac{r^2-4M r +3 Q^2}{r^4}
   \eea
   Since $w$ is  constant  for $r< r_*$ to ensure that   $w$  is  $C^1$ we need to  have  $\pr_r w(r_*)=0 $, i.e.   $r_*=2M +\sqrt{4M^2-3 Q^2} $. The function $w$ is clearly non-negative in the exterior region. 
   
For  $r>r_*$, 
   \beaa
  - \pr_r \left( r^2 \Up \pr_rw\right)&=&2\pr_r \left( \Up \frac{r^2-4M r +3 Q^2}{r^2}\right)= 2  \Up'\left(1-\frac{4M}{r}+\frac{3Q^2}{r^2}\right)+ 2  \Up \left(\frac{4M}{r^2}-\frac{6Q^2}{r^3}\right)\\
   &=&4\Big(\frac{3M}{r^2}-\frac{8M^2}{r^3}-\frac{4Q^2}{r^3}+\frac{15 M Q^2}{r^4}-\frac{6 Q^4}{r^5}\Big) \\
   &=&\frac{4}{r^5}\left(r\left(3M r^2 -(8M^2+4Q^2)r+\frac{32}{3}MQ^2 \right)+\frac{13}{3}MQ^2r-6 Q^4 \right).
   \eeaa
   and the two parts of the polynomial are clearly positive for $r \ge r_*$ in the subextremal range $|Q|<M$.
   \end{proof}

Once $w$ is defined as in Lemma \ref{le:John2'}, we can explicitly evaluate  $f$. For simplicity, we just evaluate $f$ for $r \ge r_*$. Denoting by $C_*$, with $C_*>0$,  the value for $r^2f$  ar $r=r_*$, i.e. $C_*=(r^2 f)(r_*)= \int_{r_P} ^{r_*}  \frac{r^2}{\Up}  w   dr\leq 2 r_* (r_*- r_P) $, we have
\beaa
 r^2f(r)&=&   \int_{r_P} ^r  \frac{r^2}{\Up}  w   dr=            ( r^2 f)(r_*)+  \int_{r_*} ^r 2 \la  d\la=r^2+ C_* -r_*^2
\eeaa  
Since $r_*^2-C_*>0$ in the subextremal range, the function is increasing for $r\ge r_*$. For $r < r_*$, using \eqref{eq:John2} written as $\pr_r(r^2 f)=\frac{r^2}{\Up} w$, we have 
\beaa
 r^2 \pr_r \left(r\frac{w}{\Up}\right)&=&r^2 \left[ r^{-1} \pr_r \left(\frac{r^2}{ \Up} w \right)- r^{-2} \frac{r^2}{ \Up} w \right]=  r  \pr_r \left(\frac{r^2}{ \Up} w \right)- \frac{r^2}{ \Up} w\\
 &=& r \pr_r\pr_r(r^2 f)-\pr_r(r^2 f) = 3 r^2 f'+ r^3 f''=\pr_r( r^3 f')
 \eeaa
Using Lemma \ref{le:John2'} and \eqref{eq:sign-rV}, we have
   \beaa
    \pr_r(r^3 f' )=r^2 \pr_r \left(r\frac{w}{\Up}\right)   = w_* r^2\pr_r\left (\frac{r}{ \Up}\right) \le 0
    \eeaa
We deduce  that  $r^3 f'$  has a minimum at $r=r_*$.       Thus, for all $r\in[r_\HH, r_*]$,
   \beaa
   r^3 f'(r)\ge    r_*^3 f'(r_*) =2(r_*^2-C_*) > 0
   \eeaa
This proves condition 2., and together with Lemma \ref{le:John2'}, this proves that $\EE_0[Y,w](\Psi)$ is positive definite.

 We  summarize the results in the following.
  
   \begin{proposition}
   \label{prop:pre-mor1}
   There exist functions $f\in C^2 , w\in C^1 $   verifying the relation    $w=  r^{-2}\Up\pr_r \left( r^2 f\right)  $ and such that for some $C=C(Q, M)>0$ and all $r\ge r_\HH$,
\beaa
 f\left(\frac{r^2-3M r +2 Q^2 }{r^3}\right)\ge   C\frac{\left(r^2-3M r +2 Q^2\right)^2 }{r^5}, \qquad   f'(r)\ge \frac{C}{r^3}, \qquad \frac 1 4  r^{-2} \pr_r \left( r^2 \Up \pr_rw\right)\leq 0
\eeaa
In particular $\EE_0[Y, w](\Psi)$ is positive definite for $r\ge r_\HH$.
   \end{proposition}
   
   \subsubsection{Improved lower bound away from the horizon}
In the previous subsection, we proved that $\EE_0[Y, w](\Psi)$ is positive definite in the exterior region. Nevertheless, we are interested in the positivity of  $\EE[Y, w](\Psi)$, so the term  $\EE_{V_i}[Y, w](\Psi)$ has to be taken into account. This term depends on the potential, and it can easily be seen that for the potentials we are considering it is not positive in general. 

We will control this term with the help of the Poincar\'e inequality, gaining positivity from the $|\nabb \Psi|^2$ term, and adding a well-chosen one form $M$ to the definition of $\EE$. We will follow the procedure used in \cite{stabilitySchwarzschild}. 

We begin by stating the Poincar\'e inequality.
     \begin{lemma}[Poincar\'e inequality \cite{Ch-Kl}]
   If $\Psi$ is a $2$-tensor on $S$, then
\bea\label{Poincare}
\int_S |\nabb \Psi|^2  &\ge & \frac{2}{r^2}\int_S  |\Psi |^2   
\eea
\end{lemma}
Using the Poincar\'e inequality, for $\mu>0$ we can write
  \beaa
   \int_Sf\left(\frac{r^2-3M r +2 Q^2 }{r^3}\right)|\nabb \Psi|^2            & \ge&          \mu  \int_S f\left(\frac{r^2-3M r +2 Q^2 }{r^3}\right) |\nabb\Psi|^2   \\
   &&+(1-\mu)\int_S  2f\left(\frac{r^2-3M r +2 Q^2 }{r^5}\right) | \Psi|^2 
   \eeaa
  According to   Proposition \ref{prop:pre-mor1},  for a sufficiently small  $\mu>0$, we  deduce
  \bea\label{bound-missing-positivity}
  \begin{split}
  \int_S\EE [Y, w](\Psi)&\ges \int_S f'|R \Psi|^2+\frac{(r^2-3M r +2 Q^2)^2 }{r^5}|\nabb \Psi|^2\\
  &+\Big(2\left(\frac{r^2-3M r +2 Q^2 }{r^5}\right)-\frac 1 2 \pr_r\left(\Up V_i\right)\Big)f |\Psi|^2 
  \end{split}
  \eea
  This procedure created extra-positivity in the $|\Psi|^2$ term, but it can be seen that with the given potentials it is not enough to get a positive definite expression in some region of the exterior.  To achieve   positivity for all values of $r \geq r_{\mathcal{H}^+}$ we shall modify the original  energy  densities  $\EE [Y, w](\Psi)$
   by considering instead  $  \EE [Y, w, M](\Psi) $,  with $M= 2 h R$ for a function $h=h(r)$. Then, by formula \eqref{le:divergPP-gen-EE}, we have
 \beaa\label{modificationM}
 \begin{split}
  \EE [Y, w, M](\Psi) &= \EE[fR, w](\Psi)+\frac 1 4  \D^\mu (\Psi^2 {M}_{\mu}) =\EE[Y, w](\Psi)+\frac 1 2 \D^\mu ( h R_\mu) |\Psi|^2 + h \Psi R(\Psi)
   \end{split}
  \eeaa
Observe that, using \eqref{R(r)} and Corollary \ref{lemma:componentspiR}:
\beaa
 \D^\mu (h R_\mu) &=& R(h) + h( \D^\mu  R_\mu)=\Up h' + \frac 1 2 h\tr \,(\piR)=\Up h' +  h\left( \frac{2M}{r^2}-\frac{2Q^2}{r^3}+\frac {2\Up}{ r}  \right)= r^{-2} \pr_r(\Up r^2 h)
 \eeaa
therefore
 \bea\label{modificationM-1}
 \begin{split}
  \EE [Y, w, 2hR](\Psi) &=\EE[Y, w](\Psi)+\frac 1 2 r^{-2} \pr_r(\Up r^2 h) |\Psi|^2 + h \Psi R(\Psi)
   \end{split}
  \eea

  We present below the construction of the function $h$ that will provide the improved lower bound away from the horizon.

  \begin{proposition}
  \label{prop:pre-mor2} Let $\Psi$ be a $2$-tensor solution of \eqref{general-equation}. Suppose that the potential $V_i$ satisfies the following:
  \begin{enumerate}
  \item For $r\le r_P$, then $2\left(\frac{r^2-3M r +2 Q^2 }{r^5}\right)-\frac 1 2 \pr_r(\Up V_i) \leq 0 $ , 
 \item For $r \geq r_P$, then $A_i:=-\frac 12 \pr_r(\Up V_i)-\frac{4\Up}{r^5}\left(r^2-4Mr+3Q^2 \right) \geq 0$ 
  \end{enumerate}
 Then there exists    a function $h=h(r) $    with bounded   derivative $h'$, which is   supported   in $r\ge r_P$,  such that $h= O(r^{-2}) , h'=O(r^{-3})$  for $r\ge r_*$ and such that for $r\geq \frac 5 6  r_P$
\beaa
 \int_S   \EE [Y, w, 2hR](\Psi) \ges \int_{S} \left( \frac{1}{r^3} |R(\Psi)|^2+\frac{\left(r^2-3M r +2Q^2\right)^2}{r^5}|\nabb \Psi|^2  +  \frac{\Up}{r^4} |\Psi|^2    \right) 
  \eeaa
  \end{proposition}
  \begin{proof} Putting together \eqref{bound-missing-positivity} and \eqref{modificationM-1}, we have 
  \beaa
&& \int_S \EE [Y, w, 2hR](\Psi) \\
&\ges& \int_S f'|R \Psi|^2+\frac{(r^2-3M r +2 Q^2)^2 }{r^5}|\nabb \Psi|^2\\
 &&+\int_S\Big(2\left(\frac{r^2-3M r +2 Q^2 }{r^5}\right)+\frac 1 2 \pr_r\left(\Omegab V_i\right)\Big)f |\Psi|^2 -\frac 1 2 r^{-2} \pr_r(\Omegab r^2 h) |\Psi|^2 + h \Psi R(\Psi)
  \eeaa
where we write $\Omegab=-\Up$. We now write the function $h$ as $h=-\frac{4}{r^4}\Omegab \widetilde{h}$, hence
\beaa
-\frac 1 2 r^{-2} \pr_r(\Omegab r^2 h) |\Psi|^2 + h \Psi R(\Psi) &=&2 r^{-2} \pr_r\left( \frac{\Omegab^2}{r^2} \widetilde{h}\right) |\Psi|^2 -\frac{4}{r^4}\Omegab \widetilde{h} \Psi R(\Psi)\\
&=&\left( \Omegab'- \frac{ \Omegab}{r}\right)\frac{4\Omegab}{r^4} \widetilde{h} |\Psi|^2+2  \frac{\Omegab^2}{r^4} \widetilde{h}' |\Psi|^2 -\frac{4}{r^4}\Omegab \widetilde{h} \Psi R(\Psi)
\eeaa
and observe that $ \Omegab'- \frac{ \Omegab}{r}=\frac{1}{r^3}\left(r^2-4Mr+3Q^2 \right)$. We also express
\beaa
-\frac{4}{r^4}\Omegab \widetilde{h} \Psi R(\Psi)&=&\frac{2\widetilde{h}}{r^3}(R(\Psi)- \Omegab r^{-1} \Psi)^2 - \frac{2 \widetilde{h}}{r^3} |R(\Psi)|^2 -
 \frac{2 \widetilde{h}}{r^5}\Omegab^ 2 |\Psi|^2
 \eeaa
 therefore we deduce
   \beaa
&& \int_S \EE [Y, w, 2hR](\Psi) \\
&\ges& \int_S \left(f'- \frac{2 \widetilde{h}}{r^3}\right)|R \Psi|^2+\frac{(r^2-3M r +2 Q^2)^2 }{r^5}|\nabb \Psi|^2+\frac{2\widetilde{h}}{r^3}(R(\Psi)- \Omegab r^{-1} \Psi)^2\\
 &&+\int_S \Big(2\left(\frac{r^2-3M r +2 Q^2 }{r^5}\right)f+\frac 1 2 \pr_r(\Omegab V_i) f+  \frac{4\Omegab}{r^5}\left(r^2-4Mr+3Q^2 \right)\frac{\widetilde{h}}{r^2} +2  \frac{\Omegab^2}{r^4} \widetilde{h}' - \frac{2 \widetilde{h}}{r^5}\Omegab^ 2\Big)|\Psi|^2
  \eeaa
  We write $\frac 12 \pr_r(\Omegab V_i)=-\frac{4\Omegab}{r^5}\left(r^2-4Mr+3Q^2 \right)+A_i$, and we write the coefficient on the right hand side as 
  \beaa
  2\left(\frac{r^2-3M r +2 Q^2 }{r^5}\right)f-\frac{4\Omegab}{r^5}\left(r^2-4Mr+3Q^2 \right)\left(f-\frac{\widetilde{h}}{r^2}\right)+A_if +2  \frac{\Omegab^2}{r^4} \widetilde{h}' - \frac{2 \widetilde{h}}{r^5}\Omegab^ 2
  \eeaa
We choose
\beaa
  \widetilde{h}=\begin{cases}
\quad   0\qquad \qquad \qquad   r\le r_P\\
  \quad r^2 f \qquad\qquad  r_P\le r \le r_*\\
  (r_*)^2 f(r_*)\qquad\,  r \ge r_*
\end{cases}
\eeaa
Since $f\ge 0$ for $r\ge r_P$, we have that $\widetilde{h}\ge 0$ everywhere, therefore in the estimates we can ignore   the term $\frac{2\widetilde{h}}{r^3}(R(\Psi)- \Omegab r^{-1} \Psi)^2 $, and we obtain
  \beaa
&& \int_S \EE [Y, w, 2hR](\Psi) \ges \int_S \left(f'- \frac{2 \widetilde{h}}{r^3}\right)|R \Psi|^2+\frac{(r^2-3M r +2 Q^2)^2 }{r^5}|\nabb \Psi|^2\\
 &&+\int_S \Big( 2\left(\frac{r^2-3M r +2 Q^2 }{r^5}\right)f-\frac{4\Omegab}{r^5}\left(r^2-4Mr+3Q^2 \right)\left(f-\frac{\widetilde{h}}{r^2}\right)+A_if +  \frac{2\Omegab^2}{r^4}\left( \widetilde{h}' - \frac{ \widetilde{h}}{r}\right)\Big)|\Psi|^2
  \eeaa

 We consider the following cases.
 
 \bigskip

{\bf Case 1}($r\le r_P$).     Since $\widetilde{h}=0$, we  obtain
 \beaa
 \int_S \EE [Y, w, 2hR](\Psi) &\ges& \int_S f'|R \Psi|^2+\frac{(r^2-3M r +2 Q^2)^2 }{r^5}|\nabb \Psi|^2\\
 &&+\int_S \Big(2\left(\frac{r^2-3M r +2 Q^2 }{r^5}\right)-\frac{4\Omegab}{r^5}\left(r^2-4Mr+3Q^2 \right)+A_i\Big)f|\Psi|^2
  \eeaa
Recall that for $r\le r_P$, we have $f \leq 0$ and $f'\geq 0$. Therefore if the potential $V_i$ is such that $2\left(\frac{r^2-3M r +2 Q^2 }{r^5}\right)+\frac 1 2 \pr_r(\Omegab V_i) \leq 0 $ for $r\le r_P$, we get positivity. 

\bigskip

{\bf Case 2}( $ r_P\le r    \le r_*$).  Since $\widetilde{h}=r^2 f$ and $\widetilde{h}' - \frac{ \widetilde{h}}{r}=r^2 f'+rf$, we deduce 
 \beaa
 \int_S \EE [Y, w, 2hR](\Psi) &\ges& \int_S \left(f'- \frac{2 f}{r}\right)|R \Psi|^2+\frac{(r^2-3M r +2 Q^2)^2 }{r^5}|\nabb \Psi|^2\\
 &&+\int_S \Big( \left(2\left(\frac{r^2-3M r +2 Q^2 }{r^5}\right)+A_i  +  \frac{2\Omegab^2}{r^3}\right)f+  \frac{2\Omegab^2}{r^2} f'\Big)|\Psi|^2
  \eeaa
  Recall that for $r\geq r_P$, we have $f \geq 0$ and $f'\geq 0$. Therefore if the potential $V_i$ is such that $2\left(\frac{r^2-3M r +2 Q^2 }{r^5}\right)+A_i  +  \frac{2\Omegab^2}{r^3}\geq 0 $ for $r_P\le r    \le r_*$, we get positivity of the last term. In particular, this is true if $A_i \geq 0$ for $r_P\le r    \le r_*$. It  remains to analize the first term.
\begin{lemma}\label{lemmafirst}
In the interval $[r_P, r_*]$ we have,
\beaa
 f'-\frac{2}{r} f >0.
\eeaa
\end{lemma}
\begin{proof}
Since $f$ vanishes at $r=r_P$ we have,
\beaa
f'-\frac{2}{r} f &=& r^{-2} (r^2 f)' - 4 r^{-3}(r^2 f)= r^{-2} (r^2 f)' -4 r^{-3}\int_{r_P}^ r (r^2 f)' 
\eeaa
Recall that, by Lemma \ref{le:John2'}, $w=-r^{-2}\Omegab (r^2 f)'=w_*$   is a positive constant for all $r\leq r_*$, therefore in the interval $[r_P, r_*]$. We deduce,
\beaa
(r^2 f)'=-\frac{r^2 }{\Omegab}w_*
\eeaa
Using that, for $r>r_P$, $\frac{r^2 }{\Omegab}$ is a decreasing function, we have the bound
\beaa
f'-\frac{2}{r} f &=&- \frac{1}{\Omegab} w_* +4 r^{-3}\int_{r_P}^ r \frac{r^2 }{\Omegab} w_* =w_* \left(-\frac{1}{\Omegab} +4 r^{-3}\int_{r_P}^ r \frac{r^2 }{\Omegab} \right)\\
&>& w_* \left(-\frac{1}{\Omegab} +4 r^{-3}(r-r_P) \frac{r^2 }{\Omegab}\right)= \frac{w_*}{-\Omegab} \left(1 -\frac 4 r (r-r_P) \right)=\frac{w_*}{-\Omegab r} \left(-3 r +4 r_P\right)
\eeaa
Defining  $r_*^c:=2M +\frac 2 3 \sqrt{9M^2-8e^2}<r_*$, we have that if $r<r_*^c$, then $\left(-3 r +4 r_P\right)>0$ and therefore $f'-\frac{2}{r} f>0$.

Consider now $f$ for $r>r_*$, i.e. $f(r)=1+\frac{C_*-r_*^2}{r^2}$ and $f'(r)=\frac{2(r_*^2-C_*)}{r^3}$. Since $f$ is a $C^2$ function we can compute 
\beaa
(f'-\frac{2}{r} f)(r_*)&=&\frac{2}{r_*} \left(1-\frac{2C_*}{r_*^2} \right)=\frac{2}{r_*^3} \left(r_*^2-2C_*\right) 
\eeaa
Consider $D_*(Q)=r_*(Q)^2-2C_*(Q)$ as a function of $Q$. Clearly $D_*(Q)$ is a continuous function of $Q$ and for $Q=0$, we have that $D(0)\approx 1.56 M^2 >M^2$. Therefore, for $Q$ small enough $D_*(Q) > \frac 1 2 M^2>0$. This implies that, for such $Q$, $(f'-\frac{2}{r} f)(r_*)>0$. 

Consider the derivative of $f'-\frac{2}{r} f =w_* \left(-\frac{1}{\Omegab} +4 r^{-3}\int_{r_P}^ r \frac{r^2 }{\Omegab} \right)$. We have 
\beaa
\left(f'-\frac{2}{r} f\right)'&=& w_* \left(\frac{\Omegab'}{\Omegab^2} -12 r^{-4}\int_{r_P}^ r \frac{r^2 }{\Omegab}+4r^{-3} \frac{r^2 }{\Omegab} \right)\\
&<& w_* \left(\frac{\Omegab'}{\Omegab^2} -12 r^{-4}(r-r_P) \frac{r^2 }{\Omegab}+4r^{-3} \frac{r^2 }{\Omegab} \right)= w_* \left(\frac{\Omegab'}{\Omegab^2} -\frac{4(2r-3r_P)}{r^2\Omegab} \right)
\eeaa
We have that $\Omegab'<0$ for all $r>r_\HH$ and if $r<r_* <\frac 3 2 r_P$, we obtain $\left(f'-\frac{2}{r} f\right)'<0$ in the interval $[r_P, r_*]$. This, together with $(f'-\frac{2}{r} f)(r_*)>0$, implies that $(f'-\frac{2}{r} f)(r)>0$ in the whole interval, as desired.
\end{proof}

\bigskip

{\bf Case 3}( $  r \ge r_*$)  Since $\widetilde{h} $ is constant  we deduce 
  \beaa
&& \int_S \EE [Y, w, 2hR](\Psi) \ges \int_S \left(f'- \frac{2 \widetilde{h}}{r^3}\right)|R \Psi|^2+\frac{(r^2-3M r +2 Q^2)^2 }{r^5}|\nabb \Psi|^2\\
 &&+\int_S \Big[ \Big(2\left(\frac{r^2-3M r +2 Q^2 }{r^5}\right)-\frac{4\Omegab}{r^5}\left(r^2-4Mr+3Q^2 \right)\Big)\left(f(r)-\frac{(r_*)^2 }{r^2} f(r_*)\right)\\
 &&+A_if+\Big(2\left(\frac{r^2-3M r +2 Q^2 }{r^2}\right) - 2\Omegab^2\Big)  \frac{ \widetilde{h}}{r^5}\Big]|\Psi|^2
  \eeaa
We examine the first term recalling that  that  for $r\ge r_*$ we have $f=1+\frac{C_*-r_*}{r^2}=1-c_1 r^{-2}$, for $c_1>0$. Hence, 
\beaa
f'-2 r^{-3} \widetilde{h} =2  r^{-3}( c_1 - \widetilde{h})
\eeaa
Since   $( c_1 - \widetilde{h})$  is constant it follows  that the  sign of  $(f'-2 r^{-3} \widetilde{h})$ does not change  for $r\ge r_*$.  Thus, since it was positive   for $r\le r_*$, and recall that $w$ and $\widetilde{h}$ are continuous functions,  it must remain strictly  positive for $r\ge r_*$. 

Observe that for $r>r_*$, from the definition of $r_*$ and the fact that $f$ is increasing, we have that $f(r)-\frac{(r_*)^2 }{r^2} f(r_*)\geq 0$. Moreover, the coefficient $2\left(\frac{r^2-3M r +2 Q^2 }{r^5}\right)-\frac{4\Omegab}{r^5}\left(r^2-4Mr+3Q^2 \right) \geq 0$ for $r\geq r_*$.  The last term is given by 
\beaa
2\left(\frac{r^2-3M r +2 Q^2 }{r^2}\right) - 2\Up^2&=&\frac{2}{r^4}\left( Mr(r^2-4Mr+3Q^2)+Q^2(Mr-Q^4)\right) 
\eeaa
which is positive for $r\geq r_*$. Therefore, if $A_i \geq 0$ for $r \geq r_*$, we obtain positivity of this term. 

This proves the proposition.
  \end{proof} 
  
  \begin{corollary} Let $\Psi=\qf$ or $\Psi=\qf^\F$ be solutions of the system \eqref{finalsystem}. Then for $r\geq \frac 5 6 r_P$
\beaa
 \int_S   \EE [Y, w, 2hR](\Psi) \ges \int_{S} \left( \frac{1}{r^3} |R(\Psi)|^2+\frac{\left(r^2-3M r +2Q^2\right)^2}{r^5}|\nabb \Psi|^2  +  \frac{\Up}{r^4} |\Psi|^2    \right) 
  \eeaa
  with $Y$, $w$ and $h$ as defined above.
  \end{corollary}
  \begin{proof} It suffices to show that the potentials $V_1$ and $V_2$ verify the hypothesis of Proposition \ref{prop:pre-mor2}.
  
  The potential $V_1$ verifies
  \beaa
  \frac 1 2 \pr_r(\Omegab V_1)&=& -\frac{4\Omegab}{r^5}\left(r^2-4Mr+3Q^2\right)+\frac{5Q^2}{r^4}\left( \Omegab'-\frac{4\Omegab}{r}\right)
  \eeaa
  giving $A_1=\frac{5Q^2}{r^4}\left( \Omegab'-\frac{4\Omegab}{r}\right)=\frac{10Q^2}{r^7}\left( 2r^2-5Mr+3Q^2\right)\geq 0$ for $r\geq r_P$, which verifies the second condition. To verify the first condition we have 
  \beaa
  2\left(\frac{r^2-3M r +2 Q^2 }{r^5}\right)+\frac 1 2 \pr_r(\Omegab V_1)&=&\frac{1}{r^7}\left( 6 r^4-30Mr^3+(32M^2+40Q^2)r^2-90MQ^2r+42Q^4\right)
  \eeaa
  which is negative for $r \leq r_P$.
  
  The potential $V_2$ verifies
   \beaa
  \frac 1 2 \pr_r(\Omegab V_2)&=& -\frac{\Omegab}{r^4}\left(4r-7M\right)-\frac{6}{r^7}\left( Mr-Q^2\right)^2
  \eeaa
  giving $A_2=\frac{1}{r^7}\left(9Mr^3-24M^2r^2-12Q^2r^2+45MQ^2 r-18Q^4 \right)$ which is positive for $r\geq r_P$. To verify the first condition we have 
  \beaa
  2\left(\frac{r^2-3M r +2 Q^2 }{r^5}\right)+\frac 1 2 \pr_r(\Omegab V_2)&=&\frac{1}{r^7}\left( 6 r^4-21Mr^3+(8M^2+8Q^2)r^2+5MQ^2r-6Q^4\right)
  \eeaa
  which is negative for $r \leq r_P$.
  \end{proof}

\subsubsection{Correction near the horizon}
The main problem about the estimate obtained in Proposition \ref{prop:pre-mor2} is that $f$ blows up     logarithmically  near  $r=r_\HH$. To correct for  this we can modify  our choice of $f$ and $w$, introducing a cut-off for $f$ and $w$,  in the same way as in \cite{stabilitySchwarzschild}.  

\begin{proposition}
\label{prop:pre-mor3}
Let $\Psi=\qf$ or $\Psi=\qf^\F$ be solutions of the system \eqref{finalsystem}. Given a small parameter    $\de>0$   there exist   functions  $f_\de\in C^2(r>0)  $, $w_\de\in C^1(r>0) $   and  $h\in C^2(r>0)$  
 verifying,   for all  $r>0$,
 \beaa
 |f_\de(r)|\les \de^{-1}   r^{-2}, \qquad w_\de \les  r^{-1}, \qquad h\les  r^{-4}
 \eeaa
 and a constant $ C>0$ (independent of $\de$)  such that in the exterior region
  \beaa
  \int_S   \EE [f_\de R, w_\de, 2hR](\Psi) &\ges& C^{-1} \int_{S}  \frac{1}{r^3} |R(\Psi)|^2+\frac{\left(r^2-3M r +2Q^2\right)^2}{r^5}\left( |\nabb \Psi|^2+\frac{1}{r^2} |T\Psi|^2 \right)  +  \frac{\Up}{r^4} |\Psi|^2    \\
  &-&\int_S\overline{W}_\de|\Psi|^2
 \eeaa
 where $\overline{W}_\de$ is a  non-negative  potential  supported  in the region $  r\le \frac 5 6 r_P$.
\end{proposition}
\begin{proof} We sketch here the proof, which appears in \cite{stabilitySchwarzschild}.
 Introducing $ u:=r^2 f$, we have
 \beaa
  f= r^{-2} u , \qquad w=r^{-2}\Up \partial_ru
 \eeaa
For a given $\de >0  $ we define
\beaa
u_\de  &:=& - \frac{M^2}{\de}   F(-\frac{\de}{M^2}   u), \qquad    f_\de := r^{-2} u_\de\\
w_\de &:=&r^{-2}\Up\pr_r  u_\de  = r^{-2}\Up F' (-\frac{\de}{M^2}  u) \pr_r u= w F' (-\frac{\de}{M^2} u)
\eeaa
 where  $F:\mathbb{R}\to\mathbb{R}$ is   is a fixed increasing  smooth function   such that 
  \beaa
  F(x)=
  \begin{cases}  
  &x\qquad  \mbox{for}\quad  x\le 1\\
&2 \qquad  \mbox{for}\quad  x\ge 3
  \end{cases}
 \eeaa
 Hence, for    sufficiently small $\de>0$ and constant $C>0$,
 \beaa
 f_\de =\begin{cases}
 &\frac{-2M^2}{\de r^2} \qquad \, \,\mbox{for}  \qquad      0\le  \Up \le   C    e^{-\frac{3}{\de}}\\
 &  f   \qquad \qquad \mbox{for}  \qquad      C^{-1}   e^{-\frac{1}{\de}} \le  \Up
 \end{cases}
 \eeaa
We evaluate $W_\de = -\frac 1 4 \pr_r \left(  r^2 \Up \pr_r  w_\de \right)$. We deduce
\beaa
W_\de = -\frac 1 4  \pr_r \left(  r^2 \Up \pr_r  w \right)=W \ge  \frac{M}{r^4}, \qquad   \mbox{for} \quad r\ge  r_P.
\eeaa
We introduce $\overline{W}_\de= 1_{  r_\HH\le r\le \frac 5 6 r_P }| W_\de|$. Applying Proposition \ref{prop:pre-mor2} using   the functions $f_\de$ and $w_\de $   defined  above we obtain the desired estimates, without the $T$ derivative.  To get the improved estimate   we then  modify slightly the definition of $w$, i.e. we set 
 $w= w_0- \de'    w'$   for a small  parameter  $\de'$   where $w_0$  is    our previous choice and  
\beaa
w'=2\frac{D}{r^2}     \Up  \frac{(r^2-3M r +2Q^2)^2}{r^5}  \ze(r)
\eeaa
where $\ze$ is a  suitable smooth, non-negative,  function of $r$    vanishing in a neighborhood  of  $r=r_P$  and equal to  the unity
 outside a slightly larger neighborhood. Evaluate using \eqref{le:divergPP-gen-EE},
 \beaa
\EE_1[ f_\de R, w_0-\de' \Up w', 2hR]  (\Psi)&=& \EE[ f_\de R, w_0, 2hR]-\frac 1 2 \de'   w'  \LL(\Psi)+\de'  |\Psi^2 | \frac 1 4 \square_\g( w')
\eeaa
This allows the creation of the $T$ derivative in the degenerate bulk. 
\end{proof}

\subsection{Red-shift estimate}
 Observe that the  vectorfields    $T$ and $R$  become both proportional to $e_3$ when $\Up=0$, i.e. the estimate   of Proposition \ref{prop:pre-mor3}
    degenerates  along   the horizon.
    Here we make use of  the       Dafermos-Rodnianski        red shift vectorfield 
   to     compensate for this degeneracy.   
   
   \begin{lemma} Let $N= a(r) \nabb_3+b(r) \nabb_4$ be a vectorfield. If the functions $a(r)$ and $b(r)$ verify
\beaa
a(r_\HH)=0, \qquad   b(r_\HH)=-1, \qquad -a'(r_\HH)\ge \frac{1}{4M},\qquad -b'(r_\HH)\ge \frac{1}{8M}
\eeaa
then, along the horizon, we have
  \bea
  \label{eq:red-shift1}
  \bsplit
   \QQ^{\a\b}\piN_{\a\b} &\ge  \frac{1}{8M} (  | \nabb_3\Psi|^2+ |\nabb_4\Psi|^2+|\nabb \Psi|^2)
      \end{split}
     \eea
    and
    \bea
    \bsplit
     \label{eq:red-shiftEE1}
    \EE[N,  0](\Psi)&\ge  \frac{1}{16M} \left(  | \nabb_3\Psi|^2+ |\nabb_4\Psi|^2+|\nabb \Psi|^2   +\frac{1}{M^2}|\Psi|^2        \right) 
    \end{split}
    \eea 
    \end{lemma}
\begin{proof}
Note  that on the horizon we have $\Up=0$. Using the second formula of Lemma \ref{le:QQcpidX}, we have 
 \beaa
\QQ\c \piN  &=& \left(- b'+\left(\frac{2M}{r^2}-\frac{2Q^2}{r^3} \right)a \right)|\nabb \Psi|^2+ \left( \left(-\frac{2M}{r^2}+\frac{2Q^2}{r^3} \right)b \right)|\nabb_4 \Psi|^2 -a'|\nabb_3 \Psi|^2\\
 &&+\left(\frac{2}{r}b\right)\nabb_3 \Psi \c \nabb_4 \Psi+\left(- b'+\left(\frac{2M}{r^2}-\frac{2Q^2}{r^3}\right)a -\frac{2}{r}b\right)V_i |\Psi|^2
 \eeaa
 Hence, for $a(r_\HH)=0,\, b(r_\HH)=-1$ we have
  \beaa
\QQ\c \piN  &=& \left(- b' \right)|\nabb \Psi|^2+ \left(\frac{2M}{r^2}-\frac{2Q^2}{r^3} \right)|\nabb_4 \Psi|^2 -a'|\nabb_3 \Psi|^2-\frac{2}{r}\nabb_3 \Psi \c \nabb_4 \Psi+\left(- b' +\frac{2}{r}\right)V_i |\Psi|^2\\
&\ge& \frac{1}{ 4M}  |\nabb_3(\Psi)|^2+\frac{1}{4M}|\nabb_4(\Psi)|^2+\frac{1}{8M} |\nabb \Psi|^2-\frac 2 r e_3\Psi \c e_4\Psi
 \eeaa
 since $V_i(r_\HH) \geq 0$ for $i=1,2$. The desired lower bound        in \eqref{eq:red-shift1}  easily follows. 
 
 Using \eqref{le:divergPP-gen-EE}, we have along the horizon
   \beaa
   \EE[N, 0](\Psi)&=& \frac 1 2 \QQ  \c\piN - \frac 1 2 N( V_i ) |\Psi|^2  \\
   &=&  \frac{1}{16 M} (  | \nabb_3\Psi|^2+ |\nabb_4\Psi|^2+|\nabb \Psi|^2)+\left(\frac {2}{r} V_i - \frac 1 2 N( V_i )\right) |\Psi|^2  
 \eeaa
 and since on the horizon $\frac {2}{r} V_i - \frac 1 2 N( V_i )\geq \frac{1}{10 M^3} $ we have the desired bound.
 \end{proof}
 Note that the vectorfield  
     $ N_{(0)}= a e_3 +b e_4  +2T $ with $a=-1-\frac{1}{4M}(r-r_\HH)$, $b=-1- \frac{1}{8M}(r-r_\HH) $  verify the conditions of the second part of the     previous  lemma.
By picking a  positive   bump function $\ka=\ka(r) $, supported in the region  in $[-2, 2]$ and   equal to $1$   for $[-1,1]$  and   define, for sufficiently small  $\de_\HH>0$.
     \beaa
     N_\HH &=&    \ka_\HH      N_{(0)}, \qquad  \ka_\HH:= \ka(\frac{\Up}{\de_\HH})
     \eeaa
then   $ N_\HH         $   is a   smooth vectorfield   supported in the region 
$|\Up|\le 2 \de_\HH$   such  that    the following estimate holds for both equations 
            \beaa
       \EE[N_\HH, 0](\Psi)&\ge & \frac{1}{20M}  1_{|\Up|\le \de_\HH}   \left( |\nabb_3\Psi|^2 +|\nabb_4 \Psi|^2   + |\nabb\Psi|^2+\frac{1}{M^2}|\Psi|^2\right) \\
       &-&    \frac 1 M \de_\HH^{-1}  1_{\de_\HH\le |\Up| \le 2 \de_\HH}     \left( |\nabb_3\Psi|^2 +|\nabb_4 \Psi|^2   + |\nabb \Psi|^2 +\frac{1}{M^2}|\Psi|^2 \right)
       \eeaa
As done in \cite{stabilitySchwarzschild}, we consider the combined Morawetz triplet  
   \beaa
   (Y, w,  M ):= (f_\de R, w_\de,  2hR )+\ep_\HH( N_\HH, \, 0, 0 ), 
   \eeaa
    with $\ep_\HH>0$ 
sufficiently small to be determined later.  Here   $(f_\de R, w_\de, 2hR )$   is the triplet   given by Proposition  \ref{prop:pre-mor3}.
To combine  the estimates in the entire region  we use the modified vectorfields $\Rbrev, \Tbrev$ in the   horizon region defined by \eqref{eq:Rc-Tc}. 
 Then, for sufficiently small $ \de_\HH$, in the region  $-\de_\HH \le |\Up|$, \, $r\le \frac{5}{6} r_P$, we have with a constant  $\La_\HH^{-1} :=C^{-1} \de_\HH^4>0$
      \beaa
 \int_S  ( \EE_\de+\ep_\HH \EE_\HH)& \ge&    M^{-1}\La_\HH^{-1}\int_S\left( |\Rbrev(\Psi)|^2 +|\Tbrev \Psi|^2    +|\nabb\Psi|^2  + M^{-2} |\Psi|^2\right)  - \int_S\overline{W}_\de|\Psi|^2
     \eeaa

Finally, we can eliminate the potential $\overline{W}_\de$  by a procedure  analogous to that used in \cite{stabilitySchwarzschild}, by adding an additional well-chosen one-form $(0, 0, 2 h_2 \Rbrev) $.    

\subsection{Energy estimates}\label{subsection-energy-estimates}
We add to the previous Morawetz estimates the energy estimates obtained by using the vectorfield $T$. By Corollary \ref{lemma:componentspiR}, $\piT=0$, and since $T(V)= V' T(r)=0$, we have $\EE[T,0]=0.$ We can then consider the triplet 
  \bea\label{definition-X}
   (X, w, M):=  (\widetilde{X}, w, M) + 2\Lambda(T,0,0)  
  \eea
for $\Lambda$ big enough. We apply the divergence theorem to $\PP_\mu^{(X,  w, M)}$ in the spacetime region bounded by $\Sigma_0$ and $\Sigma_\tau$ in the exterior region. Recalling \eqref{eq:modified-div}, by divergence theorem we have:
    \bea
    \label{eq:boundariesMM}
    \bsplit
             & \int_{\Si_\tau}\PP\c n_\Si+\int_{\mathcal{H}^+(0, \tau)} \PP \c e_4+\int_{\mathcal{I}^+(0, \tau)} \PP \c e_3  +\int_{\MM(0, \tau)} \EE       \\
             & = \int_{\Si_0}\PP\c n_\Si- \int_{\MM(0, \tau)} ( X(\Psi) +\frac 1 2  w \Psi) \c \M[ \Psi ]
              \end{split}
              \eea
              where      $\EE=  \EE[X,  w,  M] (\Psi)$.
                           Thus in view of the above estimates   we have,  for  $C\gg \de^{-1}$,       \bea
       \label{equation:LBforEE}
        \int_{\MM(0, \tau)}  \EE&\geq C^{-1}  \left(\Mor[\Psi]( 0, \tau)+\Mor_{\mathcal{H}^+}[\Psi]( 0, \tau)\right)
       \eea
    We now analyze the boundary terms in \eqref{eq:boundariesMM}.
           \subsubsection{Boundary term  along  the horizon}    Along  the horizon we have  $w, h, h_2=0 $. Hence,
      \beaa
      \PP\c   e_4 &=& \QQ(\widetilde{X}+2\La T, e_4)=\QQ(f_\de R+ \ep_\HH N_\HH+2\La T, e_4)=\QQ(-\frac 1 2 f_\de e_3+ \ep_\HH e_4+\La e_3 , e_4) \\
      &=&\ep_\HH |\nabb_4 \Psi|^2+\left(\Lambda - \frac 1 2 f_\de\right)\QQ_{34}
            \eeaa
 and choosing $\La$ such that $\Lambda - \frac 1 2 f_\de\geq 0$, we obtain positivity, and in particular :         
           \beaa
    \int_{\mathcal{H}^+(0, \tau)}  \PP\c   e_4\ges E_{\mathcal{H}^+}[\Psi](0, \tau). 
      \eeaa

       \subsubsection{ Boundary terms along  null infinity}
  Along null infinity, we compute
  \beaa
  \PP \c e_3&\ge &  \frac{\La}{4}\left(  \QQ_{33}   +  \QQ_{34} \right) =  \frac{\La}{4}\left( |e_3\Psi|^2 +|\nabb\Psi|^2+ V |\Psi|^2\right)
  \eeaa
  Therefore
       \beaa
     \int_{\mathcal{I}^+(0, \tau)} \PP \c e_3  &\ges &   E_{\mathcal{I}^+, 0}[\Psi](0, \tau)
      \eeaa

      \subsubsection{ Boundary terms along $\Si_{0}, \Si_{\tau}$}
        Along   a hypersurface $\Si_\tau$ with timelike  unit  future normal  $n_{\Sigma}= a e_3 + b e_4 $,    we have,       
      \beaa
      \PP\c n_{\Sigma}&=&\QQ( \widetilde{X}+2\La T, n_\Si)+\frac 1 2 w \Psi n_\Si (\Psi) -\frac 1  4  n_\Si (w) \Psi^2 + \frac 1 2        n_\Si \c(  h R+ h_2 \Rbrev )        \Psi^2
      \eeaa 
      We can consider two regions.
      \begin{enumerate}
        \item       In the  region  where   $|\Up| \le \de_\HH$        we  have   
          $w=h=0$  and  $ h_2=O(\de)$.    Thus 
      \beaa
        \PP\c n_\Si&=&\QQ( X+\La T, n_\Si)+ \frac 1 2 h_2    n\c \Rbrev         \Psi^2 \\
        &\ge & \frac 1 2 \ep_\HH\left(   a  \QQ_{44} + b\QQ_{33}+ (a+b) \QQ_{34} \right) - O(\de)        \Psi^2\\
        &\ge&  \frac 1 2 \ep_\HH\left( a | e_3\Psi|^2 + b |e_4|\Psi|^3 + (a+b)( |\nabb\Psi|^2 + V|\Psi|^2) \right)  - O(\de)        \Psi^2\\
      \eeaa 
      Recalling  the Poincar\'e inequality \eqref{Poincare},
and  the fact that  $\de$ is much smaller  that $\ep_\HH$ we deduce, in this region,
\beaa
 \int_{\Si(\tau)}   \PP\c n_\Si \ge  \frac 1 8 \ep_\HH   E[\Psi](\tau)
\eeaa

\item  In the region  $|\Up| \geq \de_\HH $ we have $w = O(r^{-1})$   thus 
     \beaa
         \PP\c n&\ge & \frac 1 4  \La         \left(   a  \QQ_{44} + b\QQ_{33}+ (a+b) \QQ_{34} \right)- 
         O( r^{-2} |\Psi|^2)\\
         &=& \frac 1 4 \La         \left( a | e_3\Psi|^2 + b |e_4 \Psi|^3 + (a+b)( |\nabb\Psi|^2 + V |\Psi|^2) \right)  -O(r^{-2} |\Psi|^2) 
         \eeaa
         and since $V_i\ge\frac{4\de_\HH}{r^2}$ for $i=1,2$, we have  for $\La \de_\HH$   sufficiently large, in this region
\beaa
 \int_{\Si(\tau)}  \PP\c n\ge   E[\Psi](\tau)
\eeaa
\end{enumerate}

 From \eqref{eq:boundariesMM}, we therefore deduce the estimate     
         \bea\label{Morawetz-1}
         \begin{split}
     & E[\Psi](\tau)+E_{\mathcal{H}^+}[\Psi](0, \tau)+E_{\mathcal{I}^+, 0}[\Psi](0, \tau)+ \Mor[\Psi](0, \tau)+\Mor_{\mathcal{H}^+}[\Psi]( 0, \tau)\\
      &\les E[\Psi](0)-\int_{\MM(0, \tau)} ( X(\Psi) +\frac 1 2  w \Psi) \c \M_i[ \Psi ]
      \end{split}
      \eea

 \subsection{Improved  Morawetz estimate}
In what follows we derive an improved Morawetz estimate  where we replace $\Mor[\Psi](0, \tau)$ by $\Morr[\Psi](0, \tau)$. We apply the vectorfield method to $T^\de=r^{-\de} T$. 

 \begin{proposition}
\label{prop:improved-Morawetz}
Let $ T^\de= r^{-\de} T$ . Then for a solution $\Psi$ to the equation \eqref{general-equation}, we have
\beaa
\EE[T^\de, 0,0](\Psi) &=&-\frac 1 4 |\nabb_4\Psi|^2   \de r^{-1-\de}\Up^2 +\frac 1 4 |\nabb_3\Psi|^2 \de r^{-1-\de} , \\
\PP^{(T^\de, 0,0)}  \c e_4 &=& \frac 1 2 r^{-\de}\Up |\nabb_4\Psi|^2+\frac 1 2 r^{-\de} (|\nabb\Psi|^2+V_i |\Psi|^2)\\
\PP^{(T^\de, 0,0)}  \c e_3 &=& \frac 1 2 r^{-\de} |\nabb_3\Psi|^2+\frac 1 2 r^{-\de}\Up (|\nabb\Psi|^2+V_i |\Psi|^2)
\eeaa
\end{proposition}

We are now ready to derive the improved Morawetz estimates for $\Psi$. 

\begin{theorem}\label{final-Morawetz} Let $\Psi$ be a solution to the equation \eqref{general-equation}. Consider a fix $\de>0$. Then the following improved Morawetz estimates hold
     \beaa
  && E[\Psi](\tau)+E_{\mathcal{H}^+}[\Psi](0, \tau)+E_{\mathcal{I}^+, 0}[\Psi](0, \tau)\\
  &&+ \Morr[\Psi](0, \tau)+\Mor_{\mathcal{H}^+}[\Psi]( 0, \tau)-\int_{\MM_{r \geq R}(0, \tau)}\frac 1 4 |\nabb_4\Psi|^2   \de r^{-1-\de}\Up^2\\
    &&\les E[\Psi](0)-          \int_{\MM(0,\tau)} \left(X(\Psi) +\frac 1 2  w \Psi+\th(r) r^{-\de}T( \Psi )\right)\c \M_i[\Psi]
  \eeaa
  \end{theorem}
\begin{proof}

 Let   $\th=\th(r) $ supported    for   $r\ge R/2$   with     $\th=1$ for $r\ge R$. Consider the vectorfield $T^{\de}_R=\th(r)T^\de$.  We apply the divergence theorem to $ \PP^{(\de)}:=\PP^{(T^{\de}_R, 0, 0)}$ in the spacetime region bounded by $\Sigma_0$ and $\Sigma_\tau$. By \eqref{eq:modified-div}, we have 
 \beaa
  \D^\mu  \PP_\mu^{(T^{\de}_R, 0, 0)}[\Psi]&=&  \EE[T^{\de}_R, 0, 0](\Psi)  +   \th(r) r^{-\de}T( \Psi )\c \M_i   
  \eeaa
By divergence theorem we then have:
     \beaa
   \int_{\Si_\tau}  \PP^{(\de)}\c  n_\Si +\int_{\mathcal{I}^+(0, \tau)}  \PP^{(\de)}\c  e_3  +\int_{\MM(0, \tau)}\EE[T^{\de}_R, 0, 0]=\int_{\Si_0} \PP^{(\de)}\c   n_\Si -          \int_{\MM(0,\tau)} \th(r) r^{-\de}T( \Psi )\c \M_i
     \eeaa
     By estimating the region where $r$ is bounded by $R$ by the energies, and using the Morawetz estimates \eqref{Morawetz-1} and Proposition \ref{prop:improved-Morawetz}, we obtain the desired estimate. 
  \end{proof}
 Observe that the term $-\int_{\MM_{r \geq R}(0, \tau)}\frac 1 4 |\nabb_4\Psi|^2   \de r^{-1-\de}\Up^2$ in the estimates above will be absorbed in the next subsection by the $r^p$-hierarchy estimates.

 \subsection{The $r^p$-hierarchy estimates}\label{section-rp-estimates}
  
  To derive the $r^p$-estimates, we apply the vector field method to $Z=l(r) e_4$. We choose $w$ as a function of $l$.
 
 \begin{proposition}\label{main-identity-r^pvf} Let $Z=l(r) e_4$ and $w=\frac{2l}{r}$. Then for a solution $\Psi$ to the equation \eqref{general-equation}, we obtain
 \beaa
 \EE[Z, w](\Psi)   &=& \frac 1 2\left(- l' +\frac 2 r l\right)\left(|\nabb \Psi|^2+V_i |\Psi|^2\right)+ \frac 1 2 l' |\nabb_4 \Psi|^2-  \frac{l''}{2r}|\Psi|^2\\
  &&+O\left(\frac{M}{r^2}, \frac{Q^2}{r^3}\right) |l| |\nabb_4 \Psi|^2+O\left(\frac{M}{r^4}, \frac{Q^2}{r^5}\right) \left[ |l|+r|l'|+r^2|l''|  \right]|\Psi|^2
   \eeaa
 \end{proposition}

We will correct the definition of $\EE$ by using a one form $M$ in order to  compensate for   the    term $ -   \frac 1 2   r^{-1} l''  |\Psi|^2$ in the above expression. 

\begin{proposition}
\label{prop:QC-general-multiplier2} Let $Z$ and $w$ be as in Proposition \ref{main-identity-r^pvf}, and let  $M=\frac{ 2l'}{r}  e_4=\frac{2l'}{rl} Z$. Then for a solution $\Psi$ to the equation \eqref{general-equation},   we have
\beaa
 \label{eq:rp-Daf-Rodn1}
\EE[Z, w , M]&=&  \frac  1 2  l'\big| \check{\nabb}_4(\Psi)\big|^2+\frac 1 2  \left(- l'+\frac{2l}{r}\right)(|\nabb\Psi|^2+V_i|\Psi|^2)\\
 &&+O\left(\frac{M}{r^2}, \frac{Q^2}{r^3}\right) |l| |\nabb_4 \Psi|^2+O\left(\frac{M}{r^4}, \frac{Q^2}{r^5}\right) \left[ |l|+r|l'|+r^2|l''|  \right]|\Psi|^2
\eeaa
\end{proposition}
\begin{proof} By \eqref{le:divergPP-gen-EE} we get
\beaa
\EE[Z, w , M]&=&\EE[Z, w]+\frac 1 4( \D^\mu M_\mu)|\Psi|^2+\frac 12 \Psi \c M(\Psi)
\eeaa
We compute, using Corollary \ref{lemma:componentspiR}
\beaa
\D^\mu M_\mu&=&\g^{\mu\nu}\D_\nu M_\mu=\g^{\mu\nu}\D_\nu(\frac{2l'}{rl} Z_\mu)=  \frac{l'}{r l}  \tr\piZ+Z(  \frac{2l'}{r l})=  \frac{l'}{r l}  (2l'+\frac{4l}{r})+2l e_4\left(  \frac{l'}{r l}\right)=\frac{2 l'}{r^2} +\frac{2l''}{r}
\eeaa
We deduce, 
\beaa
\EE[Z, w , M]&=& \frac 1 2\left(- l' +\frac 2 r l\right)\left(|\nabb \Psi|^2+V_i |\Psi|^2\right)+ \frac{l'}{2} |\nabb_4 \Psi|^2+\frac{l'}{2r^2}|\Psi|^2+  \frac{l'}{r}     \Psi \c e_4(\Psi)\\
 &&+O\left(\frac{M}{r^2}, \frac{Q^2}{r^3}\right) |l| |\nabb_4 \Psi|^2++O\left(\frac{M}{r^4}, \frac{Q^2}{r^5}\right) \left[ |l|+r|l'|+r^2|l''|  \right]|\Psi|^2
\eeaa
Writing $ \frac 1 2 l'|e_4\Psi|^2+ \frac{ l'}{2r^2}|\Psi|^2+  r^{-1} l'     \Psi \c e_4(\Psi)=\frac  1 2  l'( e_4(\Psi)+ r^{-1} \Psi)^2=\frac 1 2 l' |\check{\nabb}_4 \Psi|^2$, we get the desired expression.
\end{proof}

We now relate the bulk $\EE$ with the weighted bulk norm in the far away region. 
 \begin{corollary}\label{bound-E-rp}
 Assume that $l(r)=r^p$. Given a fixed $\de>0$, for all $\de \le p\le 2-\de$  and $ R\gg \max(\frac{M}{\de}, \frac{Q^2}{\de^2}) $, the following estimate holds
 \beaa
 \label{eq:Dafermos-Rodn1}
 \int_{\MM_{\ge R}(0,\tau)  }\EE[Z, w , M]&\ge &\frac 1  4 \MMdot_{p\,; \,R}[\Psi](0, \tau)
 \eeaa
\end{corollary}
\begin{proof}
  By Proposition \ref{prop:QC-general-multiplier2}, we have
\beaa
\EE[Z, w , M]&=&  \frac  p 2  r^{p-1}\big| \check{\nabb}_4(\Psi)\big|^2+\frac 1 2 (2-p) r^{p-1} (|\nabb\Psi|^2+V_i|\Psi|^2)\\
&&+O\left(\frac{M}{r^2}, \frac{Q^2}{r^3}\right) r^p |\nabb_4 \Psi|^2+O\left(\frac{M}{r^4}, \frac{Q^2}{r^5}\right)  r^p |\Psi|^2 
\eeaa
While integrating in $r \geq R$, the two last terms can be absorbed by the first two terms. Thus, we obtain
 \beaa
  \int_{\MM_{\ge R}(0,\tau)  }\EE[Z, w , M]&\ge &\frac 1  4  \int_{\MM_{\ge R}(\tau_1,\tau_2)  } r^{p-1}( p   |\check{\nabb}_4(\Psi)|^2+ (2-p) (|\nabb \Psi|^2 + r^{-2}|\Psi|^2))\\
  &=&\frac 1  4 \MMdot_{p\,; \,R}[\Psi](0, \tau)
   \eeaa
as desired.
\end{proof}

 We will compute now the current $\PP[Z, w]$ associated to the vector field $Z$.
  
\begin{lemma}
\label{prop:QC-general-multiplier} let $Z$, $w $ and $M$ as defined in Proposition \ref{prop:QC-general-multiplier2}. The current $\PP_\mu[Z, w, M]$ verifies
\beaa
\PP\c e_4&=&l|\check{\nabb}_4\Psi|^2 - \frac 1 2 r^{-2}  e_4( r  l |\Psi|^2) \\
 \PP\c e_3 &=&l (|\nabb\Psi|^2+V_i |\Psi|^2)+\frac  12  r^{-2}e_3\big ( r  l \Psi^2)  +O\left(\frac{M}{r^2}, \frac{Q^3}{r^5}\right) |l'|  |\Psi|^2
\eeaa
     \end{lemma}
\begin{proof} 
Recall the definition of current \eqref{definition-of-P}. Then 
\beaa
\PP\c   e_4&=& l\QQ_{44}+ \frac 1 r  l \Psi \c e_4\Psi - \frac 1 2 e_4(r^{-1} l)| \Psi|^2=l|e_4\Psi+ \frac 1 r   \Psi|^2 - \frac 1 2 r^{-2}  e_4( r  l |\Psi|^2),\\
 \PP\c e_3 &=& l  \QQ_{34}+  \frac 1 2 r^{-1}   l e_3(\Psi^2) -\frac 1  2  e_3  ( r^{-1}   l  ) \Psi^2- r^{-1} l' |\Psi|^2\\
&=&l \QQ_{34} +\frac  12  r^{-2}e_3\big ( r  l \Psi^2)+ r^{-1} f'|\Psi|^2+  O\left(\frac{M}{r^2}, \frac{Q^3}{r^5}\right) |l'|  |\Psi|^2- r^{-1} f' |\Psi|^2 \\
&=&l \QQ_{34} +\frac  12  r^{-2}e_3\big ( r  l \Psi^2)  +O\left(\frac{M}{r^2}, \frac{Q^3}{r^5}\right) |l'|  |\Psi|^2
 \eeaa
as desired.
\end{proof}

We are now ready to derive the $r^p$-estimates for $\Psi$. 
\begin{theorem}
  \label{theorem:Daf-Rodn1} Let $\Psi$ be a solution to the equation \eqref{general-equation}.  Consider  a fixed $\de>0$ and let $R\gg \max(\frac{M}{\de}, \frac{Q^2}{\de^2}) $. Then for all $\de\le p\le 2-\de $  the following $r^p$-estimates hold 
  \bea    \label{eq:Daf-Rodn-estim2}  
  \begin{split}
     &\Ed_{p\,;\, R}[\Psi](\tau)+\int_{\mathcal{I}^+(0, \tau)}\left( r^p|\nabb \Psi|^2+r^{p-2}|\Psi|^2\right)+  \MMdot_{p\,;\,R}[\Psi](0, \tau)   \\
      &\les 
          E_{p}[\Psi](0)-          \int_{\MM(0,\tau)} \left(R^p(X(\Psi) +\frac 1 2  w \Psi)+ l_p \ec_4 \Psi\right)\c \M_i[\Psi]
          \end{split}
\eea
\end{theorem}
  \begin{proof}  Let   $\th=\th(r) $ supported    for   $r\ge R/2$   with     $\th=1$ for $r\ge R$  such that $l_p=\th(r) r^p$, $Z_p=l_p e_4$, $w_p=\frac{2l_p}{r}$, $M_p=\frac{ 2l_p'}{r}  e_4$.
Let $\PP^{(p)}:=\PP[  Z_p, w_p, M_p]$.  We apply the divergence theorem to $ \PP^{(p)}$ in the spacetime region bounded by $\Sigma_0$ and $\Sigma_\tau$. We first note that 
     \beaa
       \D^\mu  \PP^{(p)} &=& \EE[Z_p,  w_p, M_p]+(Z_p(\Psi)+\frac 1 2 w_p \Psi) \c \M_i[\Psi]= \EE[Z_p,  w_p, M_p] +(l_p(r) e_4(\Psi)+\frac{ l_p}{r} \Psi) \c \M_i[\Psi]\\
       &=&\EE[Z_l,  w_l, M_l] +l_p(r) \ec_4(\Psi)\c \M_i[\Psi]
       \eeaa
By divergence theorem we then have:
     \beaa
   \int_{\Si_\tau}   \PP\c  N_\Si +\int_{\mathcal{I}^+(0, \tau)}   \PP\c e_3+\int_{\MM(0, \tau)}\EE[Z_l,  w_l, M_l]=\int_{\Si_0}   \PP\c   N_\Si -          \int_{\MM(0,\tau)} l_p \ec_4 \Psi\c \M[\Psi]
     \eeaa
We can estimate the integrals where $r$ is bounded by $R$ with the Morawetz energies.       Hence, we obtain
      \beaa
    &&\int_{\Si_{r \geq R}(\tau)}   \PP\c  N_\Si  +\int_{\mathcal{I}^+(0, \tau)}   \PP\c e_3+\int_{\MM_{r \geq R}(0, \tau)}\EE[Z_l,  w_l, M_l]\\
    &\les&\int_{\Si_{r \geq R}(0)}   \PP\c   N_\Si +R^p\left(E[\Psi](0)+E[\Psi](\tau)+R^{-1} \Mor[\Psi](0,\tau)\right) -          \int_{\MM(0,\tau)} l_p \ec_4 \Psi\c \M[\Psi]
     \eeaa
     Using the Morawetz estimates \eqref{Morawetz-1} and Corollary \ref{bound-E-rp}, we obtain 
      \beaa
   && \int_{\Si_{r \geq R}(\tau)}   \PP\c  N_\Si +\int_{\mathcal{I}^+(0, \tau)}   \PP\c e_3+ \dot{\MM}_{p\,;\,R}[\Psi](0, \tau) \\
    &\les&\int_{\Si_{r \geq R}(0)}   \PP\c   N_\Si +R^p E[\Psi](0) -          \int_{\MM(0,\tau)} \left(R^p(X(\Psi) +\frac 1 2  w \Psi)+ l_p \ec_4 \Psi\right)\c \M[\Psi]
     \eeaa
 Finally, by Lemma \ref{prop:QC-general-multiplier}, we have the bounds
 \beaa
 \int_{\Si_{r \geq R}(\tau)}   \PP\c  N_\Si&=& \int_{\Si_{r \geq R}(\tau)}   \PP\c e_4=\int_{\Si_{r \geq R}(\tau)} \left( r^p|\check{\nabb}_4\Psi|^2 - \frac 1 2 r^{-2}  e_4( r^{p+1} |\Psi|^2)\right) \ges \frac 1 2  \dot{E}_{p, R}[\Psi](\tau), \\
  \int_{\mathcal{I}^+(0, \tau)}   \PP\c e_3&=&  \int_{\mathcal{I}^+(0, \tau)} r^p (|\nabb\Psi|^2+V_i |\Psi|^2) +\frac  12  r^{-2}e_3\big ( r  l \Psi^2) +O\left(\frac{M}{r^2}, \frac{Q^3}{r^5}\right) |l'|  |\Psi|^2 \\
  &\ges&\int_{\mathcal{I}^+(0, \tau)} r^p|\nabb \Psi|^2+r^{p-2}|\Psi|^2
 \eeaa
 by performing the integration by parts   for the second terms in the integrals, and absorbing the boundary term. This proves the theorem.
         \end{proof}     
         
           We combine Theorem \ref{final-Morawetz} and Theorem \ref{theorem:Daf-Rodn1}. Then the term $-\int_{\MM_{r \geq R}(0, \tau)}\frac 1 4 |\nabb_4\Psi|^2   \de r^{-1-\de}\Up^2$ can be absorbed by the Morawetz bulk with better decay in $|\nabb_4\Psi|^2$. We finally obtain the following estimate:
              \bea\label{final-morawetz-rp}
              \begin{split}
  & E_p[\Psi](\tau)+E_{\mathcal{H}^+}[\Psi](0, \tau)+E_{\mathcal{I}^+, p}[\Psi](0, \tau)+ \MM_p[\Psi](0, \tau)\\
    &\les E_p[\Psi](0)-          \int_{\MM(0,\tau)} \left(X(\Psi) +\frac 1 2  w \Psi+\th(r) r^{-\de}T( \Psi )+ \th(r) r^p \ec_4 \Psi\right)\c \M_i[\Psi]
    \end{split}
  \eea
       
      \subsection{The  inhomogeneous  term }
      
      We now analyze the inhomogeneous term 
      \beaa
      -          \int_{\MM(0,\tau)} \left(X(\Psi) +\frac 1 2  w \Psi+\th(r) r^{-\de}T( \Psi )+ \th(r) r^p \ec_4 \Psi\right)\c \M_i[\Psi].
      \eeaa
       Recall from \eqref{definition-X} that the vectorfield $X$ is given by    $X=  f_\de   R +\ep_\HH Y_\HH +\La T$. Recall that $f$ vanishes at $r=r_P$, and $\th(r)$ is supported in the far away region, in particular away from the trapping region. We separate the terms involving $R$ and $T$ with the term involving $l_p \ec_4$, which we bound by the absolute value. 
\bea\label{inhomogeneous-term}
\begin{split}
&-          \int_{\MM(0,\tau)} \left(X(\Psi) +\frac 1 2  w \Psi +\th(r) r^{-\de}T( \Psi )+ \th(r) r^p \ec_4 \Psi\right)\c \M[\Psi]\\
&\le - \int_{\MM(0,\tau)} \left((r-r_P)R(\Psi)+ \Lambda T(\Psi )+\frac 1 2  w\Psi\right)\c \M[\Psi]\\
&+C\int_{\MM_{\ge R}(0,\tau)  }\left( r^p |\check{\nabb}_4\Psi||\M[\Psi]|+r^{-\de}|T\Psi| |\M[\Psi]|\right)
\end{split}
\eea
The integral in the far-away region can be separated using Cauchy-Schwarz:
         \beaa
\int_{\MM_{\ge R}(0,\tau)  }  r^p |\check{\nabb}_4\Psi||\M[\Psi]|&\le&  \lambda \int_{\MM_{\ge R}(0,\tau)  } r^{p-1} |\check{\nabb}_4\Psi|^2+\lambda^{-1} \int_{\MM_{\ge R}(0,\tau)  } r^{p+1} |\M[\Psi]|^2 \\
\int_{\MM_{\ge R}(0,\tau)  }  r^{-\de}|T\Psi| |\M[\Psi]|&\le&  \lambda \int_{\MM_{\ge R}(0,\tau)  } r^{-1-3\de} |T\Psi|^2+\lambda^{-1} \int_{\MM_{\ge R}(0,\tau)  } r^{1+\de} |\M[\Psi]|^2
         \eeaa
   For $\lambda$ small enough the first  integrals on the right can be absorbed in the Morawetz bulk $\MM_p[\Psi]$ in \eqref{final-morawetz-rp},  and  recalling the definition of $\II_{p\,; \,R}[ \M]$ in \eqref{definition-norm-M}, we obtain 
            \bea\label{final-estimate-Psi}
            \begin{split}
       & E_p[\Psi](\tau)+E_{\mathcal{H}^+}[\Psi](0, \tau)+E_{\mathcal{I}^+, p}[\Psi](0, \tau)+ \MM_p[\Psi](0, \tau)\\
           &\les E_p[\Psi](0)+\II_{p\,; \,R}[ \M_i](0,\tau)- \int_{\MM(0,\tau)} \left((r-r_P)R(\Psi)+ \Lambda T(\Psi )+\frac 1 2  w\Psi\right)\c \M_i[\Psi] 
        \end{split}
         \eea  
     Applying \eqref{final-estimate-Psi} to $\Psi=\qf$ and $\M[\Psi]=\M_1[\qf, \qf^\F]$, we finally obtain estimate \eqref{estimate1}. Similarly, applying it to $\Psi=\underline{\qf}$ and $\M[\Psi]=\M_1[\underline{\qf}, \underline{\qf}^\F]$, we obtain an identical estimate for the spin $-2$ quantities.

             \subsection{Higher order estimates}
             In this section we extend the above weighted estimates to higher order. We note the trivial fact that the wave operator $\Box_\g$ commutes with the Lie differentiation with the Killing fields of the Reissner-Nordstr{\"o}m metric, i.e. $T$ and $\Omega_i$. In particular, since $T(V_i)=r\nabb(V_i)=0$, when applying $T$ or $r\nabb$ to equation \eqref{general-equation}, we obtain
             \beaa
\Big(\Box_\g-V_i \Big) (T\Psi_i)&=& T(\M_i[\Psi]), \label{commuted-equation-T}\\
\Big(\Box_\g-V_i \Big) (r\nabb\Psi_i)&=& r\nabb(\M_i[\Psi])\label{commuted-equation-nabb}
\eeaa
             Recalling the higher order energies defined in Section \ref{definition-norms}, applying \eqref{final-estimate-Psi} to $n$-commuted $T^i(r\nabb)^j \Psi$, we immediately conclude the following Corollary.
             \begin{corollary}\label{higher-order-single} Let $\Psi$ be a solution to the equation \eqref{general-equation}. Then the following higher order estimates hold 
             \bea\label{final-estimate-Psi-higher}
             \begin{split}
              &  E_p^{n, T, \nabb}[\Psi](\tau)  +E^{n, T, \nabb}_{\mathcal{H}^+}[\Psi](0, \tau)+E^{n, T, \nabb}_{\mathcal{I}^+, p}[\Psi](0, \tau)     +  \MM_{p}^{n, T, \nabb}[\Psi](0, \tau)   \\
              &\les E_p^{n, T, \nabb}[\Psi](0)+\II^{n, T, \nabb}_{p\,; \,R}[ \M](0, \tau)\\
                 &- \sum_{i+j \leq n}\int_{\MM(0,\tau)} \left((r-r_P)R((T^i)(r\nabb)^j\Psi)+ \Lambda T((T^i)(r\nabb)^j\Psi )+\frac 1 2  w (T^i)(r\nabb)^j\Psi\right)\c (T^i)(r\nabb)^j\M[\Psi] 
                 \end{split}
                \eea
             \end{corollary}
               Applying the above estimate to $\Psi=\qf^\F$, $\M[\Psi]=\M_2[\qf, \qf^\F]$ with $n=1$, we finally obtain estimate \eqref{estimate2}. Similarly, applying it to $\Psi=\underline{\qf}^\F$, $\M[\Psi]=\M_2[\underline{\qf}, \underline{\qf}^\F]$ and $n=1$, we obtain an identical estimate for the spin $-2$ quantities.
               
               Observe that applying the above estimate to $\Psi=\qf^\F$, $\M[\Psi]=\M_2[\qf, \qf^\F]$ with $n=0$, we obtain the non-commuted estimate
                     \bea\label{estimate2-uncommuted}
       \begin{split}
    &  E_p[\qf^\F](\tau)    + E_{\mathcal{H}^+}[\qf^\F](0,\tau)  +E_{\mathcal{I}^+, p}[\qf^\F](0, \tau)          +\MM_p[\qf^\F](0,\tau)\\
      &\les E_p[\qf^\F](0)+  \II_{p\,; \,R}[\M_2[\qf, \qf^\F]]( 0,\tau)  \\
       &- \int_{\MM(0,\tau)} \left((r-r_P)R(\qf^\F)+ \Lambda T(\qf^\F )+\frac 1 2  w\qf^\F\right)\c \M_2[\qf, \qf^\F]
   \end{split}
  \eea

\section{Combined estimates for the system of spin $+2$}\label{lot-absorbing}
In this section, we obtain the estimates for the generalized Regge-Wheeler system of spin $+2$. In Section \ref{spin-2-all}, we will outline the similar procedure used in the case of the generalized Regge-Wheeler system of spin $-2$.

We combine estimates \eqref{estimate1} and \eqref{estimate2-uncommuted} in order to obtain the combined estimate \eqref{first-estimate-main-theorem-1} in the Main Theorem.

We sum $A$ times estimate \eqref{estimate1} to $B$ times the non-commuted estimate \eqref{estimate2-uncommuted}, with $C$ times \eqref{estimate2-uncommuted} commuted with $T$ and $D$ times estimate \eqref{estimate2-uncommuted} commuted with $r\nabb$, where $A$, $B$, $C$, $D$ are positive constants. We get  
    \bea\label{estimate3}
       \begin{split}
        E_p[\qf](\tau)&  + E^{1, T, \nabb}_p[\qf^\F](\tau)+E_{\mathcal{H}^+}[\qf](0, \tau)+E^{1, T, \nabb}_{\mathcal{H}^+}[\qf^\F](0, \tau)+E_{\mathcal{I}^+, p}[\qf](0, \tau)+E^{1, T, \nabb}_{\mathcal{I}^+, p}[\qf^\F](0, \tau)            \\
      &  +\MM_p[\qf](0,\tau) +\MM^{1, T, \nabb}_p[\qf^\F](0,\tau)\leq E_{p}[\qf](0)+E^{1, T, \nabb}_p[\qf^\F](0)\\
       &+\II_{p\,; \,R}[ \M_1[\qf, \qf^\F]](0,\tau)+  \II^{1, T, \nabb}_{p\,; \,R}[\M_2[\qf, \qf^\F]]( 0,\tau)  \\
        &- A\int_{\MM(0,\tau)} \left((r-r_P)R(\qf)+ \Lambda T(\qf )+\frac 1 2  w\qf\right)\c \M_1[\qf, \qf^\F]\\
          &- B\int_{\MM(0,\tau)} \left((r-r_P)R(\qf^\F)+ \Lambda T(\qf^\F )+\frac 1 2  w\qf^\F\right)\c \M_2[\qf, \qf^\F]\\
        &- C\int_{\MM(0,\tau)} \left((r-r_P)R(T\qf^\F)+ \Lambda T(T\qf^\F )+\frac 1 2  w T\qf^\F\right)\c T(\M_2[\qf, \qf^\F])\\
        &- D\int_{\MM(0,\tau)} \left((r-r_P)R(r\nabb_A\qf^\F)+ \Lambda T(r\nabb_A\qf^\F )+\frac 1 2  w r\nabb_A\qf^\F\right)\c r\nabb_A(\M_2[\qf, \qf^\F])
        \end{split}
  \eea
  In the combined estimate \eqref{estimate3}, the last five lines on the right hand side are not controlled from initial data at this point. Our goal is to study these terms in order to obtain estimates for them. In doing so, the particular structure of the right hand side terms $\M_1[\qf, \qf^\F]$ and $\M_2[\qf, \qf^\F]$ will play an important role. 
To simplify such structure, recalling that $\rhoF=\frac{Q}{r^2}$, we can write the system \eqref{finalsystem} in the following concise form:
\bea\label{schematic-system}
\begin{cases}
\Big(\square_\g-V_1\Big)\qf= \M_1[\qf, \qf^\F]:=Q\ \C_1[\qf^\F]+Q  \ \Lbb_1[\qf^\F]+Q^2 \ \Lbb_1[\qf], \\
\Big(\square_\g-V_2\Big)\qf^{\F}=\M_2[\qf, \qf^\F]:=Q\  \C_2[\qf] +Q^2 \Lbb_2[\qf^\F]
\end{cases}
\eea      
where
\bea\label{definition-main-coefficients}
  \C_1[\qf^\F]&=&\frac{4}{r}\lapp_2\qf^{\F}-\frac{4}{r}\kab \nabb_4\qf^{\F}-\frac{4}{r}\ka \nabb_3\qf^{\F} + \frac 1 r \left(6\ka\kab+16\rho+8\rhoF^2\right)\qf^{\F}, \\
 \C_2[\qf]&=& -\frac {1}{ r^3} \qf, \\
 \Lbb_1[\qf]&=& -\frac{2}{r^2} \psi_0-\frac{4}{r^3}\psi_1, \\
 \Lbb_1[\qf^\F]&=& -12\rho \psi_3 -Q^2 \frac{40}{r^4} \psi_3, \\
 \Lbb_2[\qf^\F]&=& \frac{4}{r^3} \psi_3
 \eea
and $|Q|\ll M$ is the charge of the Reissner-Nordstr{\"o}m spacetime.

We will first derive estimates for the coupling terms $\C_1$ and $\C_2$, and then for the lower order terms $\Lbb_1$ and $\Lbb_2$.

\subsection{Estimates for the coupling terms}\label{coupling-subsection}

Our goal is to absorb the coupling terms on the right hand side of \eqref{estimate3} by the Morawetz bulks of $\qf$ and $\qf^\F$ on the left hand side of \eqref{estimate3}. In the trapping region, the analysis is more subtle, because of the degeneracy of the Morawetz bulks.

The goal of this section is to prove the following estimate for the coupling terms:
\bea\label{global-estimate-coupling}
\begin{split}
&\II_{p\,; \,R}[ \C_1[\qf^\F]](0,\tau)+  \II^{1, T, \nabb}_{p\,; \,R}[\C_2[\qf]]( 0,\tau) - A\int_{\MM(0,\tau)} \left((r-r_P)R(\qf)+ \Lambda T(\qf )+\frac 1 2  w\qf\right)\c \C_1[\qf^\F]\\
          &- B\int_{\MM(0,\tau)} \left((r-r_P)R(\qf^\F)+ \Lambda T(\qf^\F )+\frac 1 2  w\qf^\F\right)\c \C_2[\qf]\\
        &- C\int_{\MM(0,\tau)} \left((r-r_P)R(T\qf^\F)+ \Lambda T(T\qf^\F )+\frac 1 2  w T\qf^\F\right)\c T(\C_2[\qf])\\
        &- D\int_{\MM(0,\tau)} \left((r-r_P)R(r\nabb_A\qf^\F)+ \Lambda T(r\nabb_A\qf^\F )+\frac 1 2  w r\nabb_A\qf^\F\right)\c r\nabb_A(\C_2[\qf])\\
        &\les \MM_p[\qf](0,\tau)+\MM^{1, T, \nabb}_p[\qf^\F](0,\tau)+E_p[\qf](\tau) +E_p[\qf^\F](\tau)+E_p[\qf](0) +E_p[\qf^\F](0)\\
&- C\int_{\MM_{\frac 5 6 r_P \leq r \leq \frac 7 6 r_P}}\Lambda ( -\Up Q^2 \Lbb_2[\qf^\F])\c T\left(-\frac{1}{r^3} \qf \right)+C\int_{\MM_{\frac 5 6 r_P \leq r \leq \frac 7 6 r_P}}\frac 1 2  w ( -Q^2\Up \Lbb_2[\qf^\F]) \c \left(-\frac{1}{r^3} \qf \right)
\end{split}
\eea

We separate the proof in three parts: the far-away region, outside the trapping region, and at the photon sphere. 

\subsubsection{Absorption in the far-away region}

The goal of this subsection is to prove the following estimate
\bea\label{absorption-coupling-far-away}
\II_{p\,; \,R}[  \C_1[\qf^\F]](0,\tau)+  \II^{1, T, \nabb}_{p\,; \,R}[ \C_2[\qf]]( 0,\tau) &\les& \MM_p[\qf](0,\tau)+\MM^{1, T, \nabb}_p[\qf^\F](0,\tau)
\eea

In the far-away region, we will need to keep track of the powers of $r$. In particular:
\bea\label{schematic-coupling-terms}
\C_1[\qf^\F]&\simeq&\Big[\frac{1}{r}\lapp_2\qf^{\F}, \frac{1}{r^2}   e_4\qf^{\F},\frac{1}{r^2} e_3\qf^{\F}, \frac{1}{r^3}\qf^{\F}\Big], \\
 \C_2[\qf]&\simeq& \Big[\frac {1}{ r^3} \qf\Big]
\eea
where $\simeq$ denotes the asymptotics in $r$ towards null infinity.
Therefore, we have
\beaa
\II_{p\,; \,R}[  \C_1[\qf^\F]](0,\tau)&=& \int_{\MM_{far}(0,\tau)} r^{1+p}  |\frac{1}{r}\lapp_2\qf^{\F}, \frac{1}{r^2}   e_4\qf^{\F},\frac{1}{r^2} e_3\qf^{\F}, \frac{1}{r^3}\qf^{\F}|^2 \\
&=&  \int_{\MM_{far}(0,\tau)} r^{-5+p}|(r\nabb)^2\qf^{\F}|^2+r^{-3+p} |\nabb_4\qf^{\F}|^2+r^{-3+p}| \nabb_3\qf^{\F}|^2+r^{-5+p}|\qf^{\F}|^2
\eeaa
 Since the powers of $r$ of $\II_p[ \qf, \C_1[\qf^\F]]$ decay all faster than the respective ones in $\MM^{1, T, \nabb}_p[\qf^\F]$, we have that, for $R$ big enough
 \beaa
\II_{p\,; \,R}[  \C_1[\qf^\F]](0,\tau) &\les& \MM^{1, T, \nabb}_p[\qf^\F](0,\tau)
 \eeaa
Similarly for $\II^{1, T, \nabb}_{p\,; \,R}[ \C_2[\qf]]( 0,\tau)$, which proves \eqref{absorption-coupling-far-away}.

\subsubsection{Absorption outside the photon sphere}\label{outside-trapping-subsub}
The goal of this subsection is to prove the following estimate
\bea\label{absorption-coupling-outside-trapping}
\int_{\MM(0,\tau)\setminus r=r_p} \operatorname{integrals \ in \ RHS \ of \eqref{estimate3}} \les \MM_p[\qf](0,\tau) +\MM^{1, T, \nabb}_p[\qf^\F](0,\tau)
\eea

Outside the photon sphere, the Morawetz bulks $\MM_p[\qf](0,\tau) +\MM^{1, T, \nabb}_p[\qf^\F](0,\tau)$ contain the following terms: $|R\qf|^2$, $|T\qf|^2$, $|r\nabb_A \qf|^2$, $|\qf|^2$, $|R\qf^\F|^2$, $|T\qf^\F|^2$,  $|r\nabb_A \qf^\F|^2$, $|\qf^\F|^2$, $|RT\qf^\F|^2$, $|R(r\nabb_A)\qf^\F|^2$, $|TT\qf^\F|^2$, $|T(r\nabb_A)\qf^\F|^2$, $|r^2\nabb^2_A\qf^\F|^2$,
with $\Rbrev$ and $\Tbrev$ respectively in the red-shift region. 
Therefore the integrals on the right hand side of \eqref{estimate3} outside the photon sphere can be easily bounded by $\MM_p[\qf](0,\tau) +\MM^{1, T, \nabb}_p[\qf^\F](0,\tau)$ using Cauchy-Schwarz. 

\subsubsection{Absorption at the photon sphere}
The goal of this subsection is to prove the following estimate
\bea\label{absorption-coupling-trapping}
\begin{split}
\int_{\MM_{trap}(0,\tau)} \operatorname{integrals \ in \ RHS \ of \eqref{estimate3}} &\les \MM_p[\qf](0,\tau)+\MM^{1, T, \nabb}_p[\qf^\F](0,\tau)\\
&+E_p[\qf](\tau) +E_p[\qf^\F](\tau)+E_p[\qf](0) +E_p[\qf^\F](0)\\
&- C\int_{\MM_{trap}}\Lambda ( -\Up Q^2 \Lbb_2[\qf^\F])\c T\left(-\frac{1}{r^3} \qf \right)\\
&+C\int_{\MM_{trap}}\frac 1 2  w ( -Q^2\Up \Lbb_2[\qf^\F]) \c \left(-\frac{1}{r^3} \qf \right)
\end{split}
\eea
Recall that at the photon sphere, the Morawetz bulks $\MM_p[\qf](0,\tau) +\MM^{1, T, \nabb}_p[\qf^\F](0,\tau)$ contain the following terms with no degeneracy: $|R\qf|^2$,  $|\qf|^2$, $|R\qf^\F|^2$, $|T\qf^\F|^2$, $|r\nabb_A \qf^\F|^2$, $|\qf^\F|^2$, $|RT\qf^\F|^2$, $|R(r\nabb_A)\qf^\F|^2$,
and the following degenerate terms: $$(r-r_P)^2 \left( |T\qf|^2+|r\nabb_A \qf|^2+ |TT\qf^\F|^2+|T(r\nabb_A)\qf^\F|^2+|r^2\lapp\qf^\F|^2\right).$$

We analyze each line of the last four lines of \eqref{estimate3} for the coupling terms.

\bigskip

{\bf{First line} } Consider the highest order term $\frac 4 r \lapp_2 \qf^\F$. This gives 
\beaa
- A\int_{\MM_{trap}} \left((r-r_P)R(\qf)+ \Lambda T(\qf )+\frac 1 2  w\qf\right)\c \frac 4 r \lapp_2 \qf^\F
\eeaa
For the first term, we use Cauchy-Schwarz while moving the degeneracy $r-r_P$ to the laplacian, in order to absorb it by the Morawetz bulk:
\beaa
- A\int_{\MM_{trap}} \left((r-r_P)R(\qf)\right)\c \frac 4 r \lapp_2 \qf^\F &\leq& \left(\int_{\MM_{trap}} |R(\qf)|^2\right)^{1/2} \left(\int_{\MM_{trap}} (r-r_P)^2 |\lapp_2 \qf^\F|^2\right)^{1/2}\\
&\les&\MM_p[\qf](0,\tau) +\MM^{1, T, \nabb}_p[\qf^\F](0,\tau)
\eeaa
We keep for the moment the term with $T$:
\bea\label{problematic-term-1}
- A\int_{\MM_{trap}} \Lambda T(\qf ) \c \frac 4 r \lapp_2 \qf^\F
\eea
and the term with the zero-th order term:
\bea\label{problematic-term-2}
- A\int_{\MM_{trap}} \frac 1 2  w\qf \c \frac 4 r \lapp_2 \qf^\F
\eea
Consider the first order terms $-\frac{4}{r}\kab \nabb_4\qf^{\F}-\frac{4}{r}\ka \nabb_3\qf^{\F}$. They can be written as:
\beaa
-\frac{4}{r}\kab \nabb_4\qf^{\F}-\frac{4}{r}\ka \nabb_3\qf^{\F}&=&-\frac{4}{r}\kab \frac{1}{\Up}(T+R)\qf^{\F}-\frac{4}{r}\ka (T-R)\qf^{\F}\\
&=&-\frac{4}{r}( \frac{\kab}{\Up}+\ka)T\qf^{\F}-\frac{4}{r}(\frac{\kab}{\Up}-\ka)R\qf^{\F}=\frac{16}{r^2} R\qf^\F
\eeaa 
This gives 
\beaa
- A\int_{\MM_{trap}} \left((r-r_P)R(\qf)+ \Lambda T(\qf )+\frac 1 2  w\qf\right)\c \left(\frac{16}{r^2} R\qf^\F \right)
\eeaa
Since all first derivatives of $\qf^\F$ compare in the Morawetz bulk without degeneracy, we can apply Cauchy-Schwarz for the terms involving $R\qf$ and $\qf$, which are also non-degenerate. To absorb the term in $T$ we need to perform an integration by parts in $T$, and obtain a spacetime integral for $TR \qf^\F$, which is not degenerate. The boundary terms can be estimated by the energies. 

Consider the zero-th order term $\frac 1 r \left(6\ka\kab+16\rho+8\rhoF^2\right)\qf^{\F}$. Again, the terms involving $R\qf$ and $\qf$ can be absorbed by the Morawetz bulk. Similarly as before, the term involving $T\qf$ can be absorbed upon integration by parts, picking up energies as boundary terms.  while we keep the term with $T$.

To summarize, in the first line of estimate \eqref{estimate3} we absorbed all the terms by Morawetz bulks and energies, except terms \eqref{problematic-term-1} and \eqref{problematic-term-2}.

\bigskip

{\bf{Second line} } Consider the second line of estimate \eqref{estimate3}. The coupling term is given by
\beaa
- B\int_{\MM_{trap}} \left((r-r_P)R(\qf^\F)+ \Lambda T(\qf^\F )+\frac 1 2  w\qf^\F\right)\c \left(-\frac{1}{r^3} \qf \right)
\eeaa
All the quantities in $\qf^\F$ appear in the Morawetz bulk as non-degenerate, therefore we can apply Cauchy-Schwarz and absorb the second line by the Morawetz bulks.

\bigskip

{\bf{Third line} } Consider the third line of estimate \eqref{estimate3}. The coupling term is given by 
\beaa
- C\int_{\MM_{trap}}\left((r-r_P)R(T\qf^\F)+ \Lambda T(T\qf^\F )+\frac 1 2  w T\qf^\F\right)\c T\left(-\frac{1}{r^3} \qf \right)
\eeaa
For the first term, we use Cauchy-Schwarz while moving the degeneracy $r-r_P$ to the $T\qf$. We keep for the moment the term with $TT\qf^\F$:
\bea\label{problematic-term-3}
- C\int_{\MM_{trap}}\Lambda T(T\qf^\F )\c T\left(-\frac{1}{r^3} \qf \right)
\eea
and the term with $T\qf^\F$:
\bea\label{problematic-term-4}
- C\int_{\MM_{trap}}\frac 1 2  w T\qf^\F \c T\left(-\frac{1}{r^3} \qf \right)
\eea
To summarize, in the third line of estimate \eqref{estimate3} we absorbed all the terms by Morawetz bulks, except \eqref{problematic-term-3} and \eqref{problematic-term-4}.

\bigskip

{\bf{Fourth line} } Consider the fourth line of estimate \eqref{estimate3}. The coupling term is given by 
\beaa
- D\int_{\MM_{trap}}\left((r-r_P)R(r\nabb_A\qf^\F)+ \Lambda T(r\nabb_A\qf^\F )+\frac 1 2  w r\nabb_A\qf^\F\right)\c r\nabb_A\left(-\frac{1}{r^3} \qf \right)
\eeaa
For the first term, we use Cauchy-Schwarz while moving the degeneracy $r-r_P$ to the $r\nabb\qf$. 
We keep for the moment the term with $T\nabb\qf^\F$:
\bea\label{problematic-term-5}
- D\int_{\MM_{trap}}\Lambda T(r\nabb_A\qf^\F )\c r\nabb_A\left(-\frac{1}{r^3} \qf \right)
\eea
and the term with $\nabb \qf^\F$:
\bea\label{problematic-term-6}
- D\int_{\MM_{trap}}\frac 1 2  w r\nabb_A\qf^\F\c r\nabb_A\left(-\frac{1}{r^3} \qf \right)
\eea
To summarize, in the fourth line of estimate \eqref{estimate3} we absorbed all the terms by Morawetz bulks, except \eqref{problematic-term-3} and \eqref{problematic-term-4}.

 Consider term \eqref{problematic-term-3}. Using Lemma \ref{wave-T-R} and \eqref{schematic-system}, we write:
\bea\label{use-wave}
\begin{split}
\frac 1 \Up TT\qf^\F  &=-\square_\g \qf^\F   +\frac 1 \Up RR\qf^\F  +\lapp_2\qf^\F  +\frac 2 r  R \qf^\F \\
&=\frac 1 \Up RR\qf^\F  +\lapp_2\qf^\F +\frac 2 r  R \qf^\F -V_2 \qf^\F  +\frac{Q}{r^3} \qf -Q^2 \Lbb_2[\qf^\F]
\end{split}
 \eea
The term \eqref{problematic-term-3} therefore becomes
\beaa
- C\int_{\MM_{trap}}\Lambda T(T\qf^\F )\c T\left(-\frac{1}{r^3} \qf \right)&=& - C\int_{\MM_{trap}}\Lambda (RR\qf^\F  +\Up\lapp_2\qf^\F  )\c T\left(-\frac{1}{r^3} \qf \right)\\
&& - C\int_{\MM_{trap}}\Lambda (\frac 2 r \Up R \qf^\F -\Up V_2 \qf^\F )\c T\left(-\frac{1}{r^3} \qf \right)\\
&& - C\int_{\MM_{trap}}\Lambda (\Up\frac{Q}{r^3} \qf)\c T\left(-\frac{1}{r^3} \qf \right)\\
&& - C\int_{\MM_{trap}}\Lambda ( -\Up Q^2 \Lbb_2[\qf^\F])\c T\left(-\frac{1}{r^3} \qf \right)
\eeaa
The second line of the right hand side of the above identity can be bounded by performing integration by parts in $T$, and obtaining $TR\qf^\F$ and $T\qf^\F$ which are non-degenerate in the Morawetz bulk. As boundary terms, we pick up energies.
The third line of the right hand side of the above identity can be bounded by the energy.
The fourth line will be absorbed in Section \ref{section-lower-order}, since it involves the lower order terms. 

We concentrate now on the first term on the first line on the right hand side of the above identity. Since we will have to integrate per parts in $R$ as well as in $T$, it is convenient to consider a cut-off function in $R$, $\chi(r)$ such that $\chi(r)=0$ for $r \leq \frac{6}{7} r_P$ and $r\geq \frac{8}{7} r_P$ and $\chi(r)=1$ for $\frac{7}{8} r_P \leq r \leq \frac{9}{8}r_P$. We first integrate by parts in $R$:
\beaa
- C\int_{\MM_{trap}}\Lambda \chi(r) RR\qf^\F \c T\left(-\frac{1}{r^3} \qf \right)&=& C\int_{\MM_{trap}}\Lambda \chi'(r) R\qf^\F \c T\left(-\frac{1}{r^3} \qf \right)+\Lambda \chi(r) R\qf^\F \c RT\left(-\frac{1}{r^3} \qf \right)
\eeaa
The first term on the right hand side vanishes at the photon sphere, therefore it can be bounded by the Morawetz bulks. For the second term we commute $TR$ and perform again integration by parts in $T$.

Finally, term \eqref{problematic-term-3} is bounded by:
\beaa
\eqref{problematic-term-3} &\les& - C\int_{\MM_{trap}}\Lambda \Up\lapp_2\qf^\F  \c T\left(-\frac{1}{r^3} \qf \right)- C\int_{\MM_{trap}}\Lambda ( -\Up Q^2 \Lbb_2[\qf^\F])\c T\left(-\frac{1}{r^3} \qf \right)\\
&&+\MM_p[\qf](0,\tau) +\MM^{1, T, \nabb}_p[\qf^\F](0,\tau)+E_p[\qf](\tau) +E_p[\qf^\F](\tau)+E_p[\qf](0) +E_p[\qf^\F](0)
\eeaa

Consider term \eqref{problematic-term-5}. Performing integration by parts in $T$ and then in $r\nabb$ we obtain
\beaa
- D\int_{\MM_{trap}}\Lambda T(r\nabb_A\qf^\F )\c r\nabb_A\left(-\frac{1}{r^3} \qf \right)&\les &   D\int_{\MM_{trap}}\Lambda (r\nabb_A\qf^\F )\c r\nabb_A\left(-\frac{1}{r^3} T\qf \right)\\
&&+E_p[\qf](\tau) +E_p[\qf^\F](\tau)+E_p[\qf](0) +E_p[\qf^\F](0)\\
&\les &  - D\int_{\MM_{trap}}\Lambda (r^2\lapp_2\qf^\F )\c \left(-\frac{1}{r^3} T\qf \right)\\
&&+E_p[\qf](\tau) +E_p[\qf^\F](\tau)+E_p[\qf](0) +E_p[\qf^\F](0)
\eeaa
Consider now $\eqref{problematic-term-1}+\eqref{problematic-term-3}+\eqref{problematic-term-5}$. This gives 
\beaa
\eqref{problematic-term-1}+\eqref{problematic-term-3}+\eqref{problematic-term-5}&\les& - A\int_{\MM_{trap}} \Lambda T(\qf ) \c \frac 4 r \lapp_2 \qf^\F\\
&&- C\int_{\MM_{trap}}\Lambda \Up\lapp_2\qf^\F  \c T\left(-\frac{1}{r^3} \qf \right)- C\int_{\MM_{trap}}\Lambda ( -\Up Q^2 \Lbb_2[\qf^\F])\c T\left(-\frac{1}{r^3} \qf \right)\\
&&- D\int_{\MM_{trap}}\Lambda (r^2\lapp_2\qf^\F )\c \left(-\frac{1}{r^3} T\qf \right)\\
&&+\MM_p[\qf](0,\tau) +\MM^{1, T, \nabb}_p[\qf^\F](0,\tau)\\
&&+E_p[\qf](\tau) +E_p[\qf^\F](\tau)+E_p[\qf](0) +E_p[\qf^\F](0)
\eeaa
Choosing the constant $A$, $C$ and $D$ such that 
\bea\label{condition-cancellation}
\begin{split}
4A r^2-C\Up(r)-Dr^2|_{r=r_P}&=0, \\
\left(4A r^2-C\Up(r)-Dr^2\right)'|_{r=r_P}&=(8Ar-C \Up'(r)-2Dr)|_{r=r_P}=0
\end{split}
\eea
 we obtain a cancellation of second order for the terms involving $T\qf \c \lapp_2 \qf^\F$ at the photon sphere. 
 
 \begin{remark}\label{signs-different}  Observe that a choice of positive constants $A$, $C$, $D$ verifying conditions \eqref{condition-cancellation} is possible. We first remark that the two conditions in \eqref{condition-cancellation} reduce to one relation only. Indeed, recall that at the photon sphere relation \eqref{Up'-photonsphere} holds, i.e. 
 \beaa
 \Up'(r_P)= \frac{2}{r_P} \Up(r_P).
 \eeaa
 Therefore the second condition in \eqref{condition-cancellation} for the derivative reads
 \beaa
 0&=& 8Ar_P-C \Up'(r_P)-2Dr_P= 8Ar_P-C \frac{2}{r_P} \Up(r_P)-2Dr_P=\frac{2}{r_P}(4A r_P^2-C\Up(r_P)-Dr_P^2)
 \eeaa
 which reduces to the first condition in \eqref{condition-cancellation}. Finally,  since $\Up(r_P)\geq 0$, there exists a choice of positive constants $A$, $C$, $D$ such that $(4A-D) r_P^2-C\Up(r_P)=0$. 
  
  This is implied by two facts: the higher order terms in the coupling terms $\C_1[\qf^\F]$ (i.e. $\frac 4 r \lapp_2\qf^\F$) and $\C_2[\qf]$ (i.e. $-\frac{1}{r^3} \qf$) have opposite sign,  and the higher order term is second order derivatives. 
  
  In particular, compare the above cancellations with the terms  \eqref{problematic-term-2}, \eqref{problematic-term-4}, \eqref{problematic-term-6}. Even if we still had the freedom of choosing the coefficients $A$, $C$ and $D$, we would not be able to normalize positive constants at the photon sphere to cancel those terms, because of a sign issue.  Below, we treat them differently.
 \end{remark}
 
 Creating this second order degeneracy at the photon sphere, we can distribute the degeneracy $(r-r_P)$ to both $T\qf$ and $\lapp_2\qf^\F$, and bound it by Cauchy-Schwarz with the Morawetz bulks. This finally gives 
\beaa
\eqref{problematic-term-1}+\eqref{problematic-term-3}+\eqref{problematic-term-5}&\les&- C\int_{\MM_{trap}}\Lambda ( -\Up Q^2 \Lbb_2[\qf^\F])\c T\left(-\frac{1}{r^3} \qf \right)\\
&&+\MM_p[\qf](0,\tau) +\MM^{1, T, \nabb}_p[\qf^\F](0,\tau)\\
&&+E_p[\qf](\tau) +E_p[\qf^\F](\tau)+E_p[\qf](0) +E_p[\qf^\F](0)
\eeaa
We consider now terms \eqref{problematic-term-2}, \eqref{problematic-term-4}, \eqref{problematic-term-6}. They all end up being bounded by a term of the form 
\bea\label{general-form-problematic-form}
\int_{\MM_{trap}}\frac 1 2  w \nabb_A\qf^\F\c \nabb_A\qf
\eea
up to multiplication by functions of $r$ which are continuous at the photon sphere.
Indeed, integrating by parts \eqref{problematic-term-2} in $\nabb$, we obtain a term of this form, and \eqref{problematic-term-6} is already of that form. Consider \eqref{problematic-term-4}. Performing integration by parts in $T$ we obtain
\beaa
- C\int_{\MM_{trap}}\frac 1 2  w T\qf^\F \c T\left(-\frac{1}{r^3} \qf \right)&\les&  C\int_{\MM_{trap}}\frac 1 2  w TT\qf^\F \c \left(-\frac{1}{r^3} \qf \right)\\
&&+E_p[\qf](\tau) +E_p[\qf^\F](\tau)+E_p[\qf](0) +E_p[\qf^\F](0)
\eeaa
Substituting the expression for $TT\qf^\F$ using the wave equation for $\qf^\F$ as done before we obtain
\beaa
C\int_{\MM_{trap}}\frac 1 2  w TT\qf^\F \c \left(-\frac{1}{r^3} \qf \right)&=& C\int_{\MM_{trap}}\frac 1 2  w (RR\qf^\F  +\Up \lapp_2\qf^\F) \c \left(-\frac{1}{r^3} \qf \right)\\
&&+ C\int_{\MM_{trap}}\frac 1 2  w (\frac 2 r \Up R \qf^\F -V_2 \Up\qf^\F  +\Up\frac{Q}{r^3} \qf ) \c \left(-\frac{1}{r^3} \qf \right)\\
&&+ C\int_{\MM_{trap}}\frac 1 2  w ( -Q^2\Up \Lbb_2[\qf^\F]) \c \left(-\frac{1}{r^3} \qf \right)
\eeaa
The second line can be bounded by the Morawetz bulks and the third line will be bounded in Section \ref{section-lower-order}, since it involves lower order terms. For the first term in the first line we perform integration by parts in $R$ as done before, using a cut-off function $\chi(r)$. The remaining term with the laplacian can be brought into one of the form \eqref{general-form-problematic-form} by integration by parts in $\nabb$. Finally, term \eqref{problematic-term-4} is bounded by:
\beaa
\eqref{problematic-term-4} &\les& -C\int_{\MM_{trap}}\frac 1 2  w (\Up \nabb_A\qf^\F) \c \left(-\frac{1}{r^3} \nabb_A\qf \right)+ C\int_{\MM_{trap}}\frac 1 2  w ( -Q^2\Up \Lbb_2[\qf^\F]) \c \left(-\frac{1}{r^3} \qf \right)\\
&&+\MM_p[\qf](0,\tau) +\MM^{1, T, \nabb}_p[\qf^\F](0,\tau)+E_p[\qf](\tau) +E_p[\qf^\F](\tau)+E_p[\qf](0) +E_p[\qf^\F](0)
\eeaa
It remains therefore to estimate \eqref{general-form-problematic-form}. Recall that $R(r)=\Up$, therefore we write in the integral $1=\frac 1 \Up R(r-r_P)$ in order to  introduce the degeneracy at the photon sphere, allowing for taking an extra derivative in $R$:
\beaa
\int_{\Sigma_{trap}(\tau)} \chi(r)\nabb\qf^\F \c \nabb\qf&=&\int_{\Sigma_{trap}(\tau)}\frac{1}{\Up} R(r-r_P) \chi(r)\nabb\qf^\F \c \nabb\qf
\eeaa
We integrate by parts in $R$:
\beaa
&=&-\int_{\Sigma_{trap}(\tau)}R(\frac{1}{\Up}) (r-r_P) \chi(r)\nabb\qf^\F \c \nabb\qf-\int_{\Sigma_{trap}(\tau)} (r-r_P) \chi'(r)\nabb\qf^\F \c \nabb\qf\\
&&-\int_{\Sigma_{trap}(\tau)}\frac{1}{\Up} (r-r_P) \chi(r) R\nabb\qf^\F \c \nabb\qf-\int_{\Sigma_{trap}(\tau)}\frac{1}{\Up} (r-r_P) \chi(r)\nabb\qf^\F \c R\nabb\qf
\eeaa
Upon integration in $\tau$, each term can be bounded in the following way. Since $\Up'$ is bounded at the photon sphere, the first term is easily bounded by throwing the degeneracy into $\nabb \qf$. The second term vanishes at the photon sphere. The third term can be bounded by throwing the degeneracy into $\nabb \qf$ as above, since $R\nabb \qf^\F$ appears with no degeneracy on the Morawetz bulk. The last term can be bounded by performing an integration by parts in $\nabb$, and throwing the degeneracy into $\nabb^2 \qf^\F$, since $R\qf$ appears with no degeneracy in the Morawetz bulks. This finally gives 
\beaa
\eqref{problematic-term-2}+\eqref{problematic-term-4}+\eqref{problematic-term-6}&\les& C\int_{\MM_{trap}}\frac 1 2  w ( -Q^2\Up \L_2[\qf^\F]) \c \left(-\frac{1}{r^3} \qf \right)\\
&&+\MM_p[\qf](0,\tau) +\MM^{1, T, \nabb}_p[\qf^\F](0,\tau)\\
&&+E_p[\qf](\tau) +E_p[\qf^\F](\tau)+E_p[\qf](0) +E_p[\qf^\F](0)
\eeaa
which bounds all the coupling terms at the photon sphere.

Putting all the above estimates together, we proved estimate \eqref{global-estimate-coupling}.

Observe that the procedure of the absorption of the coupling terms in the case of spin $-2$ quantities is completely similar. Conditions \eqref{condition-cancellation} can still be verified because the right hand sides of the generalized Regge-Wheeler system of spin $-2$ still have opposite signs.

\subsection{Estimates for the lower order terms}\label{section-lower-order}
We will absorb the lower order terms by the Morawetz bulks of $\qf$ and $\qf^\F$, and by allowing a necessary bound by the initial energy of the lower order terms quantities $\psi_0$, $\psi_1$ and $\psi_3$. 

The goal of this section is to prove the following estimate for the lower order terms:
\bea
\begin{split}\label{main-estimate-lower-order-terms}
&\II_{p\,; \,R}[ \Lbb_1[\qf^\F]](0,\tau)+\II_{p\,; \,R}[ \Lbb_1[\qf]](0,\tau)+  \II^{1, T, \nabb}_{p\,; \,R}[\Lbb_2[\qf^\F]]( 0,\tau)\\
& - A\int_{\MM(0,\tau)} \left((r-r_P)R(\qf)+ \Lambda T(\qf )+\frac 1 2  w\qf\right)\c \Lbb_1[\qf^\F]\\
& - A\int_{\MM(0,\tau)} \left((r-r_P)R(\qf)+ \Lambda T(\qf )+\frac 1 2  w\qf\right)\c \Lbb_1[\qf]\\
          &- B\int_{\MM(0,\tau)} \left((r-r_P)R(\qf^\F)+ \Lambda T(\qf^\F )+\frac 1 2  w\qf^\F\right)\c \Lbb_2[\qf^\F]\\
        &- C\int_{\MM(0,\tau)} \left((r-r_P)R(T\qf^\F)+ \Lambda T(T\qf^\F )+\frac 1 2  w T\qf^\F\right)\c T(\Lbb_2[\qf^\F])\\
        &- D\int_{\MM(0,\tau)} \left((r-r_P)R(r\nabb_A\qf^\F)+ \Lambda T(r\nabb_A\qf^\F )+\frac 1 2  w r\nabb_A\qf^\F\right)\c r\nabb_A(\Lbb_2[\qf^\F])\\
        &- C\int_{\MM_{trap}}\Lambda ( -\Up Q^2 \Lbb_2[\qf^\F])\c T\left(-\frac{1}{r^3} \qf \right)+ C\int_{\MM_{trap}}\frac 1 2  w ( -Q^2\Up \Lbb_2[\qf^\F]) \c \left(-\frac{1}{r^3} \qf \right)\\
        &\les \MM_p[\qf](0,\tau)+\MM^{1, T, \nabb}_p[\qf^\F](0,\tau)+E^{1, T, \nabb}_p[\qf^\F](\tau)+E_p[\qf](\tau)\\
&+E_p[\qf](0)+E^{1, T, \nabb}_p[\qf^\F](0)+E_p[\ff](0)+E_p[\psi_1](0)+E_p[\a](0)
\end{split}
\eea
Notice that in the seventh line of \eqref{main-estimate-lower-order-terms}, we inserted the additional terms involving the lower order terms obtained in the coupling estimates in the previous section.

We first derive transport estimates for the lower order terms, and then we estimate the lower order terms of the main estimates above.

\subsubsection{Transport estimates for lower order terms}
We derive transport estimates for $\psi_0$, $\psi_1$ and $\psi_3$ using their differential relations. In order to simplify the derivation of the transport estimates, applied few times, we summarize in the following lemma the main computation.

\begin{lemma}\label{general-lemma-transport} Let $\Phi_1$ and $\Phi_2$ be two symmetric traceless $2$-tensor which verify the differential relation
\bea\label{generic-differential-relation}
\nabb_3(r^n \kab^m \Phi_1)&=& \kab \Phi_2
\eea
for $n$, $m$ integers. Then, in the far away region, for all $l >0$
\beaa
\int_{\Sigma_{far}(\tau)} r^{2n-2m+l-2} |\Phi_1|^2+\int_{\MM_{far}(0, \tau)}r^{2n-2m+l-3}|\Phi_1|^2 &\les& \int_{\Sigma_{far}(0)} r^{2n-2m+l-2} |\Phi_1|^2\\
&&+ \int_{\MM_{far}(0, \tau)} r^{l-3} |\Phi_2|^2
\eeaa
\end{lemma}
\begin{proof} From \eqref{generic-differential-relation}, we obtain
\beaa
\nabb_3(r^{2n}\kab^{2m} |\Phi_1|^2)&=& 2r^n \kab^{m+1} \Phi_1  \Phi_2
\eeaa
Multiplying by $r^{l-2}$ and recalling that $\nabb_3r=\frac 1 2 r \kab$ and $\kab \leq 0$, we obtain
\beaa
&&\nabb_3(r^{2n+l-2}\kab^{2m} |\Phi_1|^2)+\left(\frac{l}{2}-1\right)r^{2n+l-2}(-\kab)^{2m+1} |\Phi_1|^2\\
&=& 2r^{n+l-2} \kab^{m+1} \Phi_1  \Phi_2\leq \frac{l}{4} r^{2n+l-2}(-\kab)^{2m+1}|\Phi_1|^2+\frac{4}{l} r^{l-2} (-\kab) |\Phi_2|^2
\eeaa
giving 
\beaa
\nabb_3(r^{2n+l-2}\kab^{2m} |\Phi_1|^2)+\left(\frac{l}{4}-1\right)r^{2n+l-2}(-\kab)^{2m+1} |\Phi_1|^2&\leq& \frac{4}{l} r^{l-2} (-\kab) |\Phi_2|^2
\eeaa
Since $\div e_3=\frac 1 2 \tr \pi^{(3)}=\kab -\frac{2M}{r^2}+\frac{2Q^2}{r^3}$, we write
\beaa
&&\div(r^{2n+l-2}\kab^{2m} |\Phi_1|^2 e_3)+\left(\frac{l}{4}-1\right)r^{2n+l-2}(-\kab)^{2m+1} |\Phi_1|^2\\
&=&\nabb_3(r^{2n+l-2}\kab^{2m} |\Phi_1|^2)+\left(\frac{l}{4}-1\right)r^{2n+l-2}(-\kab)^{2m+1} |\Phi_1|^2+r^{2n+l}\kab^{2m} |\Phi_1|^2 \div e_3\\
&\leq& \frac{4}{l} r^{l-2} (-\kab) |\Phi_2|^2+r^{2n+l}\kab^{2m+1} |\Phi_1|^2 
\eeaa
which gives
\beaa
&&\div(r^{2n+l-2}\kab^{2m} |\Phi_1|^2 e_3)+\frac{l}{4} r^{2n+l-2}(-\kab)^{2m+1} |\Phi_1|^2 \leq \frac{4}{l} r^{l-2} (-\kab) |\Phi_2|^2 
\eeaa
Integrating the above inequality on $\MM_{far}(0, \tau)$ and using divergence theorem on the left hand side, we obtain the desired estimate, where we allow for a constant on the right hand side depending on $R_0$.
\end{proof}

We will make use of Lemma \ref{general-lemma-transport} to obtain the control, and the optimal decay, of the lower order terms through transport equations. The differential relations given by \eqref{quantities} are not sufficient to gain control of all derivatives of the lower order terms. In order to do so, we will make use of Bianchi identities and commutation formulas.

\begin{proposition}\label{transport-estimates}[Transport estimates for $\a$, $\psi_1$ and $\ff$] Let $\a$, $\psi_1$, $\qf$, $\ff$, $\qf^\F$ be defined as in \eqref{quantities}. Then, for all $\de \le p \le 2-\de$, we have
\beaa
&&E_p[\a](\tau)+E_p[\psi_1](\tau)+E^{1, T, \nabb}_p[\ff](\tau)+ E_{\mathcal{H}^+}[\a](0, \tau)+E_{\mathcal{H}^+}[ \psi_1](0, \tau)+E_{\mathcal{H}^+}^{1, T, \nabb}[\ff](0, \tau)\\
&&+\widehat{\MM}_p[\a](0, \tau)+ \widehat{\MM}_p[\psi_1](0, \tau)+ \widehat{\MM}^{1, T, \nabb}_p[\ff](0, \tau)\\
&\les& \MM_p[\qf](0,\tau)+\MM^{1, T, \nabb}_p[\qf^\F](0,\tau)+E_p[\a](0)+ E_p[\psi_1](0)+ E^{1, T, \nabb}_p[\ff](0)
\eeaa
\end{proposition}
\begin{proof} We separate the transport estimates into the three regions: the redshift region, the trapping region and the far-away region. 

\bigskip

{\bf{In the far-away region:}} In this region, we want to get the best decay in $r$. We will apply Lemma \ref{general-lemma-transport}, being careful to exploit the best decay on the left hand side.

By \eqref{quantities}, $\ff$ and $\qf^\F$ verify relation \eqref{generic-differential-relation} with $n=3$, $m=1$. Since  $\int_{\MM_{far}(0, \tau)} r^{-3+p}|\qf^\F|^2$ is controlled by  the Morawetz bulk in the far away region by definition \eqref{Morawetz-far-away}, by Lemma \ref{general-lemma-transport}, for $l=p$,
\beaa
\int_{\Sigma_{far}(\tau)} r^{2+p} |\ff|^2+\int_{\MM_{far}(0, \tau)}r^{1+p}|\ff|^2 \les \int_{\Sigma_{far}(0)} r^{2+p} |\ff|^2+ \MM^{1, T, \nabb}_p[\qf^\F](0,\tau)
\eeaa

By \eqref{quantities}, $\psi_1$ and $\qf$ verify relation \eqref{generic-differential-relation} with $n=1$, $m=0$. As before, applying Lemma \ref{general-lemma-transport}, for $l=p$, we obtain
\beaa
\int_{\Sigma_{far}(\tau)} r^{p} |\psi_1|^2+\int_{\MM_{far}(0, \tau)}r^{-1+p}|\psi_1|^2 \les \int_{\Sigma_{far}(0)} r^{p} |\psi_1|^2+ \MM_p[\qf](0,\tau)
\eeaa

By \eqref{quantities}, $\a$ and $\psi_1$ verify relation \eqref{generic-differential-relation} with $n=3$, $m=2$. As before, applying Lemma \ref{general-lemma-transport}, for $l=2+p$, we obtain
\beaa
\int_{\Sigma_{far}(\tau)} r^{2+p} |\a|^2+\int_{\MM_{far}(0, \tau)}r^{1+p}|\a|^2 \les \int_{\Sigma_{far}(0)} r^{2+p} |\a|^2+ \int_{\MM_{far}(0, \tau)} r^{-1+p} |\psi_1|^2
\eeaa
We estimated the zero-th order terms of the non-degenerate Morawetz bulks for $\a$, $\psi_1$ and $\ff$. We can estimate the $\nabb_3$ derivative, by making use of the differential relations \eqref{quantities}. Indeed, we immediately obtain bounds on $\int_{\MM_{far}(0, \tau)} r^{p+3}|\nabb_3\ff|^2$, $\int_{\MM_{far}(0, \tau)} r^{p+1}|\nabb_3\psi_1|^2$ and $\int_{\MM_{far}(0, \tau)} r^{p+3}|\nabb_3\a|^2$. 

To estimate the angular derivative, we can commute relations \eqref{quantities} with $r\nabb$. We obtain $\nabb_3(r^4 \kab \nabb \ff)=r\kab \nabb \qf^\F$, and therefore we can apply Lemma \ref{general-lemma-transport} to $\nabb\ff$ and $\nabb\qf^\F$ with $n=4$, $m=1$ and $l=p$. 
We can similarly commute the relations for $\psi_1$ and $\a$.
To estimate the $\nabb_4$ derivative, we take $\check{\nabb}_4$ of relations \eqref{quantities}. For example we compute:
\beaa
\nabb_3(\Up \nabb_4(r\psi_1))= \kab \Up \check{\nabb}_4\qf-\Up(\frac 2 r \kab -2\rho) \qf
\eeaa
Writing $\Up=\frac{r}{2}(-\kab)$, we can apply Lemma \ref{general-lemma-transport} to $\nabb_4\psi_1$ and $\check{\nabb}_4\qf$ and obtain
\beaa
\int_{\Sigma_{far}(\tau)} r^{2+p} |\check{\nabb}_4 \psi_1|^2+\int_{\MM_{far}(0, \tau)}r^{1+p}|\check{\nabb}_4 \psi_1|^2 &\les& \int_{\Sigma_{far}(0)} r^{2+p} |\check{\nabb}_4 \psi_1|^2+ \MM_p[\qf](0,\tau)
\eeaa
Similarly for $\a$, and for $\ff$ we can in addition commute with $T$ and $r\nabb$, and still be able to bound by $\MM^{1, T, \nabb}_p[\qf^\F](0,\tau)$.

Putting all the above together we obtain in the far-away region:
\beaa
 E^{1, T, \nabb}_{p\,; \,R}[\ff](\tau)+\widehat{\MM}^{1, T, \nabb}_{p\,; \,R}[\ff](0, \tau)&\les&  E^{1, T, \nabb}_{p\,; \,R}[\ff](0)+\MM^{1, T, \nabb}_p[\qf^\F](0,\tau), \\
 E_{p\,; \,R}[\psi_1](\tau)+\widehat{\MM}_{p\,; \,R}[\psi_1](0, \tau)&\les&  E_{p\,; \,R}[\psi_1](0)+\MM_p[\qf](0,\tau), \\
 E_{p\,; \,R}[\a](\tau)+\widehat{\MM}_{p\,; \,R}[\a](0, \tau)&\les&  E_{p\,; \,R}[\a](0)+\widehat{\MM}_{p\,; \,R}[\psi_1](0, \tau)
\eeaa

\bigskip

{\bf{In the trapping region:}} In this bounded region, we don't keep track of the powers of $r$, but we have to be careful about the degeneracy of the $T$ and $\nabb$ derivatives. The zero-th order terms of $\a$ and $\psi_1$ are straightforwardly bounded by the Morawetz bulks using Lemma \ref{general-lemma-transport} applied to \eqref{quantities}. For $\ff$, commuting with $T$ and $r\nabb$,  we can use Lemma \ref{general-lemma-transport} to also bound those derivatives in terms of the Morawetz bulk.

Since the Morawetz bulk $\MM_p[\qf](0,\tau)$ contains derivatives $T$ and $\nabb$ which are degenerate at the photon sphere, we don't commute with those Killing vector fields. On the contrary, since $R$ appears with no degeneracy on the bulk, we commute with $R$ instead. Recalling that $\omb=\frac 1 2  \Up'$, we obtain
 \beaa
 [R, e_3 ]&=& [\frac 1 2 (-e_3+\Up e_4), e_3]=\frac 1 2 [\Up e_4, e_3]= \frac 1 2 (\Up [e_4, e_3] - e_3(\Up) e_4)= \frac 1 2 (-2\omb \Up e_4 +\Up' \Up  e_4) =0
 \eeaa
Commuting \eqref{quantities} with $R$ and applying Proposition \ref{transport-estimates} we obtain 
 \beaa
 \int_{\Si_{trap}(\tau)}|R\psi_1|^2 +\int_{\MM_{trap}(0,\tau)} |R\psi_1|^2 \les \int_{\Si_{trap}(0)}|R\psi_1|^2+\MM_p[\qf](0,\tau)
 \eeaa
Since $T\psi_1=\frac 1 r \kab \qf-\frac 1 2 \kab \psi_1+R\psi_1$, we can bound $|T\psi_1|^2$ by the bulk in $\qf$ and initial data. Recalling the definition of $E[\psi_1]$ and $\widehat{\MM}[\psi_1]$ we have obtained
\beaa
E[\psi_1](\tau)+\widehat{\MM}[\psi_1](0, \tau) \les E[\psi_1](0)+\MM_p[\qf](0,\tau)
\eeaa
 A similar computation done for $\a$ gives
 \beaa
 \int_{\Si_{trap}(\tau)}|R\a|^2+|T\a|^2 +\int_{\MM_{trap}(0,\tau)} |R\a|^2+|T\a|^2 &\les& \int_{\MM_{trap}(0,\tau)}|R\psi_1|^2  +\int_{\Si_{trap}(0)}|R\a|^2+|T\a|^2
  \eeaa
It only remains to control the angular derivative of $\a$. Using Corollary \ref{laplacian-a}, we can write the laplacian of $\a$ in terms of first and zero-th order derivatives of $\psi_1$, in terms of $\a$ and in terms of first null derivatives of $\ff$, which are controlled by the Morawetz bulks. Indeed, using elliptic estimates, we obtain
 \beaa
 \sum_{i=0}^2\int_{\Si_{trap}(\tau)}|\nabb^i\a|^2+\int_{\MM_{trap}(0,\tau)} |\nabb^i\a|^2 \les  \MM_p[\qf](0,\tau)+\MM^{1, T, \nabb}_p[\qf^\F](0,\tau)+\sum_{i=0}^2\int_{\Si_{trap}(0)}|\nabb^i\a|^2
 \eeaa

 Commuting $\kab \qf^\F= \nabb_3(r^3 \kab \ff)$ with $R$, $T$ and $r \nabb$ and applying Proposition \ref{transport-estimates}  we obtain bounds for $R\ff$, $RT\ff$ and $R \nabb\ff$). Since $RT\qf^\F$ and $R\nabb\qf^\F$ appear with no degeneracy on $\MM^{1, T, \nabb}_p[\qf^\F](0,\tau)$, we can use it to bound the desired norms of $\ff$. Observe that we can write 
 \bea\label{Tpsi-3-as-R}
T\psi_3&=&\nabb_3\psi_3+R\psi_3=\frac 1 r \kab \qf^\F-\frac 1 2 \kab \psi_3+R\psi_3
\eea
 Commuting \eqref{Tpsi-3-as-R} with $T$, we obtain $TT\psi_3=\frac 1 r \kab T\qf^\F-\frac 1 2 \kab T\psi_3+RT\psi_3$ which therefore gives control on $TT\ff$ on the trapping region. Similarly,  commuting \eqref{Tpsi-3-as-R} with $r\nabb$, we obtain $r\nabb T\psi_3=\frac 1 r \kab r\nabb\qf^\F-\frac 1 2 \kab r\nabb\psi_3+r\nabb R\psi_3$ which therefore gives control on $T\nabb\ff$ on the trapping region.  Using Corollary \ref{laplacian-psi-3}, we can write the laplacian of $\ff$ in terms of first derivative of $\qf^\F$ and zero-th order terms for $\ff$, $\psi_1$, $\a$, which are controlled by the bulk norms. Therefore,  
 \beaa
 \sum_{i=0}^2\int_{\Si_{trap}(\tau)}|\nabb^i\ff|^2+\int_{\MM_{trap}(0,\tau)} |\nabb^i\ff|^2 \les  \MM_p[\qf](0,\tau)+\MM^{1, T, \nabb}_p[\qf^\F](0,\tau)+\sum_{i=0}^2\int_{\Si_{trap}(0)}|\nabb^i\ff|^2
 \eeaa
 This proves the desired bound in the trapping region:
 
 \bigskip

{\bf{In the red-shift region:}} In red-shift region,  we can get estimates for the $\nabb_3$ derivatives of $\a$, $\psi_1$, $\ff$ again using the definitions \eqref{quantities}. At the horizon $T$ becomes parallel to $\nabb_4$, therefore in order to control the $\nabb_4$ derivatives it suffices to commute \eqref{quantities} by $T$. By applying the divergence theorem, we get the desired bound near the horizon.

Putting together the estimates above, we prove the Proposition. 
\end{proof}

We shall use Proposition \ref{transport-estimates} to obtain estimates for the lower order terms in the main combined estimate \eqref{estimate3}. 

\subsubsection{Absorption in the far-away region}

The goal of this subsection is to prove the following estimate
\bea\label{absorption-lower-order-far-away}
\begin{split}
&\II_{p\,; \,R}[ \Lbb_1[\qf^\F]](0,\tau)+\II_{p\,; \,R}[ \Lbb_1[\qf]](0,\tau)+  \II^{1, T, \nabb}_{p\,; \,R}[\Lbb_2[\qf^\F]]( 0,\tau) \\
&\les \MM_p[\qf](0,\tau)+\MM^{1, T, \nabb}_p[\qf^\F](0,\tau)+ E_p[\ff](0)+ E_p[\psi_1](0)+ E_p[\a](0)
\end{split}
\eea
In the far-away region, we will need to keep track of the powers of $r$. In particular:
\bea\label{schematic-lower-order-terms}
 \Lbb_1[\qf]&\simeq& \Big[\frac{1}{r^3}\psi_1, \frac{1}{r^2} \psi_0\Big], \\
 \Lbb_1[\qf^\F], \Lbb_2[\qf^\F]&\simeq & \Big[\frac{1}{r^3} \psi_3 \Big]
\eea
Therefore, we have
\beaa
\II_{p\,; \,R}[ \Lbb_1[\qf^\F]](0,\tau)&=& \int_{\MM_{far}(0,\tau)} r^{1+p}  |\frac{1}{r^3} \psi_3|^2=  \int_{\MM_{far}(0,\tau)} r^{-5+p}|\psi_3|^2\\
&=&  \int_{\MM_{far}(0,\tau)} r^{-3+p}|\ff|^2
\eeaa
which decays faster than the non-degenerate bulk norm, therefore 
\beaa
\II_{p\,; \,R}[ \Lbb_1[\qf^\F]](0,\tau)+\II_{p\,; \,R}[ \Lbb_2[\qf^\F]](0,\tau)&\les & \widehat{\MM}_p[\ff](0, \tau)
\eeaa
Similarly, 
\beaa
\II^{1, T, \nabb}_{p\,; \,R}[\Lbb_2[\qf^\F]]( 0,\tau)&=& \int_{\MM_{far}(0,\tau)} r^{1+p}  |\frac{1}{r^3} \psi_3, \frac{1}{r^3} T\psi_3, \frac{1}{r^3}(r\nabb) \psi_3 |^2\\
&=& \int_{\MM_{far}(0,\tau)}r^{-3+p}|\ff|^2+ r^{-3+p}|T\ff|^2+r^{-1+p}|\nabb\ff|^2 \les  \widehat{\MM}_p[\ff](0, \tau)
\eeaa
Finally 
\beaa
\II_{p\,; \,R}[ \Lbb_1[\qf]](0,\tau)&=& \int_{\MM_{far}(0,\tau)} r^{1+p}  |\frac{1}{r^3}\psi_1, \frac{1}{r^2} \psi_0|^2 = \int_{\MM_{far}(\tau_1,\tau_2)} r^{-5+p}|\psi_1|^2+r^{-3+p} |\a|^2\\
&\les& \widehat{\MM}_p[\psi_1](0, \tau) +\widehat{\MM}_p[\a](0, \tau)
\eeaa
The above estimates give 
\beaa
\II_{p\,; \,R}[ \Lbb_1[\qf^\F]](0,\tau)+\II_{p\,; \,R}[ \Lbb_1[\qf]](0,\tau)+  \II^{1, T, \nabb}_{p\,; \,R}[\Lbb_2[\qf^\F]]( 0,\tau) &\les& \widehat{\MM}_p[\ff](0, \tau)+\widehat{\MM}_p[\psi_1](0, \tau) +\widehat{\MM}_p[\a](0, \tau)
\eeaa
and using Proposition \ref{transport-estimates}, we finally get \eqref{absorption-lower-order-far-away}.

\subsubsection{Absorption outside the photon sphere}\label{outside-trapping-lower-subsub}
The goal of this subsection is to prove the following estimate
\bea\label{absorption-lower-outside-trapping}
\begin{split}
&\int_{\MM(0,\tau)\setminus r=r_p} \operatorname{integrals \ in \ LHS \ of \eqref{main-estimate-lower-order-terms}} \les \MM_p[\qf](0,\tau)+\MM^{1, T, \nabb}_p[\qf^\F](0,\tau)\\
&+E_p[\qf](0)+E^{1, T, \nabb}_p[\qf^\F](0)+E_p[\ff](0)+E_p[\psi_1](0)+E_p[\a](0)
\end{split}
\eea
As for the coupling terms, the integrals outside the photon sphere can be easily bounded by $\MM_p[\qf](0,\tau) +\MM^{1, T, \nabb}_p[\qf^\F](0,\tau)$ and $\widehat{\MM}_p[\ff](0, \tau)+\widehat{\MM}_p[\psi_1](0, \tau) +\widehat{\MM}_p[\a](0, \tau)$ using Cauchy-Schwarz. This can be done for all the terms in the integrals in the left hand side of \eqref{main-estimate-lower-order-terms}, and using again Proposition \ref{transport-estimates}, we obtain \eqref{absorption-lower-outside-trapping}.

\subsubsection{Absorption at the photon sphere}
The goal of this subsection is to prove the following estimate
\bea\label{absorption-lower-trapping}
\begin{split}
&\int_{\MM_{trap}(0,\tau)} \operatorname{integrals \ in \ LHS \ of \eqref{main-estimate-lower-order-terms}} \\
&\les \MM_p[\qf](0,\tau)+\MM^{1, T, \nabb}_p[\qf^\F](0,\tau)+E_p[\qf](\tau)+E^{1, T, \nabb}_p[\qf^\F](\tau)\\
&+E_p[\qf](0)+E^{1, T, \nabb}_p[\qf^\F](0)+ E_p[\ff](0)+ E_p[\psi_1](0)+ E_p[\a](0)
\end{split}
\eea
We analyze each line of the six lines of integrals on the left hand side of \eqref{main-estimate-lower-order-terms}. Recall that the Morawetz bulks for the lower order terms $\a$, $\psi_1$ and $\ff$ have no degeneracy. This simplifies the analysis with respect to the coupling terms. 

\bigskip

{\bf{First line}} Consider the first line of integrals in \eqref{main-estimate-lower-order-terms}. This gives
\beaa
- A\int_{\MM_{trap}} \left((r-r_P)R(\qf)+ \Lambda T(\qf )+\frac 1 2  w\qf\right)\c (-12\rho-40\rhoF^2)\psi_3
\eeaa
The terms involving $R\qf$ and $\qf$ can be absorbed by the Morawetz bulk in $\qf$ by Cauchy-Schwarz.
The term involving $T$ can be absorbed by integrating by parts in $T$, and $T\psi_3$ is controlled by the non-degenerate bulk norm in $\ff$.

\bigskip

{\bf{Second line}} Consider the second line of integrals in \eqref{main-estimate-lower-order-terms}. This gives
\beaa
 - A\int_{\MM_{trap}} \left((r-r_P)R(\qf)+ \Lambda T(\qf )+\frac 1 2  w\qf\right)\c \left(-\frac{2}{r^2} \psi_0-\frac{4}{r^3}\psi_1 \right)
\eeaa
The terms involving $R\qf$ and $\qf$ can be again absorbed by the Morawetz bulks by Cauchy-Schwarz. For the term involving $T\qf$, we integrate by parts in $T$, and obtain:
\beaa
&&- A\int_{\MM_{trap}} \left( \Lambda T(\qf )\right)\c \left(\psi_0+\psi_1 \right)\les A\int_{\MM_{trap}} \left( \Lambda \qf \right)\c \left(T\psi_0+T\psi_1 \right)\\
&&+E_p[\a](\tau)+E_p[\a](0)+E_p[\psi_1](\tau)+E_p[\psi_1](0)+E_p[\qf](\tau)+E_p[\qf](0)
\eeaa
We now make use of the fact that $\psi_0$ and $\psi_1$ are lower order terms with respect to $\qf$. Therefore we can write
\beaa
T\psi_0&=&\nabb_3\psi_0+R\psi_0=\frac 1 r \kab \psi_1-\frac 1 2 \kab \psi_0+R\psi_0, \\
T\psi_1&=&\nabb_3\psi_1+R\psi_1=\frac 1 r \kab \qf-\frac 1 2 \kab \psi_1+R\psi_1
\eeaa
We therefore obtain, ignoring the powers of $r$:
\beaa
A\int_{\MM_{trap}} \left( \Lambda \qf \right)\c \left(T\psi_0+T\psi_1 \right)&=& A\int_{\MM_{trap}} \left( \Lambda \qf \right)\c \left( \psi_1+\psi_0+R\psi_0+ \qf+R\psi_1\right)
\eeaa
The term given by $\psi_1+\psi_0+\qf$ can be bounded by the Morawetz bulks directly. The terms of the form $\qf \c (R\psi_0+R\psi_1)$ can be bounded by integrating per parts in $R$, using a cut-off function $\chi(r)$, and therefore using the Morawetz bulk in $\qf$ which is non-degenerate for $R\qf$. 

\bigskip

{\bf{Third line}} Consider the third line of integrals in \eqref{main-estimate-lower-order-terms}. This gives
\beaa
- B\int_{\MM_{trap}} \left((r-r_P)R(\qf^\F)+ \Lambda T(\qf^\F )+\frac 1 2  w\qf^\F\right)\c \left(\frac{4}{r^3} \psi_3\right)
\eeaa
Each term can directly be bounded by the Morawetz bulks and the non-degenerate bulk norm. 

\bigskip

{\bf{Fourth line}} Consider the fourth line of integrals in \eqref{main-estimate-lower-order-terms}. This gives
\beaa
- C\int_{\MM_{trap}}\left((r-r_P)R(T\qf^\F)+ \Lambda T(T\qf^\F )+\frac 1 2  w T\qf^\F\right)\c T \left(\frac{4}{r^3} \psi_3\right)
\eeaa
The terms involving $R(T\qf^\F)$ and $T\qf^\F$ can be absorbed by the Morawetz bulk in $\qf^\F$ by Cauchy-Schwarz.
For the terms involving $TT\qf^\F$, we can use the fact that $\psi_3$ is a lower order term with respect to $\qf^\F$. Using \eqref {Tpsi-3-as-R}, we obtain a term of this form
\beaa
- C\int_{\MM_{trap}}\left( \Lambda T(T\qf^\F )\right)\c ( \qf^\F+ \psi_3+R\psi_3)
\eeaa
Integrating by parts in $T$ the first product, we obtain $\int |T\qf^\F|^2$, which is bounded by the Morawetz bulk. Integrating by parts in $T$ the second product, we obtain $\int T\qf^\F \c T\psi_3$, which can be bounded by Morawetz bulks. Finally, integrating by parts in $R$ and in $T$ the last product, we obtain $\int RT\qf^\F \c T\psi_3$, which can be bounded in the same way. In the integrating by parts process we pick up some energies as boundary terms.

\bigskip 

{\bf{Fifth line}} Consider the fifth line of integrals in \eqref{main-estimate-lower-order-terms}. This gives
\beaa
- D\int_{\MM_{trap}} \left((r-r_P)R(r\nabb_A\qf^\F)+ \Lambda T(r\nabb_A\qf^\F )+\frac 1 2  w r\nabb_A\qf^\F\right)\c r\nabb_A\left(\frac{4}{r^3} \psi_3\right)
\eeaa
As above, the terms involving $R\nabb \qf^\F$ and $\nabb\qf^\F$ can be bounded by the Morawetz bulk using Cauchy-Schwarz. For the term involving $T\nabb \qf^\F$ we perform integration by parts in $T$ and $\nabb$:
\beaa
- D\int_{\MM_{trap}}\left(\Lambda T(r\nabb_A\qf^\F )\right)\c r\nabb_A\left(\frac{4}{r^3} \psi_3\right)&\les & D\int_{\MM_{trap}} \left(\Lambda (r\nabb_A\qf^\F )\right)\c rT\nabb_A\left(\frac{4}{r^3} \psi_3\right)\\
&&+E^{1, T, \nabb}_p[\qf^\F](\tau)+E^{1, T, \nabb}_p[\qf^\F](0)+ E_p[\ff](\tau)+E_p[\ff](0)\\
&\les & -D\int_{\MM_{trap}}\left(\Lambda (r^2\nabb^2_A\qf^\F )\right)\c T\left(\frac{4}{r^3} \psi_3\right)\\
&&+E^{1, T, \nabb}_p[\qf^\F](\tau)+E^{1, T, \nabb}_p[\qf^\F](0)+ E_p[\ff](\tau)+E_p[\ff](0)
\eeaa
We write $T\psi_3=\frac 1 r \kab \qf^\F-\frac 1 2 \kab \psi_3+R\psi_3$ and therefore obtain 
\beaa
-D\int_{\MM_{trap}}\left(\Lambda (r^2\nabb^2_A\qf^\F )\right)\c T\left( \psi_3\right)&=& -D\int_{\MM_{trap}}\left(\Lambda (r^2\nabb^2_A\qf^\F )\right)\c ( \qf^\F+ \psi_3+R\psi_3)
\eeaa
Integrating by parts in $r\nabb$ the first product, we obtain $\int |(r\nabb)\qf^\F|^2$, which is bounded by the Morawetz bulk. Integrating by parts in $r\nabb$ the second product, we obtain $\int r\nabb\qf^\F \c r\nabb\psi_3$, which can be bounded by Morawetz bulks. Finally, integrating by parts in $\nabb$ and in $R$ the last product, we obtain $\int R\nabb\qf^\F \c \nabb\psi_3$, which can be bounded in the same way.

\bigskip

{\bf{Sixth line}} Consider the sixth line of integrals in \eqref{main-estimate-lower-order-terms}. This gives
\beaa
- C\int_{\MM_{trap}}\Lambda ( -\Up Q^2 \left(\frac{4}{r^3} \psi_3\right))\c T\left(-\frac{1}{r^3} \qf \right)+ C\int_{\MM_{trap}}\frac 1 2  w ( -Q^2\Up \left(\frac{4}{r^3} \psi_3\right)) \c \left(-\frac{1}{r^3} \qf \right)
\eeaa
Integrating the first term by parts in $T$, we can write it as $\int T\psi_3 \c \qf$ which can be bounded by the Morawetz bulks and non-degenerate bulk. The second term can be bounded using Cauchy-Schwarz by the Morawetz bulks. 

To summarize all the cases above we obtained
\beaa
\int_{\MM_{trap}(0,\tau)} \operatorname{integrals \ in \ LHS \ of \eqref{main-estimate-lower-order-terms}} &\les& \MM_p[\qf](0,\tau)+\MM^{1, T, \nabb}_p[\qf^\F](0,\tau)+E^{1, T, \nabb}_p[\qf^\F](\tau)+E_p[\qf](\tau)\\
&&+E_p[\qf](0)+E^{1, T, \nabb}_p[\qf^\F](0)+E_p[\ff](0)+E_p[\psi_1](0)+E_p[\a](0)\\
&&+\widehat{\MM}_p[\ff](0, \tau) +\widehat{\MM}_p[\a](0, \tau)+\widehat{\MM}_p[\psi_1](0, \tau)\\
&&+E_p[\a](\tau)+ E_p[\ff](\tau)+E_p[\psi_1](\tau)
\eeaa
Using Proposition \ref{transport-estimates}, we finally proved \eqref{absorption-lower-outside-trapping}.

Putting together the estimates for the lower order terms in the far away region, outside the photon sphere and at the photon sphere, we obtain the proof of estimate \eqref{main-estimate-lower-order-terms}.

\section{Proof of the Main Theorem}\label{proof-of-theorem-this-is-the-end}
We prove here the Main Theorem for spin $+2$.

\subsection{Proof of estimate \eqref{first-estimate-main-theorem-1}}

We derive here estimate \eqref{first-estimate-main-theorem-1} for the system. 

Recall estimate \eqref{estimate3} which depends on positive constants $A, B, C, D$. Pick $A$, $C$ and $D$ verifying conditions \eqref{condition-cancellation}. Recalling the definition of $\M_1[\qf, \qf^\F]$ and $\M_2[\qf, \qf^\F]$ in \eqref{schematic-system}, then estimate \eqref{estimate3} can be written as 
 \beaa
       && E_p[\qf](\tau)  + E^{1, T, \nabb}_p[\qf^\F](\tau)+ E_{\mathcal{H}^+}[\qf](0, \tau)  + E_{\mathcal{H}^+}^{1, T, \nabb}[\qf^\F](0,\tau)  +E_{\mathcal{I}^+, p}[\qf](0, \tau) +E^{1, T, \nabb}_{\mathcal{I}^+, p}[\qf^\F](0, \tau)   \\
       &&    +\MM_p[\qf](0,\tau) +\MM^{1, T, \nabb}_p[\qf^\F](0,\tau)\\
       &\les &E_{p}[\qf](0)+E^{1, T, \nabb}_p[\qf^\F](0)+Q\left(\operatorname{ LHS \ of \ \eqref{global-estimate-coupling}}+ \operatorname{ LHS \ of \ \eqref{main-estimate-lower-order-terms}}\right)
  \eeaa

     By estimate \eqref{global-estimate-coupling} and \eqref{main-estimate-lower-order-terms}, we obtain
  \beaa
       && E_p[\qf](\tau)  + E^{1, T, \nabb}_p[\qf^\F](\tau)  + E_{\mathcal{H}^+}[\qf](0, \tau)  + E_{\mathcal{H}^+}^{1, T, \nabb}[\qf^\F](0,\tau)  +E_{\mathcal{I}^+, p}[\qf](0, \tau) +E^{1, T, \nabb}_{\mathcal{I}^+, p}[\qf^\F](0, \tau)            \\
        &&+\MM_p[\qf](0,\tau) +\MM^{1, T, \nabb}_p[\qf^\F](0,\tau)\\
        &\les& E_{p}[\qf](0)+E^{1, T, \nabb}_p[\qf^\F](0)\\
       &&+Q\Big(\MM_p[\qf](0,\tau)+\MM^{1, T, \nabb}_p[\qf^\F](0,\tau)+E^{1, T, \nabb}_p[\qf^\F](\tau)+E_p[\qf](\tau)\\
&&+E_p[\qf](0)+E^{1, T, \nabb}_p[\qf^\F](0)+E_p[\ff](0)+E_p[\psi_1](0)+E_p[\a](0)\Big)
  \eeaa
  For $Q \ll M$ we can absorb the second line of the right hand side of the above estimate on the left hand side, and finally obtain \eqref{first-estimate-main-theorem-1}.

 \subsection{Proof of estimate \eqref{first-estimate-a-f-main-theorem-1}}
 Combining Proposition \ref{transport-estimates} and estimate \eqref{first-estimate-main-theorem-1}, we immediately obtain \eqref{first-estimate-a-f-main-theorem-1}.

\subsection{Proof of estimates \eqref{higher-estimate-main-theorem} and \eqref{higher-estimate-main-theorem-a-f}}\label{higher-order-combined}
Here we extend the above weighted estimates to higher order.

We first derive higher order transport estimates for the lower order terms. We can revisit the proof of Proposition \ref{transport-estimates}, commuting first \eqref{quantities} $k$ times with $T$ and $r\nabb$. We will then obtain the equivalent higher order estimate:
\bea\label{higher-order-estimate-subsection}
\begin{split}
&E^{k, T, \nabb}_p[\a](\tau)+E^{k, T, \nabb}_p[\psi_1](\tau)+E^{k+1, T, \nabb}_p[\ff](\tau)\\
&+\widehat{\MM}^{k, T, \nabb}_p[\a](0, \tau)+ \widehat{\MM}^{k, T, \nabb}_p[\psi_1](0, \tau)+ \widehat{\MM}^{k+1, T, \nabb}_p[\ff](0, \tau)\\
&\les \MM^{k, T, \nabb}_p[\qf](0,\tau)+\MM^{k+1, T, \nabb}_p[\qf^\F](0,\tau)+E^{k, T, \nabb}_p[\a](0)+ E^{k, T, \nabb}_p[\psi_1](0)+ E^{k+1, T, \nabb}_p[\ff](0)
\end{split}
\eea

We now derive the higher order estimates for $\qf$ and $\qf^\F$. We apply Corollary \ref{higher-order-single} to 
\begin{enumerate}
\item $\Psi=\qf$, $n=k$, $\M=\M_1[\qf, \qf^\F]$, 
\item $\Psi=\qf^\F$, $n=k+1$, $\M=\M_2[\qf, \qf^\F]$,
\end{enumerate}
We sum the resulting estimates, allowing multiplication by positive constants. We therefore obtain 
\small
\beaa
&&  E_p^{k, T, \nabb}[\qf](\tau) + E_p^{k+1, T, \nabb}[\qf^\F](\tau)       +  \MM_{p}^{k, T, \nabb}[\qf](0, \tau) +  \MM_{p}^{k+1, T, \nabb}[\qf^\F](0, \tau)   \\
&&\les E_p^{k, T, \nabb}[\qf](0)+E_p^{k+1, T, \nabb}[\qf^\F](0)\\
&&+\II^{k, T, \nabb}_{p\,; \,R}[ \M_1[\qf, \qf^\F]](0, \tau)+\II^{k+1, T, \nabb}_{p\,; \,R}[ \M_2[\qf, \qf^\F]](0, \tau)\\
                 &&- \sum_{i+j \leq k}\int_{\MM(0,\tau)} \left((r-r_P)R((T^i)(r\nabb)^j\qf)+ \Lambda T((T^i)(r\nabb)^j\qf )+\frac 1 2  w (T^i)(r\nabb)^j\qf\right)\c (T^i)(r\nabb)^j\M_1[\qf, \qf^\F] \\
                      &&- \sum_{i+j \leq k+1}\int_{\MM(0,\tau)} \left((r-r_P)R((T^i)(r\nabb)^j\qf^\F)+ \Lambda T((T^i)(r\nabb)^j\qf^\F )+\frac 1 2  w (T^i)(r\nabb)^j\qf^\F\right)\c (T^i)(r\nabb)^j\M_2[\qf, \qf^\F] 
\eeaa
\normalsize
The absorption follows a very similar pattern than the estimates for the coupling terms and the lower order terms analyzed in detail in Sections \ref{coupling-subsection} and \ref{section-lower-order}, making use of the higher order terms transport estimates \eqref{higher-order-estimate-subsection}. This proves \eqref{higher-estimate-main-theorem}.

Combining \eqref{higher-order-estimate-subsection} and \eqref{higher-estimate-main-theorem} we obtain \eqref{higher-estimate-main-theorem-a-f}.

 \subsection{Proof of estimate \eqref{polynomial-decay-theorem}}\label{polynomial-combined}
 We combine the $r^p$-weighted hierarchy estimates derived above to obtain polynomial decay of the energy.

 We sum \eqref{higher-estimate-main-theorem} and \eqref{higher-estimate-main-theorem-a-f} and obtain 
 \small
 \bea\label{higher-derivative-degenerate}
\begin{split}
&E_p^{n, T, \nabb}[\a](\tau)+E_p^{n, T, \nabb}[\psi_1](\tau)+E_p^{n+1, T}[\ff](\tau)+ E_p^{n, T, \nabb}[\qf](\tau) + E_p^{n+1, T, \nabb}[\qf^\F](\tau)    \\
&  +  \widehat{\MM}_{p}^{n, T, \nabb}[\a](0, \tau) +  \widehat{\MM}_{p}^{n, T, \nabb}[\psi_1](0, \tau) +  \widehat{\MM}_{p}^{n+1, T}[\ff](0, \tau)  +  \MM_{p}^{n, T, \nabb}[\qf](0, \tau)  +  \MM_{p}^{n+1, T, \nabb}[\qf^\F](0, \tau)   \\
&\les E_p^{n, T, \nabb}[\a](0)+E_p^{n, T, \nabb}[\psi_1](0)+ E_p^{n, T, \nabb}[\qf](0)+E_p^{n+1, T, \nabb}[\ff](0)+E_p^{n+1, T, \nabb}[\qf^\F](0)
\end{split}
\eea
\normalsize

Estimating the degenerate Morawetz bulks by the higher order non-degenerate bulks we obtain
\bea\label{higher-derivative-non-degenerate}
\begin{split}
&E_p^{n, T, \nabb}[\a](\tau)+E_p^{n, T, \nabb}[\psi_1](\tau)+E_p^{n+1, T}[\ff](\tau)+ E_p^{n, T, \nabb}[\qf](\tau) + E_p^{n+1, T, \nabb}[\qf^\F](\tau)    \\
&\les E_p^{n+1, T, \nabb}[\a](0)+E_p^{n+1, T, \nabb}[\psi_1](0)+ E_p^{n+1, T, \nabb}[\qf](0)+E_p^{n+2, T, \nabb}[\ff](0)+E_p^{n+2, T, \nabb}[\qf^\F](0)
\end{split}
\eea
We denote the right hand side of estimate \eqref{higher-derivative-non-degenerate} by $\mathbb{D}_{n+1, p}[\a, \psi_1, \qf, \ff, \qf^\F](0)$. 

Applying \eqref{higher-derivative-non-degenerate} for $n=1$ and $p=2-\de$ implies along a dyadic sequence $\tau_n$, of the form $\tau_n=2 \tau_{n-1}$ the estimate 
\beaa
&&E^{1, T, \nabb}_{1-\de}[\a](\tau_n)+E^{1, T, \nabb}_{1-\de}[\psi_1](\tau_n)+ E^{1, T, \nabb}_{1-\de}[\qf](\tau_n)+E^{2, T, \nabb}_{1-\de}[\ff](\tau_n) + E^{2, T, \nabb}_{1-\de}[\qf^\F](\tau_n) \\
&\les& \frac{1}{\tau_n}\mathbb{D}_{2, 2-\de}[\a, \psi_1, \qf, \ff, \qf^\F](0)
\eeaa 
Using the above and applying \eqref{higher-derivative-degenerate} for $p=1-\de$ and $n=1$ between the time $\tau_n$ and any $\tau \in (\tau_n, \tau_{n+1}]$ yields the previous estimate for any $\tau$: 
\beaa
&&E^{1, T, \nabb}_{1-\de}[\a](\tau)+E^{1, T, \nabb}_{1-\de}[\psi_1](\tau)+ E^{1, T, \nabb}_{1-\de}[\qf](\tau)+E^{2, T, \nabb}_{1-\de}[\ff](\tau) + E^{2, T, \nabb}_{1-\de}[\qf^\F](\tau) \\
&\les& \frac{1}{\tau}\mathbb{D}_{2, 2-\de}[\a, \psi_1, \qf, \ff, \qf^\F](0)
\eeaa 
Turning back to \eqref{higher-derivative-non-degenerate} with $n=0$ and $p=1-\de$, along a dyadic sequence and using the previous estimate, we can produce the estimate
\beaa
&&E_{\de}[\a](\tau_n)+E_{\de}[\psi_1](\tau_n)+E^{1, T, \nabb}_{\de}[\ff](\tau_n) + E_{\de}[\qf](\tau_n)+ E_{\de}^{1, T, \nabb}[\qf^\F](\tau_n)\\
&\les& \frac{1}{\tau_n^{2-\de}}\mathbb{D}_{2, 2-\de}[\a, \psi_1, \qf, \ff, \qf^\F](0)
\eeaa
Using the above and applying \eqref{higher-derivative-degenerate} with $p=\de$ and $n=0$ between the time $\tau_n$ and any $\tau \in (\tau_n, \tau_{n+1}]$, yields the previous estimate for any $\tau$, therefore proving \eqref{polynomial-decay-theorem}.

\subsection{The case spin $-2$}\label{spin-2-all}
We will outline here the proof of the case spin $-2$, remarking the differences with the proof of Main Theorem. 
As remarked in Section \ref{separate-estimates} and Section \ref{coupling-subsection}, the separated estimates and the estimates for the coupling terms are identical in the case of spin $-2$. 
The first main difference is in the derivation of Proposition \ref{transport-estimates}, for transport estimates. Indeed, in the case of spin $-2$, the transport estimates have to be integrated along the $e_4$ direction, with a consequent loss of decay in $r$ in the equivalent of Lemma \ref{general-lemma-transport}. We show the version of the Lemma in spin $-2$.

\begin{lemma}\label{general-lemma-transport-2} Let $\Phi_1$ and $\Phi_2$ be two symmetric traceless $2$-tensor which verify the differential relation
\bea\label{generic-differential-relation-2}
\nabb_4(r^n \ka^m \Phi_1)&=& \ka \Phi_2
\eea
for $n$, $m$ integers. Then, in the far away region, for all $l >0$
\beaa
\int_{\MM_{far}(0, \tau)}r^{2n-2m-l-3}|\Phi_1|^2 \les \int_{\MM_{far}(0, \tau)} r^{-l-3} |\Phi_2|^2
\eeaa
\end{lemma}
\begin{proof} From \eqref{generic-differential-relation-2}, we obtain
\beaa
\nabb_4(r^{2n}\ka^{2m} |\Phi_1|^2)&=& 2r^n \ka^{m+1} \Phi_1  \Phi_2
\eeaa
Multiplying by $r^{-l-2}$ for $l>0$ and recalling that $\nabb_4r=\frac 1 2 r \ka$ and $\ka\geq 0$, we obtain
\beaa
\nabb_4(r^{2n}\ka^{2m} |\Phi_1|^2 \c r^{-l-2})+\left(1+\frac{l}{2}\right)r^{2n-l-2}\ka^{2m+1} |\Phi_1|^2&=& 2r^{n-l-2} \ka^{m+1} \Phi_1  \Phi_2\\
&\leq& \frac{l}{4} r^{2n-l-2}\ka^{2m+1}|\Phi_1|^2+\frac{4}{l} r^{-l-2}\ka |\Phi_2|^2
\eeaa
giving 
\beaa
\nabb_4(r^{2n}\ka^{2m} |\Phi_1|^2 \c r^{-l-2})+\left(1+\frac{l}{4}\right)r^{2n-l-2}\ka^{2m+1} |\Phi_1|^2&\leq& \frac{4}{l} r^{-l-2}\ka |\Phi_2|^2
\eeaa
Snce $\div e_4=\frac 1 2 \tr \pi^{(4)}=\ka$, we write
\beaa
&&\div(r^{2n-l-2}\ka^{2m} |\Phi_1|^2  e_4)+\left(1+\frac{l}{4}\right)r^{2n-l-2}\ka^{2m+1} |\Phi_1|^2\\
&=&\nabb_4(r^{2n}\ka^{2m} |\Phi_1|^2 \c r^{-l-2})+\left(1+\frac{l}{4}\right)r^{2n-l-2}\ka^{2m+1} |\Phi_1|^2+r^{2n-l-2}\ka^{2m} |\Phi_1|^2 \div e_4\\
&\leq& \frac{4}{l} r^{-l-2}\ka |\Phi_2|^2+r^{2n-l-2}\ka^{2m+1} |\Phi_1|^2 
\eeaa
which therefore gives
\beaa
&&\div(r^{2n-l-2}\ka^{2m} |\Phi_1|^2  e_4)+\frac{l}{4} r^{2n-l-2}\ka^{2m+1} |\Phi_1|^2 \leq \frac{4}{l} r^{-l-2}\ka |\Phi_2|^2
\eeaa
Integrating the above inequality on $\MM_{far}(0, \tau)$ and using divergence theorem on the left hand side (recalling that $e_4$ is normal to $\Sigma_{far}$, we obtain:
\beaa
\int_{\MM_{far}(0, \tau)}r^{2n-2m-l-3}|\Phi_1|^2 \les \int_{\MM_{far}(0, \tau)} r^{-l-3} |\Phi_2|^2
\eeaa
where we allow for a constant on the right hand side depending on $R_0$.
\end{proof}

Using the above Lemma for the spin $-2$ lower order quantities $\aa$, $\underline{\psi}_1$ and $\underline{\ff}$, we have a loss of decay in $r$ for the spin $-2$ quantities. Indeed, by \eqref{quantities-2}, $\underline{\ff}$ and $\underline{\qf}^\F$ verify relation \eqref{generic-differential-relation-2} with $n=3$, $m=1$. Since for $l>0$, $r^{-3-l}\leq r^{-3+p}$, we get for $l=\de$:
\beaa
\int_{\MM_{far}(0, \tau)}r^{1-\de}|\underline{\ff}|^2 \les  \MM^{1, T, \nabb}_p[\underline{\qf}^\F](0,\tau)
\eeaa

By \eqref{quantities-2}, $\underline{\psi}_1$ and $\underline{\qf}$ verify relation \eqref{generic-differential-relation-2} with $n=1$, $m=0$. As before, for $l=\de$, we obtain
\beaa
\int_{\MM_{far}(0, \tau)}r^{-1-\de}|\underline{\psi}_1|^2 \les \MM_p[\underline{\qf}](0,\tau)
\eeaa

By \eqref{quantities-2}, $\aa$ and $\underline{\psi}_1$ verify relation \eqref{generic-differential-relation-2} with $n=3$, $m=2$. As before,  for $l=\de $, we obtain
\beaa
\int_{\MM_{far}(0, \tau)}r^{-1-\de}|\aa|^2 \les  \int_{\MM_{far}(0, \tau)}r^{-1-\de}|\underline{\psi}_1|^2
\eeaa
The above bounds motivate the definitions of energies for the negative spin case. To bound the higher derivatives, we follow a similar pattern as in proof of Proposition \ref{transport-estimates}.

The absorption for the lower order terms differ from the case of spin $+2$ only in the far-away region. Observe that 
\beaa
\II_{p\,; \,R}[ \Lbb_1[\underline{\qf}^\F]](0,\tau)&=&  \int_{\MM_{far}(0,\tau)} r^{-3+p}|\underline{\ff}|^2
\eeaa
which still decays faster than the non-degenerate bulk norm for $\underline{\ff}$.
Similarly, 
\beaa
\II^{1, T, \nabb}_{p\,; \,R}[\Lbb_2[\underline{\qf}^\F]]( 0,\tau)&=& \int_{\MM_{far}(0,\tau)}r^{-3+p}|\underline{\ff}|^2+ r^{-3+p}|T\underline{\ff}|^2+r^{-1+p}|\nabb\underline{\ff}|^2 \les  \widehat{\MM}_\de[\underline{\ff}](0, \tau)
\eeaa
Finally 
\beaa
\II_{p\,; \,R}[ \Lbb_1[\underline{\qf}]](0,\tau)&=& \int_{\MM_{far}(\tau_1,\tau_2)} r^{-5+p}|\underline{\psi}_1|^2+r^{-3+p} |\aa|^2\\
&\les& \widehat{\MM}_\de[\underline{\psi}_1](0, \tau) +\widehat{\MM}_\de[\aa](0, \tau)
\eeaa
The absorption for the lower order terms in the trapping region and in the bounded redshift region, where powers of $r$ do not matter, is identical to the spin $+2$ case. Putting all this together we prove the Main Theorem for spin $-2$.

\appendix

\section{Derivation of the generalized Teukolsky system}\label{computations-appendix}

The derivation of the equations below were derived in \cite{Giorgi2} in the context of axially symmetric polarized perturbations.
Here, the computations are done in full generality, with no gauge assumption or symmetry. To derive them we will use the equations in Section \ref{linearized-equations}.

\subsection{Preliminaries}
We compute the wave operator $\Box_\g$, where $\g$ is a linear perturbation of the Reissner-Nordstr{\"o}m metric, applied to a $n$-tensor $\Psi$. In particular, for $\g$ the Ricci coefficients and curvature components $\chih, \chibh, \eta,  \etab,  \ze,  \xi, \xib$ and $\a,  \b,  \sigma, \bb,  \aa, \bF,   \sigmaF, \bbF$ are $O(\ep)$. We will keep this in mind while performing the computations at the linear level.
\begin{lemma}\label{square-RN}\label{wave-T-R} The wave operator of a $p$-rank $S$-tensor $\Psi$ is given by 
\bea
\Box_\g\Psi&=& - \nabb_{3}\nabb_{4} \Psi +\lapp_p \Psi+\left(-\frac 12 \kab+2\omb\right) \nabb_4\Psi-\frac 1 2 \ka\nabb_3\Psi+2\eta^C\nabb_C \Psi, \label{formula-1-wave}\\
&=& - \nabb_{4}\nabb_{3} \Psi +\lapp_p \Psi+\left(-\frac 12 \ka+2\om\right) \nabb_3\Psi-\frac 1 2 \kab\nabb_4\Psi+2\etab^C\nabb_C \Psi \label{formula-2-wave} \\
&=& -\frac 1 \Up TT\Psi+\frac 1 \Up RR\Psi +\lapp_p \Psi +\frac 2 r  R \Psi
\eea
where $\lapp_n\Psi=\g^{CD} \nabb_{C}\nabb_{D}\Psi$ is the Laplacian for $p$-tensors.
\end{lemma}

In the derivation of the equations, we will rescale the quantities using powers of $r$ and $\kab$. We collect here an useful general lemma. 
\begin{lemma}\label{wave-rescaled} Consider a rescaled tensor $r^n \kab^m \Psi$, for $n$ and $m$ two numbers. Then 
\bea
r^n \kab^m \nabb_3(\Psi)&=&\nabb_3(r^n \kab^m \Psi)+\left(\frac {m-n}{ 2}  \kab +2m \omb\right)r^n\kab^m \Psi, \label{rescaled-derivatives-r-kab}\\
r^n \kab^m \nabb_4(\Psi)&=&\nabb_4(r^n \kab^m \Psi)+\left(\frac {m-n}{ 2}  \ka -2m \om-2m \rho \kab^{-1}\right)r^n\kab^m  \Psi
\eea
Moreover, it verifies the following wave equation:
\beaa
\Box_\g (r^n \kab^m \Psi)&=& \big( -\frac{n(n+1)+m(m-1)-2nm}{4} \ka\kab +m(m-n-1)\om\kab+(m^2+2m-n-nm) \rho+2m \rhoF^2\\
&&-m(m-n-1) \omb\ka+4m(m+1) \om\omb+4m^2\rho\omb \kab^{-1}-2m \nabb_3\om\big)r^n \kab^m \Psi+r^n\kab^m \Box_\g( \Psi)\\
&&+\left(\frac {m-n} 2 \kab+2m\omb\right)r^n\kab^{m}\nabb_4( \Psi)+\left(\frac {m-n}{ 2} \ka-2m\om-2m\rho\kab^{-1}\right) r^n\kab^m\nabb_3(\Psi)
\eeaa
\end{lemma}

We collect here the commutation formulas with the wave operator $\Box_\g$ and the angular operators. 
\begin{lemma}\label{commute-wave-angular} Let $\Psi$ be a symmetric traceless $2$-tensor, $\Phi$ be a $1$-form and $\phi$ a scalar function. Then 
 \beaa
(-r\DDd_1 \lapp_1+\lapp_0 r\DDd_1) \Phi &=& -K r\DDd_1 \Phi, \qquad ( -r\DDs_1 \lapp_0+\lapp_1 r\DDs_1)\phi=K r\DDs_1 \phi \\
( -r\DDd_2 \lapp_2+\lapp_1 r\DDd_2)\Psi&=&-3K r\DDd_2 \Psi, \qquad ( -r\DDs_2 \lapp_1+\lapp_2 r\DDs_2)\Phi=3K r\DDs_2 \Phi
 \eeaa
 and
\beaa
(-r\DDd_1 \Box_1+\Box_0 r\DDd_1) \Phi &=& -K r\DDd_1 \Phi, \qquad ( -r\DDs_1 \Box_0+\Box_1 r\DDs_1)\phi=K r\DDs_1 \phi \\
( -r\DDd_2 \Box_2+\Box_1 r\DDd_2)\Psi&=&-3K r\DDd_2 \Psi, \qquad ( -r\DDs_2 \Box_1+\Box_2 r\DDs_2)\Phi=3K r\DDs_2 \Phi
 \eeaa
\end{lemma}

\subsection{Spin $\pm 1$ Teukolsky equations}
In what follows, we treat the case of spin $+1$ Teukolsky equations. The equivalent equations for spin $-1$ are obtained in a similar way.

\subsubsection{Spin $+1$ Teukolsky equation for $\bF$}
We derive here the spin $+1$ Teukolsky equation for the electromagnetic component $\bF$. Notice that $\bF$ is a gauge-dependent quantity, so is its wave equation. 
\begin{proposition}\label{Teukolsky-bF}[Spin $+1$ Teukolsky equation for $\bF$] Let $\mathcal{S}$ be a linear gravitational and electromagnetic perturbation around Reissner-Nordstr{\"o}m. Consider the (gauge-dependent) curvature component $\bF$, which is part of the solution $\mathcal{S}$. Then $\bF$ satisfies the following equation:
\beaa
\Box_\g \bF&=& -2\omb \nabb_4\bF+\left( \ka+2\om\right) \nabb_3\bF+\left(\frac 1 4 \ka\kab+\om\kab-3\omb \ka+\rhoF^2-2\nabb_4\omb\right) \bF \\
&& +\rhoF\left(2\divv \chih +4\b-2\nabb_3\xi+(\kab +8\omb)\xi \right) 
\eeaa
\end{proposition}
\begin{proof} Applying $\nabb_4$ to \eqref{nabb-3-bF}, and using \eqref{nabb-4-kab}, \eqref{nabb-3-xi} and \eqref{nabb-4-rhoF},  we obtain
\beaa
\nabb_4\nabb_3 \bF&=&-\left(\frac 1 2 \nabb_4\kab-2\nabb_4\omb\right) \bF-\left(\frac 1 2 \kab-2\omb\right) \nabb_4\bF -\nabb_4\DDs_1(\rhoF, \sigmaF) +2\nabb_4\eta\rhoF+2\eta\nabb_4\rhoF\\
&=&\left(\frac 1 4 \ka\kab-\om\kab-\rho-2\rhoF^2+2\nabb_4\omb\right) \bF+\left(-\frac 1 2 \kab+2\omb\right) \nabb_4\bF -\nabb_4\DDs_1(\rhoF, \sigmaF)\\
&& +\rhoF\left(-2\b-3\ka\eta+\ka\etab+2\nabb_3\xi-8\omb\xi\right)
\eeaa
We simplify the term $\nabb_4\DDs_1(\rhoF, \sigmaF)$ using the commutation formulas \eqref{commutators}. Using \eqref{nabb-4-rhoF}, \eqref{nabb-3-rhoF} and \eqref{nabb-4-sigmaF} and using that $\DDs_1(\divv,\slashed{\curl})=-\lapp_1+K$, we obtain:
\beaa
\nabb_4\DDs_1(\rhoF, \sigmaF)&=& \DDs_1(-\ka \rhoF +\divv \bF, -\ka \sigmaF+\curll \bF)-\frac 1 2 \ka\DDs_1(\rhoF,\sigmaF)\\
&&-(\etab+\ze) (-\ka \rhoF)-\xi (-\kab \rhoF)\\
&=&\nabb(\ka)\rhoF+ \DDs_1(\divv\bF, \slashed{\curl}\bF)-\frac 32 \ka\DDs_1(\rhoF,\sigmaF)+\ka(\etab+\ze) \rhoF+\kab \xi \rhoF\\
&=&\nabb(\ka)\rhoF+ (-\lapp_1\bF+K\bF)-\frac 3 2 \ka\DDs_1(\rhoF,\sigmaF)+\ka(\etab+\ze) \rhoF+\kab \xi \rhoF
\eeaa
Substituting $\DDs_1(\rhoF,\sigmaF)$ using again \eqref{nabb-3-bF} and using \eqref{Codazzi-chi} and \eqref{Gauss}, we obtain
\beaa
\nabb_4\nabb_3 \bF&=&\left(-\frac 1 4 \ka\kab-\om\kab+3\omb \ka-\rhoF^2+2\nabb_4\omb\right) \bF+\left(-\frac 1 2 \kab+2\omb\right) \nabb_4\bF-\frac 3 2 \ka \nabb_3 \bF  + \lapp_1\bF\\
&& +\rhoF\left(-2\divv \chih -4\b+2\nabb_3\xi-\kab \xi-8\omb\xi \right)
\eeaa
Therefore, using Lemma \ref{square-RN}, we obtain
\beaa
\Box_\g\bF&=& - \Bigg(\left(-\frac 1 4 \ka\kab-\om\kab+3\omb \ka-\rhoF^2+2\nabb_4\omb\right) \bF+\left(-\frac 1 2 \kab+2\omb\right) \nabb_4\bF-\frac 3 2 \ka \nabb_3 \bF  + \lapp_1\bF\\
&& +\rhoF\left(-2\divv \chih -4\b+2\nabb_3\xi-\kab \xi-8\omb\xi \right)\Bigg) +\lapp_1 \bF+\left(-\frac 12 \ka+2\om\right) \nabb_3\bF-\frac 1 2 \kab\nabb_4\bF \\
&=& \left(\frac 1 4 \ka\kab+\om\kab-3\omb \ka+\rhoF^2-2\nabb_4\omb\right) \bF-2\omb \nabb_4\bF+\left( \ka+2\om\right) \nabb_3\bF \\
&& +\rhoF\left(2\divv \chih +4\b-2\nabb_3\xi+\kab \xi+8\omb\xi \right)  
\eeaa
as desired. 
\end{proof}

\subsection{Spin $\pm2$ Teukolsky equations}
In this section, we derive the case of spin $+2$ Teukolsky equations. The equivalent equations for spin $-2$ are obtained in a similar way. 

\subsubsection{Spin $+2$ Teukolsky equation for $\a$}
We derive here the spin $+2$ Teukolsky equation for the gauge independent curvature component $\a$. This generalizes the celebrated Teukolsky equation in Schwarzschild. 

\begin{proposition}\label{squarea}[Generalized spin $+2$ Teukolsky equation for $\a$] Let $\mathcal{S}$ be a linear gravitational and electromagnetic perturbation around Reissner-Nordstr{\"o}m. Consider the gauge-invariant curvature component $\a$, which is part of the solution $\mathcal{S}$. Then $\a$ satisfies the following equation:
\beaa
\Box_\g \a&=& -4\omb \nabb_4\a+\left(2 \ka+4\om\right) \nabb_3\a+\left( \frac 1 2 \ka\kab+2\om\,\kab-10\omb\ka-8\om\omb -4\nabb_4\omb   -4\rho+4\rhoF^2\right)\,\a\\
&& +4\rhoF\Big(\nabb_4\ff+ \left( \ka  +2\om \right) \ff\Big)
\eeaa
\end{proposition}
\begin{proof}
Applying $\nabb_4$ to \eqref{nabb-3-a}, we obtain
\beaa
\nabb_4\nabb_3\a&=&-\frac 1 2 \nabb_4\kab\,\a-\frac 1 2 \kab\,\nabb_4\a+4\nabb_4\omb \a+4\omb \nabb_4\a-2 \nabb_4(\DDs_2\, \b) -3\nabb_4\chih \rho-3\chih \nabb_4\rho-2\nabb_4\rhoF \ff-2\rhoF \nabb_4\ff
\eeaa
 and using \eqref{nabb-3-kab}, \eqref{nabb-3-chibh}, \eqref{nabb-4-rho} and \eqref{nabb-4-rhoF}, we obtain
 \beaa
\nabb_4\nabb_3\a&=&-\frac 1 2 \left(-\frac 1 2 \ka\kab+2\om\kab+2\rho \right)\,\a-\frac 1 2 \kab\,\nabb_4\a+4\nabb_4\omb \a+4\omb \nabb_4\a-2 \nabb_4(\DDs_2\, \b) \\
&&-3\left(-(\ka+2\om)\chih-2\DDs_2 \xi-\a \right) \rho-3\chih \left(-\frac 3 2 \ka \rho-\ka\rhoF^2 \right)-2\left(-\ka\rhoF \right) \ff-2\rhoF \nabb_4\ff\\
&=& \left(\frac 1 4 \ka\kab-\om\kab+2\rho +4\nabb_4\omb\right)\,\a-\left(\frac 1 2 \kab-4\omb\right)\,\nabb_4\a-2 \nabb_4(\DDs_2\, \b) +3\rho\left(\left(\frac 5 2 \ka+2\om\right)\chih+2\DDs_2 \xi \right) \\
&&+\rhoF \left(3 \ka\rhoF\chih -2 \nabb_4\ff+2\ka  \ff\right)
\eeaa
Using the commutation formula \eqref{commutator-nabb-4-DDs} and Bianchi identity \eqref{nabb-4-b}, and the identity $\DDs_2\divv\a=(-\frac 1 2 \lapp+K)\a$, we simplify the term $\nabb_4(\DDs_2\, \b)$:
\beaa
&&\nabb_4(\DDs_2\, \b)= \DDs_2(\nabb_4\b)-\frac 1 2 \ka\DDs_2\b\\
&=&  \left(-\frac 5 2 \ka  -2\om \right)\DDs_2\b +\left(-\frac 1 2 \lapp+K\right)\a+3\rho \DDs_2 \xi+\rhoF\left(\nabb_4(\DDs_2\bF)+\left(\frac 1 2 \ka+2\om\right) \DDs_2\bF-2\rhoF\DDs_2\xi\right)
\eeaa
Using Lemma \ref{square-RN}, and using \eqref{nabb-3-a} to substitute $-2\DDs_2\b=\nabb_3 \a+\left(\frac 1 2 \kab-4\omb \right)\a+3\rho \chih+2\rhoF \ff$  and the definition of $\ff$ to substitute $\DDs_2\bF=\ff-\rhoF \chih$ , we obtain:
\beaa
\Box_\g \a&=& \left(-\frac 3 4 \ka\kab+\om\kab-4\rho+2\rhoF^2 -4\nabb_4\omb\right)\,\a-4\omb\,\nabb_4\a+\left(-\frac 12 \ka+2\om\right) \nabb_3\a\\
&&+\left(\frac 5 2 \ka  +2\om \right)\left(\nabb_3 \a+\left(\frac 1 2 \kab-4\omb \right)\a+3\rho \chih+2\rhoF \ff\right) -3\rho\left(\frac 5 2 \ka+2\om\right)\chih \\
&&+\rhoF\left(2\nabb_4\left(\ff-\rhoF \chih\right)+\left( \ka+4\om\right) \left(\ff-\rhoF \chih \right)-4\rhoF\DDs_2\xi-3 \ka\rhoF\chih +2 \nabb_4\ff-2\ka \ff\right)\\
&=& \left(\frac 1 2  \ka\kab+2\om\kab-10\omb\ka-8\om\omb-4\rho+4\rhoF^2 -4\nabb_4\omb\right)\,\a-4\omb\,\nabb_4\a+\left(2 \ka+4\om\right) \nabb_3\a\\
&&+4\rhoF\left(\nabb_4\ff+\left( \ka+2\om\right) \ff\right)
\eeaa
where we used \eqref{Gauss-general}. We get the desired equation.
\end{proof}

\begin{corollary}\label{wave-psi_0} The derived quantity $\psi_0=r^2\kab^2\a$ verifies the following wave equation:
\beaa
\Box_\g \psi_0&=& \left(-\frac 1 2  \ka\kab -4 \rho+4 \rhoF^2\right)\psi_0+\frac 1 r\left(2\ka\kab-4\rho\right) \psi_1+4\rhoF  \Big(\kab\nabb_4(\psi_3)+\left(\frac {1}{ 2}  \ka\kab -2 \rho \right)\psi_3\Big)
\eeaa
\end{corollary}
\begin{proof} Using Lemma \ref{wave-rescaled} and Proposition \ref{squarea}, we obtain the desired equation.
\end{proof}

\begin{corollary}\label{laplacian-a} The Teukolsky equation for $\a$ can be rewritten as 
\beaa
\lapp_2 \psi_0&=&\frac 1 r \kab \nabb_{4}\psi_1+\frac 1 r\left(\frac 3 2 \ka\kab-2\rho\right) \psi_1+ \left(-\frac 1 2  \ka\kab -5 \rho+4 \rhoF^2\right)\psi_0+4\rhoF  \Big(\kab\nabb_4(\psi_3)+\left(\frac {1}{ 2}  \ka\kab -2 \rho \right)\psi_3\Big)
\eeaa
\end{corollary}
\begin{proof} We have
\bea\label{box-in-terms-of-lapp}
\Box_\g \psi_0&=& \left(\frac 1 2 \ka\kab-2\rho \right)\frac 1 r  \psi_1- \frac 1 r \kab \nabb_{4}\psi_1+\rho \psi_0+\lapp_2 \psi_0
\eea
 Using Corollary \ref{wave-psi_0}, we obtain the desired formula. 
\end{proof}

\subsubsection{Spin $+2$ Teukolsky equation for $\ff$}
We derive here the spin $+2$ Teukolsky equation for the gauge independent curvature component $\ff$, which is part of the generalized system.

\begin{proposition}\label{squareff}[Generalized spin $+2$ Teukolsky equation for $\ff$] Let $\mathcal{S}$ be a linear gravitational and electromagnetic perturbation around Reissner-Nordstr{\"o}m. Consider the gauge-invariant curvature component $\ff=\DDs_2 \bF+\rhoF \chih$, which is part of the solution $\mathcal{S}$. Then $\ff$ satisfies the following equation:
\beaa
\Box_\g (r\ff)&=&-2\omb \nabb_4(r\ff)+\left( \ka+2\om\right) \nabb_3(r\ff)+  \left(-\frac 1 2 \ka\kab-3\rho-3\omb\ka+\om\kab-2\nabb_4\omb\right)(r\ff)\\
&&+r\rhoF \left(-\nabb_3\a-( \kab-4\omb)\a\right)
\eeaa
\end{proposition}
\begin{proof} By definition of $\ff$, we compute
\beaa
\nabb_3(\ff)+\left( \kab -2 \omb\right)\ff&=& \nabb_3(\DDs_2 \bF+\rhoF \chih)+\left( \kab -2 \omb\right)(\DDs_2 \bF+\rhoF \chih)\\
&=& \DDs_2\nabb_3 \bF-\frac 1 2 \kab \DDs_2 \bF+\nabb_3\rhoF \chih+\rhoF \nabb_3\chih+\left( \kab -2 \omb\right)(\DDs_2 \bF+\rhoF \chih)
\eeaa
and using \eqref{nabb-3-bF} and \eqref{nabb-3-chih}, we obtain
\bea\label{nabb-3-ff}
\nabb_3(\ff)+\left( \kab -2 \omb\right)\ff&=&  -\DDs_2\DDs_1(\rhoF, \sigmaF) -\frac 12\rhoF \left(  \kab \chih + \ka \chibh \right)
\eea
Applying $\nabb_4$ to \eqref{nabb-3-ff}, and using \eqref{nabb-4-kab}, \eqref{nabb-4-rhoF}, \eqref{nabb-4-sigmaF}, \eqref{nabb-4-chih}, \eqref{nabb-4-chibh}, \eqref{nabb-4-ka}, we obtain
\beaa
\nabb_4\nabb_3(\ff)&=& \left( \frac 1 2 \ka\kab-2\om\kab-2\rho +2 \nabb_4\omb\right)\ff -\left( \kab -2 \omb\right)\nabb_4\ff -\nabb_4\DDs_2\DDs_1(\rhoF, \sigmaF) \\
&&-\frac 12\rhoF \left( (-3 \ka\kab+2\rho) \chih- 2\ka^2 \chibh+ \kab (-2\DDs_2 \xi-\a) -2\ka \DDs_2\etab \right)
\eeaa
We simplify $\nabb_4\DDs_2\DDs_1(\rhoF, \sigmaF)$ recalling the formulas in Proposition \ref{Teukolsky-bF}, and using Lemma \ref{commute-wave-angular} and $\DDs_2\divv\chih=(-\frac 1 2 \lapp+K)\chih$. We obtain 
\beaa
\nabb_4\DDs_2\DDs_1(\rhoF, \sigmaF)&=&(2K\chih+2\DDs_2\b-2\rhoF \DDs_2\bF+\ka\DDs_2\etab+\kab\DDs_2\xi)\rhoF\\
&&-\lapp_2(\DDs_2\bF+\rhoF \chih)+4K\DDs_2\bF-2 \ka\DDs_2\DDs_1(\rhoF,\sigmaF)
\eeaa
This gives, using \eqref{nabb-3-ff}
\beaa
\nabb_4\nabb_3(\ff)-\lapp_2(\ff)&=& \left( -\frac 3 2 \ka\kab+4\omb\ka-2\om\kab-2\rho +2 \nabb_4\omb\right)\ff -\left( \kab -2 \omb\right)\nabb_4\ff-2\ka\nabb_3(\ff)\\
&& +( \ka\kab+\rho-2\rhoF^2)\rhoF\chih+(\ka\kab+4\rho-2\rhoF^2)\DDs_2\bF\\
&& +\rhoF(-2 \DDs_2\b+\frac 1 2 \kab \a)
\eeaa
and using \eqref{nabb-3-a} to substitute $-2\DDs_2\b=\nabb_3 \a+\left(\frac 1 2 \kab-4\omb \right)\a+3\rho \chih+2\rhoF \ff$, we have
\beaa
\Box_\g \ff&=& \left( \frac 1 2 \ka\kab-4\omb\ka+2\om\kab-2\rho -2 \nabb_4\omb\right)\ff +\left(\frac 1 2 \kab -2 \omb\right)\nabb_4\ff+\left(\frac 32 \ka+2\om\right) \nabb_3\ff\\
&& -\rhoF(\nabb_3 \a+\left( \kab-4\omb \right)\a)
\eeaa
Using Lemma \ref{wave-rescaled}, we finally have 
\beaa
\Box_\g (r \ff)&=&\left(-\frac 1 2 \ka\kab-3\omb\ka+\om\kab-3\rho -2 \nabb_4\omb\right)\ff  -2 \omb \nabb_4(r\ff)+\left( \ka+2\om\right) \nabb_3(r\ff)\\
&& -r\rhoF(\nabb_3 \a+\left( \kab-4\omb \right)\a)
\eeaa
as desired.
\end{proof}

\begin{corollary}\label{wave-psi_3} The derived quantity $\psi_3=r^2\kab \ff$ verifies the following wave equation:
\beaa
\Box_\g \psi_3&=& \left(- \ka\kab-3 \rho\right)\psi_3+\frac 1 r\left( \ka\kab-2\rho\right) \qf^\F+\rhoF \left(-\frac 1 r \psi_1-\frac 1 2  \psi_0\right) 
\eeaa
\end{corollary}
\begin{proof} Using Lemma \ref{wave-rescaled} for $\psi_3=r\kab (r \ff)$ with $n=m=1$, we obtain the desired equation.
\end{proof}

\begin{corollary}\label{laplacian-psi-3}  The Teukolsky equation for $\ff$ can be rewritten as 
 \beaa
\lapp_2 \psi_3&=&\frac 1 r \kab \nabb_{4}\qf^\F+\frac 1 r\left( \frac 1 2 \ka\kab\right) \qf^\F+\left(- \ka\kab-4 \rho\right)\psi_3+\rhoF \left(-\frac 1 r \psi_1-\frac 1 2  \psi_0\right)
 \eeaa
\end{corollary}

\section{Derivation of the generalized Regge-Wheeler system}\label{derivation-Regge}

Since the derived quantities are defined in terms of the operator $\underline{P}$, we derive a general lemma to compute the wave equation for a derived quantity.

\begin{lemma}\label{squareP} Let $\Psi$ be a $p$-tensor. Then $\underline{P}\Psi$ is a $p$-tensor which verifies the following wave equation:
\beaa
\Box_\g(\underline{P} \Psi)&=& \frac{1}{r}\left( -\ka\kab+2\rho\right) \underline{P}(\underline{P}(\Psi))+\left(\frac 1 2 \ka\kab-4\rho-2\rhoF^2\right)\underline{P}(\Psi)+\left(\frac 1 2 \rho+\rhoF^2 \right)r\Psi\\
&&+ \frac 3 2r \Box_\g(\Psi)+\kab^{-1}r \nabb_3(\Box_\g(\Psi))
\eeaa
\end{lemma}
\begin{proof} Writing $\underline{P}\Psi= r \kab^{-1} \nabb_3\Psi +\frac 12 r \Psi$, we have 
\beaa
\Box_\g(\underline{P} \Psi)&=&\Box_\g(r \kab^{-1} \nabb_3\Psi )+\frac 12 \Box_\g(r \Psi)
\eeaa
We apply Lemma \ref{wave-rescaled} for both terms. For $n=1$, $m=0$ we have
\beaa
\Box_\g (r \Psi)&=& \big( -\frac{1}{2} \ka\kab -\rho\big)r \Psi+r \Box_\g( \Psi)+\left(-\frac 12 \kab\right)r\nabb_4( \Psi)+\left(-\frac {1}{ 2} \ka\right) r\nabb_3(\Psi)
\eeaa
and for $n=1$, $m=-1$ we have
\beaa
\Box_\g (r \kab^{-1} \nabb_3\Psi)&=& \left( -\frac{3}{2} \ka\kab +3\om\kab- \rho-2 \rhoF^2-3\omb\ka+4\rho\omb \kab^{-1}+2 \nabb_3\om\right)r \kab^{-1} \nabb_3\Psi+r\kab^{-1} \Box_\g( \nabb_3\Psi)\\
&&+\left(- \kab-2\omb\right)r\kab^{-1}\nabb_4( \nabb_3\Psi)+\left(-\ka+2\om+2\rho\kab^{-1}\right) r\kab^{-1}\nabb_3(\nabb_3\Psi)
\eeaa
Using that 
\beaa
\Box_\g( \nabb_3(\Psi))&=&\nabb_3(\square_\g(\Psi))+[\Box_\g, \nabb_3]\Psi=\\
&=&\nabb_3(\Box_\g(\Psi))+\kab\Box_\g \Psi - 2\om \nabb_3(\nabb_3( \Psi)) +(\kab+2\omb)\nabb_4(\nabb_3( \Psi))+\\
&&+\left(\frac{1}{4}\ka\kab  - 3\om\kab  +\omb\ka -8\om\omb -\rho-2\rhoF^2  +2\nabb_4(\omb)\right)\nabb_3(\Psi)  +\frac{1}{4}\kab^2 \nabb_4(\Psi)
\eeaa
 we obtain
 \beaa
\Box_\g (r \kab^{-1} \nabb_3\Psi)&=& \left( -\frac{5}{4} \ka\kab -2 \rhoF^2-2\omb\ka+4\rho\omb \kab^{-1}\right)r \kab^{-1} \nabb_3\Psi+r\kab^{-1}\nabb_3(\Box_\g(\Psi))+r\Box_\g \Psi  +\frac{1}{4}\kab r\nabb_4(\Psi)\\
&&+\left(-\ka+2\rho\kab^{-1}\right) r\kab^{-1}\nabb_3(\nabb_3\Psi)
\eeaa
 Putting these two together we obtain 
 \beaa
 \Box_\g(\underline{P} \Psi)&=&\left(-\ka+2\rho\kab^{-1}\right) r\kab^{-1}\nabb_3\nabb_3\Psi+\left( -\frac{3}{2} \ka\kab -2\omb\ka+4\rho\omb \kab^{-1}-2 \rhoF^2\right)r \kab^{-1} \nabb_3\Psi+ \left( -\frac{1}{4} \ka\kab -\frac 1 2 \rho\right)r \Psi\\
 &&+r\kab^{-1}\nabb_3(\Box_\g(\Psi))+\frac 3 2 r\Box_\g \Psi
 \eeaa
Writing 
\beaa
 \nabb_3(\Psi)&=&\frac{1}{r}\kab (\underline{P}\Psi)-\frac 1 2 \kab \Psi, \\
\nabb_3\nabb_3\Psi&=&\frac{1}{r^2}\kab^2 \underline{P}(\underline{P}(\Psi))-\frac 2 r(\kab^2+\omb \kab) (\underline{P}\Psi)+(\frac 1 2 \kab^2+\omb\kab) \Psi
 \eeaa
we obtain the desired formula.
\end{proof}

 We derive the following useful lemma to compute the wave equations of the derived quantities. 

\begin{lemma}\label{generalwavePhi} 
Let $\Psi$ be a $p$-tensor which verifies a wave equation of the form: 
\bea\label{wave-eq-ABC}
\Box_\g \Psi &=& A\underline{P}^{-1}(\Psi)+B\Psi+C \underline{P}(\Psi) +M,
\eea
where $\underline{P}^{-1}(\Psi)$ is a $p$-tensor such that $P\underline{P}^{-1}(\Psi)=\Psi$,  $A$, $B$, $C$ are coefficients of the wave equation, and $M$ is the right hand side. Then the $p$-tensor $\underline{P}(\Psi)$ verifies the wave equation:
\beaa
\Box_\g(\underline{P} \Psi)&=& \left( \kab^{-1}r\nabb_3(A)+ r A\right)\underline{P}^{-1}(\Psi)+\left(\kab^{-1}r\nabb_3(B)+ r B+ A+\frac 1 2 r\rho+r\rhoF^2  \right)\ \Psi \\
&&+\left( B +\kab^{-1}r\nabb_3(C)+ r C+\frac 1 2 \ka\kab-4\rho-2\rhoF^2\right)\ \underline{P}(\Psi)+ \left( C+\frac{1}{r}\left( -\ka\kab+2\rho\right) \right)\ \underline{P}(\underline{P}(\Psi))\\
&&+\kab^{-1}r\nabb_3 M  +\frac 3 2r M
\eeaa
\end{lemma}

\begin{proof} 
Substituting the expression \eqref{wave-eq-ABC} for $\Box_\g \Psi$ in $\nabb_3(\Box_\g(\Psi))$, we compute:
\beaa
\nabb_3(\Box_\g(\Psi))&=& \nabb_3(A)\underline{P}^{-1}(\Psi)+A\nabb_3\underline{P}^{-1}(\Psi)+\nabb_3(B)\Psi+B\nabb_3\Psi+\nabb_3(C) \underline{P}(\Psi) +C \nabb_3\underline{P}(\Psi)+\nabb_3 M
\eeaa
Writing  $\nabb_3 f=\frac 1 r \kab \ \underline{P}f -\frac 12 \kab f$, we have
\beaa
\nabb_3(\Box_\g(\Psi))&=&\left( \nabb_3(A)-\frac 12 \kab A\right)\underline{P}^{-1}(\Psi)+\left(\nabb_3(B)-\frac 1 2 \kab B+\frac 1 r \kab A \right)\ \Psi +\left(\frac 1 r \kab B +\nabb_3(C)-\frac 1 2 \kab C\right)\ \underline{P}\Psi  \\
&&+ \frac 1 r \kab  C\ \underline{P}\underline{P}\Psi+\nabb_3 M
\eeaa
Substituting the above in Lemma \ref{squareP}, we obtain the desired equation.
\end{proof}

We will use the spin $+2$ Teukolsky equations for $\psi_0$ and $\psi_3$ and Lemma \ref{generalwavePhi} to compute the wave equations of $\qf$ and $\qf^\F$.

\subsection{Wave equation for $\qf^\F$}
We derive here the wave equation for $\qf^\F$.

\begin{proposition}\label{wave-qfF} Let $\mathcal{S}$ be a linear gravitational and electromagnetic perturbation around Reissner-Nordstr{\"o}m. Consider the gauge-invariant derived quantity $\qf^\F$, which is part of the solution $\mathcal{S}$. Then $\qf^\F$ satisfies the following equation:
\beaa
\Box_\g \qf^\F+\left( \ka\kab+3 \rho \right)\ \qf^\F &=&\rhoF \left(-\frac 1 r  \qf+4r\rhoF \ \psi_3 \right)
\eeaa
\end{proposition}
\begin{proof} Recall that $\qf^\F=\psi_4=\underline{P}(\psi_3)$. By Corollary \ref{wave-psi_3},  we have that $\psi_3$ verifies a wave equation of the form \eqref{wave-eq-ABC} , so we can make use of Lemma \ref{generalwavePhi}, with $A=0$, $B=- \ka\kab-3 \rho$, $C=\frac 1 r\left( \ka\kab-2\rho\right)$ and $M=\rhoF \left(-\frac 1 r \psi_1-\frac 1 2  \psi_0\right)$, and obtain 
\beaa
\Box_\g(\qf^\F)&=&\left(\kab^{-1}r(\ka\kab^2+\frac 5 2\kab \rho+3\kab \rhoF^2)+ r (- \ka\kab-3 \rho)+\frac 1 2 r\rho+r\rhoF^2  \right)\ \psi_3 \\
&&+\left( - \ka\kab-3 \rho -\frac 1 2 (\ka\kab-2\rho)+\kab^{-1}(-\ka\kab^2+5\kab\rho+2\kab\rhoF^2)+ r \frac 1 r\left( \ka\kab-2\rho\right)+\frac 1 2 \ka\kab-4\rho-2\rhoF^2\right)\ \qf^\F\\
&&+\kab^{-1}r\nabb_3 M  +\frac 3 2r M\\
&=&\left( - \ka\kab-3 \rho \right)\ \qf^\F+4r\rhoF^2 \ \psi_3 +\kab^{-1}r\nabb_3 M  +\frac 3 2r M
\eeaa
We compute $\kab^{-1}r\nabb_3 M  +\frac 3 2r M$:
\beaa
\kab^{-1}r\nabb_3 M  +\frac 3 2r M&=& \kab^{-1}r\nabb_3(\rhoF \left(-\frac 1 r \psi_1-\frac 1 2  \psi_0\right))  +\frac 3 2r \rhoF \left(-\frac 1 r \psi_1-\frac 1 2  \psi_0\right)\\
&=& -r\rhoF \left(-\frac 1 r \psi_1-\frac 1 2  \psi_0\right)+\kab^{-1}r\rhoF \nabb_3\left(-\frac 1 r \psi_1-\frac 1 2  \psi_0\right)  +\frac 3 2r \rhoF \left(-\frac 1 r \psi_1-\frac 1 2  \psi_0\right)\\
&=& \kab^{-1}r\rhoF \left(r^{-1}\frac 1 2 \kab \psi_1-\frac 1 r \nabb_3\psi_1-\frac 1 2  \nabb_3\psi_0\right)  +\frac 1 2r \rhoF \left(-\frac 1 r \psi_1-\frac 1 2  \psi_0\right)
\eeaa
and writing $\nabb_3(\psi_0)=\frac 1 r \kab \psi_1-\frac 1 2  \kab \psi_0$ and  $\nabb_3(\psi_1)=\frac 1 r \kab \qf-\frac 1 2  \kab \psi_1$, we obtain
\beaa
\kab^{-1}r\nabb_3 M  +\frac 3 2r M&=& \kab^{-1}r\rhoF \left(-\frac 1 r (\frac 1 r \kab \qf-\frac 1 2  \kab \psi_1)-\frac 1 2  (\frac 1 r \kab \psi_1-\frac 1 2  \kab \psi_0)\right)  +\frac 1 2r \rhoF \left(-\frac 1 2  \psi_0\right)\\
&=& -\rhoF \left(\frac 1 r  \qf\right)
\eeaa
giving the desired equation.
\end{proof}

\subsection{Wave equation for $\qf$}
We first derive the wave equation for $\psi_1$. 
\begin{proposition}\label{wave-psi_1} The derived quantity $\psi_1$ verifies the following wave equation 
\beaa
\Box_\g \psi_1&=&\left(\frac 3 2 r\rho+r\rhoF^2  \right)\ \psi_0 +\left( - \ka\kab+6\rhoF^2\right)\ \psi_1+  \frac 1 r\left( \ka\kab-2\rho\right)\ \psi_2\\
&&+2\rhoF \Big(2 \kab \nabb_4(\qf^\F)- \ka\kab \qf^\F-  r\kab\nabb_4\psi_3+\left(-\frac {1}{ 2}  \ka\kab +6 \rho +4\rhoF^2\right)r\psi_3\Big)
\eeaa
\end{proposition}
\begin{proof} Recall that $\psi_1=\underline{P}(\psi_0)$. By Corollary \ref{wave-psi_0},  we have that $\psi_0$ verifies a wave equation of the form \eqref{wave-eq-ABC} , so we can make use of Lemma \ref{generalwavePhi}, with $A=0$, $B=-\frac 1 2  \ka\kab -4 \rho+4 \rhoF^2$, $C=\frac 1 r\left( 2\ka\kab-4\rho\right)$ and $M=4\rhoF  \Big(\kab\nabb_4(\psi_3)+\left(\frac {1}{ 2}  \ka\kab -2 \rho \right)\psi_3\Big)$, and obtain 
\beaa
\Box_\g(\psi_1)&=& \left(\kab^{-1}r(\frac 1 2 \ka\kab^2+5\kab\rho-4\kab\rhoF^2)+ r (-\frac 1 2  \ka\kab -4 \rho+4 \rhoF^2)+\frac 1 2 r\rho+r\rhoF^2  \right)\ \psi_0 \\
&&+\left( -\frac 1 2  \ka\kab -4 \rho+4 \rhoF^2 + (\ka\kab-2\rho)+2\kab^{-1}(-\ka\kab^2+5\kab\rho+2\kab\rhoF^2)+\frac 1 2 \ka\kab-4\rho-2\rhoF^2\right)\ \psi_1\\
&&+  \frac 1 r\left( \ka\kab-2\rho\right)\ \psi_2+\kab^{-1}r\nabb_3 M  +\frac 3 2r M\\
&=& \left(\frac 3 2 r\rho+r\rhoF^2  \right)\ \psi_0 +\left( - \ka\kab+6\rhoF^2\right)\ \psi_1+  \frac 1 r\left( \ka\kab-2\rho\right)\ \psi_2+\kab^{-1}r\nabb_3 M  +\frac 3 2r M
\eeaa
We compute $\kab^{-1}r\nabb_3 M  +\frac 3 2r M$:
\beaa
&&\kab^{-1}r\nabb_3 M  +\frac 3 2r M=4\kab^{-1}r\rhoF \Big(\nabb_3\kab\nabb_4(\psi_3)+\kab\nabb_3\nabb_4(\psi_3)+\nabb_3\left(\frac {1}{ 2}  \ka\kab -2 \rho \right)\psi_3+\left(\frac {1}{ 2}  \ka\kab -2 \rho \right)\nabb_3\psi_3\Big) \\
&& +2r \rhoF  \Big(\kab\nabb_4(\psi_3)+\left(\frac {1}{ 2}  \ka\kab -2 \rho \right)\psi_3\Big)\\
&=&4\kab^{-1}r\rhoF \Big(\kab\nabb_3\nabb_4(\psi_3)+\left(\frac {1}{ 2}  \ka\kab -2 \rho \right)\nabb_3\psi_3\Big) +2r \rhoF  \Big(-4\omb \nabb_4(\psi_3)+\left(-\frac {1}{ 2}  \ka\kab +6 \rho +4\rhoF^2\right)\psi_3\Big)
\eeaa
Writing $\nabb_3(\psi_3)=\frac 1 r \kab \qf^\F-\frac 1 2  \kab \psi_3$, and 
\beaa
\nabb_3\nabb_4(\psi_3)&=&\nabb_4\nabb_3(\psi_3)-2\om\nabb_3\psi_3+2\omb \nabb_4\psi_3=\nabb_4(\frac 1 r \kab \qf^\F-\frac 1 2  \kab \psi_3)-2\om(\frac 1 r \kab \qf^\F-\frac 1 2  \kab \psi_3)+2\omb \nabb_4\psi_3\\
&=&\nabb_4(\frac 1 r \kab) \qf^\F+\frac 1 r \kab \nabb_4(\qf^\F)-\frac 1 2  \nabb_4\kab \psi_3-\frac 1 2  \kab \nabb_4\psi_3-2\om(\frac 1 r \kab \qf^\F-\frac 1 2  \kab \psi_3)+2\omb \nabb_4\psi_3\\
&=&\frac 1 r(-\ka\kab+2\rho) \qf^\F+\frac 1 r \kab \nabb_4(\qf^\F)+ (\frac 1 4 \ka\kab-\rho) \psi_3+(-\frac 1 2  \kab+2\omb )\nabb_4\psi_3
\eeaa
which therefore gives
\beaa
\kab^{-1}r\nabb_3 M  +\frac 3 2r M&=&2\rhoF \Big(2 \kab \nabb_4(\qf^\F)- \ka\kab \qf^\F-  r\kab\nabb_4\psi_3+\left(-\frac {1}{ 2}  \ka\kab +6 \rho +4\rhoF^2\right)r\psi_3\Big)
\eeaa
We get the desired equation. 
\end{proof}

We derive here the wave equation for $\qf$.

\begin{proposition}\label{wave-qf} Let $\mathcal{S}$ be a linear gravitational and electromagnetic perturbation around Reissner-Nordstr{\"o}m. Consider the gauge-invariant derived quantity $\qf^\F$, which is part of the solution $\mathcal{S}$. Then $\qf^\F$ satisfies the following equation:
\beaa
\Box_\g\qf+\left( \ka\kab-10\rhoF^2\right)\ \qf&=&r \rhoF \Big(4 \lapp_2 \qf^\F-4 \kab \nabb_4\qf^\F-4 \ka\nabb_3\qf^\F+\left( 6\ka\kab+16 \rho +8\rhoF^2\right)\ \qf^\F\Big) \\
&&+r\rhoF\Big(-2r\rhoF\psi_0-4\rhoF  \psi_1+(-12\rho-40\rhoF^2) \ r\psi_3\Big)
\eeaa
\end{proposition}
\begin{proof} Recall that $\qf=\psi_2=\underline{P}(\psi_1)$. By Proposition \ref{wave-psi_1},  we have that $\psi_1$ verifies a wave equation of the form \eqref{wave-eq-ABC} , so we can make use of Lemma \ref{generalwavePhi}, with $A=\frac 3 2 r\rho+r\rhoF^2$, $B=- \ka\kab+6\rhoF^2$, $C=\frac 1 r\left( \ka\kab-2\rho\right)$ and $M=2\rhoF \Big(2 \kab \nabb_4(\qf^\F)- \ka\kab \qf^\F-  r\kab\nabb_4\psi_3+\left(-\frac {1}{ 2}  \ka\kab +6 \rho +4\rhoF^2\right)r\psi_3\Big)$, and obtain 
\beaa
\Box_\g(\qf)&=& \left( r^2\frac 1 2 (\frac 3 2 \rho+\rhoF^2)+\kab^{-1}r^2(-\frac 9 4 \kab \rho-\frac 7 2 \kab\rhoF^2)+ r (\frac 3 2 r\rho+r\rhoF^2)\right)\psi_0\\
&&+\left(\kab^{-1}r(\ka\kab^2-2\kab\rho-12\kab\rhoF^2)+ r (- \ka\kab+6\rhoF^2)+ \frac 3 2 r\rho+r\rhoF^2+\frac 1 2 r\rho+r\rhoF^2  \right)\ \psi_1 \\
&&+\left( - \ka\kab+6\rhoF^2 -\frac 1 2 (\ka\kab-2\rho)+\kab^{-1}(-\ka\kab^2+5\kab\rho+2\kab\rhoF^2)+ r \frac 1 r\left( \ka\kab-2\rho\right)+\frac 1 2 \ka\kab-4\rho-2\rhoF^2\right)\ \qf\\
&&+\kab^{-1}r\nabb_3 M  +\frac 3 2r M\\
&=& -2r^2\rhoF^2\psi_0-4r\rhoF^2  \psi_1 +\left( - \ka\kab+6\rhoF^2\right)\ \qf+\kab^{-1}r\nabb_3 M  +\frac 3 2r M
\eeaa
We compute $\kab^{-1}r\nabb_3 M  +\frac 3 2r M$:
\beaa
\kab^{-1}r\nabb_3 M  +\frac 3 2r M&=&2 \kab^{-1}r \rhoF \Big(2 \nabb_3\kab \nabb_4(\qf^\F)+2 \kab \nabb_3\nabb_4(\qf^\F)- \nabb_3(\ka\kab) \qf^\F- \ka\kab \nabb_3\qf^\F-  \nabb_3(r\kab)\nabb_4\psi_3-  r\kab\nabb_3\nabb_4\psi_3\Big) \\
&&+2 \kab^{-1}r \rhoF \Big(\nabb_3\left(-\frac {1}{ 2}  \ka\kab +6 \rho +4\rhoF^2\right)r\psi_3+\left(-\frac {1}{ 2}  \ka\kab +6 \rho +4\rhoF^2\right)\nabb_3(r\psi_3)\Big)  \\
&&+r \rhoF \Big(2 \kab \nabb_4(\qf^\F)- \ka\kab \qf^\F-  r\kab\nabb_4\psi_3+\left(-\frac {1}{ 2}  \ka\kab +6 \rho +4\rhoF^2\right)r\psi_3\Big)
\eeaa
which gives
\beaa
\kab^{-1}r\nabb_3 M  +\frac 3 2r M&=&2 \kab^{-1}r \rhoF \Big(2 (-\frac 12 \kab^2-2\omb\kab)\nabb_4(\qf^\F)+2 \kab \nabb_3\nabb_4(\qf^\F)- (-\ka\kab^2+2\kab\rho) \qf^\F- \ka\kab \nabb_3\qf^\F\\
&&-  (-2\omb r\kab)\nabb_4\psi_3-  r\kab\nabb_3\nabb_4\psi_3\Big) \\
&&+2 \kab^{-1}r \rhoF \Big((\frac 1 2 \ka\kab^2-10\kab\rho-14\kab\rhoF^2)r\psi_3+\left(-\frac {1}{ 2}  \ka\kab +6 \rho +4\rhoF^2\right)(\frac 1 2 \kab r \psi_3+r\nabb_3 \psi_3)\Big)  \\
&&+r \rhoF \Big(2 \kab \nabb_4(\qf^\F)- \ka\kab \qf^\F-  r\kab\nabb_4\psi_3+\left(-\frac {1}{ 2}  \ka\kab +6 \rho +4\rhoF^2\right)r\psi_3\Big)
\eeaa
Substituting the expression for $\nabb_3\nabb_4\psi_3$ and writing $\nabb_3(\psi_3)=\frac 1 r \kab \qf^\F-\frac 1 2  \kab \psi_3$, we obtain 
\beaa
\kab^{-1}r\nabb_3 M  +\frac 3 2r M&=& r \rhoF \Big( -8\omb\nabb_4(\qf^\F)+4 \nabb_3\nabb_4(\qf^\F)+ (\ka\kab-4\rho) \qf^\F- 2\ka \nabb_3\qf^\F+ (-\kab+4\omb) r \nabb_4\psi_3-  2r\nabb_3\nabb_4\psi_3\Big) \\
&&+ \kab^{-1}r \rhoF \left(-  \ka\kab +12 \rho +8\rhoF^2\right)r\nabb_3 \psi_3 +r \rhoF \left(-8 \rho -20\rhoF^2\right)r\psi_3\\
&=& r \rhoF \Big(4 \nabb_3\nabb_4(\qf^\F)+ (-2\kab-8\omb)\nabb_4(\qf^\F)- 2\ka \nabb_3\qf^\F+ (2\ka\kab+4\rho+8\rhoF^2) \qf^\F\Big) \\
&&+r \rhoF \left(-12 \rho -24\rhoF^2\right)r\psi_3
\eeaa
From \eqref{formula-1-wave} and Proposition \ref{wave-qfF}, we write 
\beaa
\nabb_{3}\nabb_{4}\qf^\F &=& -\Box_\g\qf^\F +\lapp_2 \qf^\F+\left(-\frac 12 \kab+2\omb\right) \nabb_4\qf^\F-\frac 1 2 \ka\nabb_3\qf^\F\\
&=& \lapp_2 \qf^\F+\left(-\frac 12 \kab+2\omb\right) \nabb_4\qf^\F-\frac 1 2 \ka\nabb_3\qf^\F+\left( \ka\kab+3 \rho \right)\ \qf^\F +\rhoF \left(\frac 1 r  \qf-4r\rhoF \ \psi_3 \right) 
 \eeaa
 and substituting in the above we obtain
 \beaa
\kab^{-1}r\nabb_3 M  +\frac 3 2r M&=& r \rhoF \Big(4 \lapp_2 \qf^\F+\left(-2 \kab+8\omb\right) \nabb_4\qf^\F-2 \ka\nabb_3\qf^\F+\left( 4\ka\kab+12 \rho \right)\ \qf^\F +4\rhoF \left(\frac 1 r  \qf-4r\rhoF \ \psi_3 \right) \\
&&+ (-2\kab-8\omb)\nabb_4(\qf^\F)- 2\ka \nabb_3\qf^\F+ (2\ka\kab+4\rho+8\rhoF^2) \qf^\F\Big) \\
&&+r \rhoF \left(-12 \rho -24\rhoF^2\right)r\psi_3\\
&=& r \rhoF \Big(4 \lapp_2 \qf^\F-4 \kab \nabb_4\qf^\F-4 \ka\nabb_3\qf^\F+\left( 6\ka\kab+16 \rho +8\rhoF^2\right)\ \qf^\F +(-12\rho-40\rhoF^2) \ r\psi_3\Big) \\
&&+4\rhoF^2 \qf 
\eeaa
This finally gives the desired equation. 
 
\end{proof}



\end{document}